\def\simPath{simulation}
\def\cvPath{cvExperiments}

\documentclass[11pt,a4paper,oneside]{article}

\usepackage[english]{babel}
\usepackage[T1]{fontenc}
\usepackage[utf8]{inputenc}
\usepackage{titletoc}
\usepackage{microtype}

\usepackage{graphicx}
\usepackage{grffile}

\usepackage[position = bottom]{subcaption} 
\PassOptionsToPackage{hyphens}{url}\usepackage[hidelinks, hypertexnames=false, hyperfootnotes=false]{hyperref} 
\usepackage[round]{natbib}
\usepackage{setspace}
\usepackage{geometry}
\usepackage{tabularx}
\usepackage{float}
\usepackage{floatflt}
\usepackage{threeparttable}
\usepackage{fancyhdr}
\usepackage{booktabs}
\usepackage{lscape}
\usepackage{appendix}
\usepackage{ragged2e}
\usepackage{titlesec}
\usepackage{titletoc}
\usepackage{enumitem}
\usepackage{threeparttablex}
\usepackage{caption, subcaption}
\usepackage{adjustbox} 
\usepackage{soul}
\usepackage{color}

\usepackage{amsthm}
\usepackage{amsmath}
\usepackage{amstext} 
\usepackage{amssymb} 
\usepackage{epsfig} 
\usepackage{longtable} 
\usepackage{bm}
\usepackage{bbm}
\numberwithin{equation}{section}
\usepackage{cleveref} 
\usepackage{aligned-overset} 

\usepackage{algorithm}
\usepackage[noend]{algpseudocode}

\usepackage{walsCommands}

\usepackage{ifthen}

\fancypagestyle{main}{%
	\onehalfspacing
	\lhead{}
	\rhead{\leftmark} 
}

\fancypagestyle{appendix}{%
	\onehalfspacing 
	\fancyhead{}%
	\lhead{}
	\rhead{APPENDIX}
}

\fancypagestyle{supplement}{%
	\onehalfspacing 
	\fancyhead{}%
	\lhead{}
	\rhead{SUPPLEMENTARY MATERIALS}
}

\fancypagestyle{literature}{%
	\onehalfspacing
	\fancyhead{}%
	\lhead{}
	\rhead{REFERENCES}
}

\providecommand{\keywords}[2][Keywords]
{
  \noindent\textbf{\textit{#1:}} #2
}

\addto\extrasenglish{%

}

\usepackage{etoolbox}
\makeatletter
\patchcmd{\ttl@select}{\strut}{}{}{}
\patchcmd{\ttlh@hang}{\strut}{}{}{}
\patchcmd{\ttlh@hang}{\strut}{}{}{}
\makeatother

\newcommand{\refSup}{supplementary materials}

\theoremstyle{plain}
\newtheorem{theorem}{Theorem}[section] 
\crefname{theorem}{Theorem}{Theorems}

\newtheorem{proposition}[theorem]{Proposition} 
\crefname{proposition}{Proposition}{Propositions}

\theoremstyle{plain}
\crefname{lemma}{Lemma}{Lemmas}

\theoremstyle{plain}
\newtheorem{corollary}[theorem]{Corollary} 
\crefname{corollary}{Corollary}{Corollaries}

\theoremstyle{definition}

\crefname{ass}{Assumption}{Assumptions}

\title{Weighted-Average Least Squares for Negative Binomial Regression}
\author{Kevin Huynh\textsuperscript{a,}\thanks{Corresponding author: Kevin Huynh, Faculty of Business and Economics, University of Basel, Peter Merian-Weg 6, 4052 Basel, Switzerland, E-Mail: \href{mailto:kevin.huynh@unibas.ch}{kevin.huynh@unibas.ch}}
}
\date{%
	\vspace{-0.5ex}
	{\footnotesize\textsuperscript{a}\textit{Faculty of Business and Economics, University of 	Basel, Peter Merian-Weg 6, 4052 Basel, Switzerland}}\\[2em]
}

\begin{document}

\maketitle
\thispagestyle{empty} 

\begin{abstract}
	\noindent Model averaging methods have become an increasingly popular tool for improving predictions and dealing with model uncertainty, especially in Bayesian settings. Recently, frequentist model averaging methods such as information theoretic and least squares model averaging have emerged. This work focuses on the issue of covariate uncertainty where managing the computational resources is key: The model space grows exponentially with the number of covariates such that averaged models must often be approximated. Weighted-average least squares (WALS), first introduced for (generalized) linear models in the econometric literature, combines Bayesian and frequentist aspects and additionally employs a semiorthogonal transformation of the regressors to reduce the computational burden. This paper extends WALS for generalized linear models to the negative binomial (NB) regression model for overdispersed count data. A simulation experiment and an empirical application using data on doctor visits were conducted to compare the predictive power of WALS for NB regression to traditional estimators. The results show that WALS for NB improves on the maximum likelihood estimator in sparse situations and  is competitive with lasso while being computationally more efficient.
\bigskip

\begin{small}

\end{small}

\end{abstract}


\newpage

\startlist[main]{lof}
\startlist[main]{lot}

\pagestyle{main}

\section{Introduction}\label{sec:intro}

In many empirical applications, model uncertainty emerges for a variety of reasons. For example, competing theories exist that can describe the data, or different assumptions are imposed on the data-generating process (DGP).
The two most common approaches for dealing with model uncertainty are model selection and model averaging. In model selection, the user selects the best performing model according to an estimation criterion and then carries out inference based on the chosen model. This approach is problematic because the uncertainty in the initial model selection step is often ignored, which could lead to overly confident decisions and predictions \citep{steel2020ma}. In contrast, model averaging accounts for model uncertainty by averaging over a set of candidate models, typically aiming at improving predictive accuracy \citep{ando2014ma}.

As datasets become larger, researchers commonly find themselves in high-dimensional settings with many potential covariates to model their response variable. Choosing appropriate regressors is particularly difficult in these situations because the number of candidate models grows exponentially with the number of regressors, i.e.\ for $k$ regressors, $2^k$ different subsets exist that may be considered as candidates. For the same reason, managing the model space and computational resources is key to applying model averaging procedures in the presence of covariate uncertainty. Bayesian model averaging (BMA) provides two general approaches: 1.\ Markov chain Monte Carlo methods (MCMC) and 2.\ non-MCMC approximation methods, see e.g.\ \citet[p.~384~ff.]{hoeting1999bma} for an early overview. A common solution adopted in frequentist model averaging (FMA), e.g.\ in \citet{zhang2016optmaglm}, is to prescreen for a viable set of models. In contrast, weighted-average least squares (WALS), first proposed by \citet{magnus2010growth} for the linear regression model and then extended by \citet{deluca2018glm} to generalized linear models (GLMs), omits a preselection of models by combining Bayesian and frequentist aspects and, especially, leveraging a semiorthogonal transformation of the regressors allowing for fast computation times. Earlier work by \citet{heumann2010logit} generalizes WALS to logistic regression using a similar transformation as in \citet{deluca2018glm}.

Most of the literature, particularly in economics, has focused on model averaging for linear regression models. However, many interesting applications require nonlinear models, e.g.\ classification, count data modeling and survival analysis. The negative binomial (NB) distribution, especially of type 2 (NB2), is a popular distribution featuring overdispersion for count data regression, see e.g.\ \citet{cameron1986doctor} and \citet{cameron1988health} for applications in health economics. Notably, the NB2 regression model is not a GLM when its dispersion parameter is estimated from the data. \citet{deb2002rand} extend it to hurdle and finite mixture models and \citet{greene2008nbp} develops a more general form, called NBP, which encompasses the NB of type 1 and 2.

Despite its wide application, very limited literature exists on model averaging methods for the NB regression model that jointly estimate the regression coefficients and the dispersion parameter.
One of the few open-source packages for model averaging is \textsf{BMA} by \citet{raftery2020bma}, which currently supports BMA for GLMs and survival models. Hence, it is only able to fit an NB2 with pre-specified dispersion parameter, which is a GLM.

In this paper, I extend WALS GLM by \citet{deluca2018glm} to the NB2 regression model (WALS NB) to account for covariate uncertainty in the specification of the linear predictor.
WALS is particularly well suited as it elegantly circumvents a preselection of models by transforming the regressors, allowing me to focus on the averaging procedure. Analogous to \citet{deluca2018glm}, I first derive the one-step maximum likelihood estimator based on a Taylor expansion of the NB2 log-likelihood function and then employ a transformation akin to the semiorthogonal transformation used in WALS GLM.

At the time of writing, the asymptotic distribution of the WALS estimator for GLMs is still an open research topic and its variance estimator has been a subject of debate. Recent work by \citet{deluca2022sampling} proposes a new estimator for the variance of WALS in the linear regression model instead of the Bayesian posterior variance that has traditionally been used. \citet{deluca2023interval} further analyze the confidence and prediction intervals of WALS in the linear model and propose a new simulation-based method that corrects for bias in the WALS estimator. In contrast, this work focuses on the predictive power of model averaging and leaves the challenging issue of inference (after model averaging) for future research. Model averaging estimators typically improve the predictive accuracy compared to using a single model. For example, in an early application of BMA, \citet{madigan1994occam} find that BMA achieves better logarithmic predictive score than any single model. Moreover, \citet{min1993bma} show that the expected squared error loss of predictive mean forecasts is always minimized by BMA, if the data-generating model is included in the model space considered for averaging. In this paper, I compare the proposed WALS NB method to traditional maximum likelihood (ML) estimation of the NB2 regression model in a simulation experiment using the classical precision measure, root mean squared error (RMSE), and scoring rules \citep{gneiting2007scores} as measures for the distributional fit. Finally, the method is also compared to the lasso estimator \citep{wang2016pencount} in an empirical application on modeling doctor visits. Both the simulation experiment and the empirical application show that WALS NB improves on the ML estimator in sparse situations with few observations and many covariates. In the latter, its fit is competitive with lasso while being computationally more efficient.

\section{Setup}\label{sec:setup}

The setup and derivation of WALS NB mostly follow the steps in \citet{deluca2018glm} for WALS GLM. Assume that data $y_i, i = 1, 2, \dotsc, n$, are conditionally independent given $k$-dimensional regressors $x_{i}$ and follow an NB2 distribution with mean $\mu_i$ and dispersion parameter $\rho$, i.e.\ $y_{i} | x_{i} \sim \mathrm{NB2}(\mu_i, \rho)$. As in the standard GLM setup, I model the mean using an inverse link function $h$ on $\mu_i := \mu(\eta(\beta, x_i)) = h(\eta(\beta, x_i))$ with linear predictor $\eta_i := \eta(\beta, x_i) = x_{i}^{\top} \beta$ and regression coefficients $\beta$. The NB2 distribution has the probability mass function
\begin{equation}\label{eq:nb2density}
	f(y_i | \mu_i, \rho) = \frac{\Gamma(y_i + \rho)}{\Gamma(\rho) \Gamma(y_i + 1)} \frac{\mu_{i}^{y_i} \rho^{\rho}}{(\mu_{i} + \rho)^{y_i + \rho}}, \quad y_{i} \in \mathbb{N}_{0}, \ \rho > 0,
\end{equation}
where $\Gamma$ is the gamma function, and its conditional variance is given by
\begin{equation}\label{eq:nb2meanvar}
	\sigma_{i}^2 := \var(y_i | \mu_i, \rho) = \mu_i + \frac{\mu_{i}^2}{\rho}.
\end{equation}
A distribution from the exponential family has the following density
\begin{equation*}
	f(y_i | \theta_i) = \exp( y_i \theta_i - b(\theta_i) + l(y_i)),
\end{equation*}
where $b$ and $l$ are known functions. Typical formulations as in e.g.\ \citet[p.~301]{fahrmeir2013regression} include a dispersion parameter which, without loss of generality, I set equal to one. Moreover, the following two identities hold for the mean and variance
\begin{equation*}
	\mu_i = \frac{\partial b(\theta_i)}{\partial \theta_i}, \qquad \sigma_{i}^{2} = \frac{\partial^2 b(\theta_i)}{\partial \theta_{i}^2}.
\end{equation*}

For WALS estimation, I rewrite the NB2 probability mass function into a similar form as the exponential family with a $\log$-link on $\rho$ by using
\begin{align*}
	\theta_i &:= \theta\left(\mu_i, \rho(\alpha)\right) = \theta\left( h \left(\eta(\beta, x_i)\right), \rho(\alpha) \right) = \log \left( \frac{\mu_i}{\mu_i + \rho} \right) = \log(\mu_i) - \log(\mu_i + \rho), \\
	\rho(\alpha) &= \exp(\alpha).
\end{align*}
Thus, the probability mass function becomes
\begin{equation*}
	f(y_i | \theta_i, \rho) = \exp\left(y_i \theta_i + \rho \log(1 - \exp(\theta_i)) + \log\Gamma(y_i + \rho) - \log\Gamma(\rho) - \log\Gamma(y_i + 1) \right),
\end{equation*}
where I dropped the dependence of $\theta$ on $\mu$ and $\rho$, and of $\rho$ on $\alpha$ for notational brevity. From the last line, we can identify the following building blocks of the exponential family:
\begin{align*}
	b(\theta_i, \rho) &= -\rho \log(1 - \exp(\theta_i)), \\
	l(y_i, \rho) &= \log\Gamma(y_i + \rho) - \log\Gamma(\rho) - \log\Gamma(y_i + 1).
\end{align*}
Thus, for fixed $\rho$, the NB2 is a member of the exponential family and leads to a GLM. However, $\rho$ is estimated from the data in the WALS procedure and, hence, the underlying model is not a GLM anymore.
Furthermore, I separate $l(y_i, \rho)$ into two terms using 
\begin{equation*}
	a(y_i, \rho) := \log\Gamma(y_i + \rho) - \log\Gamma(\rho), \qquad d(y_i) := - \log\Gamma(y_i + 1),
\end{equation*}
so the NB2 probability mass function can be rewritten as
\begin{align*}
	f(y_i | \theta_i, \rho) = \exp(y_i \theta_i - b(\theta_i, \rho) + l(y_i, \rho)) = \exp(y_i \theta_i - b(\theta_i, \rho) + a(y_i, \rho) + d(y_i)),
\end{align*}
which will simplify the derivation of the WALS estimator later.

I allow for uncertainty in the specification of the linear predictor while assuming that the (conditional) probability mass function of $y_i$ and the inverse link $h$ are correctly specified. First, collect over all observations $n$ the response $y_i$ to an $n$-vector $y$ and the regressors $x_i$ to an $n \times k$ matrix $X$ that contains $x_{i}^{\top}$ as $i$th row. Then, partition the regressors into focus and auxiliary regressors $X = (X_{1}, X_{2})$, where $X_{p}$ is an $n \times k_{p}$ matrix with $i$th row equal to $x_{ip}^{\top}, p = 1, 2$, and $k_1 + k_2 = k$. Further, let $\beta = (\beta_{1}^{\top}, \beta_{2}^{\top})^{\top}$ so the linear predictor can be expressed as $\eta_i = x_{i1}^{\top} \beta_1 + x_{i2}^{\top} \beta_2$. Stacking the linear predictors over all $n$ observations then gives the vector $\eta(\beta) = X_1 \beta_1 + X_2 \beta_2$.

Consider averaging over models containing all focus regressors $X_1$ but arbitrary subsets of the $k_2$ auxiliary regressors in $X_2$, which leads to a total of $2^{k_2}$ possible models. The $j$th model is represented by the restriction $R_{j}^{\top} \beta_{2} = 0$, where $R_{j}$ denotes a $k_2 \times r_j$ matrix of rank $0 \leq r_j \leq k_2$, such that $R_{j}^{\top} = (I_{r_j}, 0)$ or column-permutations thereof. Thus, the matrix $R_j$ specifies which auxiliary regressors are excluded from the $j$th model and its rank $r_j$ denotes the number of excluded auxiliary regressors. 
Note that $0$ represents a scalar, vector or matrix filled with zeroes of matching dimension unless otherwise stated. For example, $0$ in $R_{j}^{\top} = (I_{r_j}, 0)$ is an $r_{j} \times (k_2 - r_{j})$ matrix.

\section{ML estimation}\label{sec:estimation}

I start with the classical maximum likelihood estimator of the NB2 regression model. Under conditional independence, the (conditional) $\log$-likelihood is
\begin{equation}\label{eq:nb2loglik}
	\begin{aligned}
		\ell(\beta, \alpha) &= \sumin \log f\bigl(y_i | \theta(\mu_i(\beta), \rho(\alpha)), \rho(\alpha)\bigr) = \sumin \left[ y_i \theta_i - b(\theta_i, \rho) + a(y_i, \rho) + d(y_i) \right] \\
		&= \text{constant} + \sumin \left[y_i \theta_i - b_i + a_i \right],
	\end{aligned}
\end{equation}
where $b_i := b(\theta_i, \rho)$ and $a_i := a(y_i, \rho)$. In the following, I will generally omit the dependence of $\theta$, $\mu$, $\eta$, $\rho$, $b$ and $a$ on their parameters to reduce clutter. Moreover, only the $\log$-link is considered for the mean (and dispersion) parameter, i.e.\ $h(\eta_{i}) = \exp(\eta_{i})$, but the general notation using $h$ is retained in many places below to facilitate comparisons with WALS GLM by \citet{deluca2018glm} and to allow easier extension of the method to other link functions in the future.

The score functions follow:
\begin{align*}
	s_p(\beta, \alpha) := \frac{\partial \ell(\beta, \alpha)}{\partial \beta_{p}} &= \sumin \left[ y_i \frac{\partial \theta_i}{\partial \eta_i} - \frac{\partial b_i}{\partial \theta_i} \frac{\partial \theta_i}{\partial \eta_i} \right] x_{ip} = \sumin v_{i} [y_i - \mu_{i}] x_{ip}, \quad p = 1, 2, \\
	s_{\alpha}(\beta, \alpha) := \frac{\partial \ell(\beta, \alpha)}{\partial \alpha} &= \sumin \left[ y_i \frac{\partial \theta_i}{\partial \rho} - \frac{\partial b_i}{\partial \theta_i} \frac{\partial \theta_i}{\partial \rho} - \frac{\partial b_i}{\partial \rho} + \frac{\partial a_i}{\partial \rho} \right] \frac{\partial \rho}{\partial \alpha} = \sumin \kappa_i \frac{\partial \rho}{\partial \alpha},
\end{align*}
with
\begin{align*}
	v_i &:= v(\eta_i, \rho) := \frac{\partial \theta_i}{\partial \eta_i}, \\
	\kappa_i &:= \kappa(\eta_{i}, \rho, y_i) := y_i \frac{\partial \theta_i}{\partial \rho} - \frac{\partial b_i}{\partial \theta_i} \frac{\partial \theta_i}{\partial \rho} - \frac{\partial b_i}{\partial \rho} + \frac{\partial a_i}{\partial \rho} \nonumber \\
	&= -\frac{y_i - \mu_i}{\mu_i + \rho} + \log(\rho) - \log(\mu_i + \rho) + \dig(y_i + \rho) - \dig(\rho), 
\end{align*}
where $\dig(x) := \partial \log \Gamma(x) / \partial x$ is the digamma function. Furthermore, let $H(\beta, \alpha)$ be the negative Hessian of the log-likelihood, which is composed of several submatrices that are listed below. The first components are
\begin{equation*}
	\begin{aligned}
		H_{pq}(\beta, \alpha) &:= - \frac{\partial^2 \ell(\beta, \alpha)}{\partial \beta_p \partial \beta_{q}^{\top}} \\
		&= - \sumin \left[ y_{i} \frac{\partial^2 \theta_i}{\partial \eta_{i}^2} - \left(\frac{\partial^2 b_i}{\partial \theta_{i}^2} \left(\frac{\partial \theta_i}{\partial \eta_i} \right)^2 + \frac{\partial b_i}{\partial \theta_i} \frac{\partial^2 \theta_i}{\partial \eta_{i}^2} \right) \right] x_{ip}x_{iq}^{\top} \\
		&= \sumin \left[v_{i}^2 \sigma_{i}^2 - \omega_{i}(y_i - \mu_i) \right] x_{ip} x_{iq}^{\top} =  \sumin \psi_i x_{ip} x_{iq}^{\top}, \quad p, q = 1, 2,
	\end{aligned}
\end{equation*}
where
\begin{equation*}
	\omega_{i} := \omega(\eta_i, \rho) := \frac{\partial^2 \theta_i}{\partial \eta_{i}^2}, \qquad
	\psi_i := \psi(\eta_i, \rho, y_i) := v_i^2 \sigma_i^{2} - \omega_i(y_i - \mu_{i}).
\end{equation*}
The next submatrices are defined as
\begin{align*}\label{eq:Hpalpha}
	H_{p\alpha}(\beta, \alpha) &:= -\frac{\partial^2 \ell(\beta, \alpha)}{\partial \beta_{p} \partial \alpha} \nonumber \\
	&= \sumin \left[- y_i \frac{\partial^2 \theta_i}{\partial \eta_i \partial \rho} + \left( \frac{\partial^2 b_i}{\partial \theta_{i}^2} \frac{\partial \theta_i}{\partial \rho} + \frac{\partial^2 b_i}{\partial \theta_i \partial \rho} \right) \frac{\partial \theta_i}{\partial \eta_i} + \frac{\partial b_i}{\partial \theta_i} \frac{\partial^2 \theta_i}{\partial \eta_i \partial \rho} \right] x_{ip} \frac{\partial \rho}{\partial \alpha} \nonumber \\
	&= H_{\alpha p}(\beta, \alpha)^{\top}, \quad p = 1, 2.
\end{align*}
They further simplify thanks to
\begin{equation}\label{eq:a1}
\frac{\partial^2 b_i}{\partial \theta_{i}^2} \frac{\partial \theta_i}{\partial \rho} + \frac{\partial^2 b_i}{\partial \theta_i \partial \rho} = -\left(   \frac{\mu_{i}^2}{\rho} + \mu_i \right) \frac{1}{\mu_i + \rho} + \frac{\mu_i}{\rho} = 0,
\end{equation}
so $H_{p\alpha}(\beta, \alpha)$ may be rewritten as
\begin{equation*}\label{eq:Hpalpha_simple}
	H_{p \alpha}(\beta, \alpha) = - \sumin c_i [y_i - \mu_i] x_{ip} \frac{\partial \rho}{\partial \alpha},  \quad p = 1, 2,
\end{equation*}
using $c_{i} := c(\eta_{i}, \rho) = \partial^2 \theta_i / \partial \eta_{i} \partial \rho$. Finally, the last part is
\begin{align*}\label{eq:Halpha2}
	H_{\alpha \alpha}(\beta, \alpha) &:= - \frac{\partial^2 \ell(\beta, \alpha)}{\partial \alpha^2} = -\sumin \left[ \frac{\partial \kappa_i}{\partial \rho} \left(\frac{\partial \rho}{\partial \alpha} \right)^2 + \kappa_i \frac{\partial^2 \rho}{\partial \alpha^2} \right] = -\sumin \left[ k_i g^2 + \kappa_i \varrho \right],
\end{align*}
with
\begin{equation*}\label{eq:ki}
	\begin{aligned}
		k_i &:= k(\eta_i, \rho, y_i) := \frac{\partial \kappa_i}{\partial \rho} \\
		&= y_i \frac{\partial^2 \theta_i}{\partial \rho^2} - \left(\frac{\partial^2 b_i}{\partial \theta_{i}^2}  \frac{\partial \theta_i}{\partial \rho}	 + \frac{\partial^2 b_i}{\partial \theta_i \partial \rho} \right) \frac{\partial \theta_i}{\partial \rho} - \frac{\partial b_i}{\partial \theta_i} \frac{\partial^2 \theta_i}{\partial \rho^2} - \frac{\partial^2 b_i}{\partial \theta_i \partial \rho} \frac{\partial \theta_i}{\partial \rho} - \frac{\partial^2 b_i}{\partial \rho^2} + \frac{\partial^2 a_i}{\partial \rho^2}   \\
		\overset{\eqref{eq:a1}}&{=} \frac{y_i - \mu_i}{(\mu_i + \rho)^2} + \frac{\mu_i}{\rho (\mu_i + \rho)} + \trig(y_i + \rho) - \trig(\rho),
	\end{aligned}
\end{equation*}
where $\trig(x) := \partial^2 \log \Gamma(x) / \partial x^2$ is the trigamma function, and
\begin{equation*}
	g := g(\alpha) := \frac{\partial \rho}{\partial \alpha}, \qquad \varrho := \varrho(\alpha) := \frac{\partial^2 \rho}{\partial \alpha^2}.
\end{equation*}

The ML estimator for the $j$th model solves the following constrained optimization problem
\begin{equation}\label{eq:optNB}
\begin{aligned}
	&\max_{\beta, \alpha} & \quad & \ell(\beta, \alpha) \\
	& \text{subject to} &	 & R_{j}^{\top}\beta_2 = 0.
\end{aligned}
\end{equation}
As a first step towards the solution, I construct the Lagrangian
\begin{equation*}
	L(\beta,\alpha, \nu_{j}) = \ell(\beta, \alpha) - \nu_{j}^{\top}(R_{j}^{\top}\beta_{2}),
\end{equation*}
where $\nu_{j}$ denotes the $r_{j}$-vector of Lagrange multipliers. Setting the first derivatives equal to zero yields the system of nonlinear equations
\begin{align}\label{eq:systemBeta}
	& s_{1}(\beta, \alpha) = 0,&  & s_{2}(\beta, \alpha) - R_{j}\nu_{j} = 0,& & s_{\alpha}(\beta, \alpha) = 0, & & R_{j}^{\top} \beta_2 = 0.
\end{align}
Following \citet[p.~3~f.]{deluca2018glm}, I consider a one-step ML estimator that approximates the solution of the system. In contrast to iterative procedures such as Newton-Raphson, which are typically used for solving nonlinear equation systems, the one-step ML estimator admits closed-form expressions.

\newpage
\subsection{One-step ML estimator}\label{sec:onestepNB}

In the remainder of the paper, I assume that all necessary conditions for the algebraic manipulations, e.g.\ rank conditions on the regressor matrix $X$, are satisfied. Detailed proofs are found in \Cref{sec:proofs}.

I expand the estimating equations of \eqref{eq:systemBeta} (except for $R_{j}^{\top} \beta_2 = 0$) around starting values $\bar{\beta} = (\bar{\beta}_{1}^{\top}, \bar{\beta}_{2}^{\top})^{\top}$ and $\bar{\alpha}$. Further, $\bar{\rho} = \rho(\bar{\alpha})$, since the mapping from $\alpha$ to $\rho$ is strictly monotonic (log-link). Using a first-order Taylor expansion and ignoring the remainder term yields
\begin{equation}\label{eq:taylor}
	\begin{aligned}
		0 &\approx \bs_1 - \bH_{11} (\beta_1 - \bbeta_1) - \bH_{12} (\beta_2 - \bbeta_2) - \bH_{1\alpha} (\alpha - \balpha),\\
		0 &\approx \bs_2 - \bH_{21} (\beta_1 - \bbeta_1) - \bH_{22} (\beta_2 - \bbeta_2) - \bH_{2\alpha} (\alpha - \balpha)  - R_{j}\nu_{j},\\
		0 &\approx \bs_{\alpha} - \bH_{\alpha1} (\beta_1 - \bbeta_1) - \bH_{\alpha2} (\beta_2 - \bbeta_2) - \bH_{\alpha\alpha} (\alpha - \balpha),\\
		0 &= R_{j}^{\top} \beta_{2},
	\end{aligned}
\end{equation}
where $\bs_{p} := s_{p}(\bbeta, \balpha), \bH_{pq} := H_{pq}(\bbeta, \balpha)$ and $\bH_{p\alpha} = H_{p \alpha}(\bbeta,  \balpha)$, $p = 1, 2$. In the following, all quantities evaluated at $(\bbeta, \balpha)$ are denoted by a bar and the approximations in \eqref{eq:taylor} are treated as equalities for a simpler notation.

First, consider the unrestricted model with $R_u = 0$. Define the data transformations
\begin{equation}\label{eq:trafos}
	\begin{aligned}
		\bar{y} &:= \bX_1 \bbeta_1 + \bX_2 \bbeta_2 + \bar{u},
		&\bX_{p} &:= \bPsi^{1/2} X_{p}, \quad  p = 1,2, \\
		\bar{y}_{0} &:= \bar{y} - \bar{g} \bPsi^{-1/2} \bC (y - \bmu) \balpha,
		&\bar{u} &:= \bPsi^{-1/2} \bV (y - \bmu),
	\end{aligned}
\end{equation}
which involve the $n \times n$ matrices
\begin{equation*}\label{eq:diagmats}
	\begin{aligned}
		\bV &:= V(\bar{\eta}, \bar{\rho}) := \diag\left(v(\bar{\eta}_1, \bar{\rho}), v(\bar{\eta}_2, \bar{\rho}), \dotsc, v(\bar{\eta}_n, \bar{\rho})\right), \\
		\bPsi &:= \Psi(\bar{\eta}, \bar{\rho}, y) := \diag\left(\psi(\bar{\eta}_1, \bar{\rho}, y_1), \psi(\bar{\eta}_2, \bar{\rho}, y_2), \dotsc, \psi(\bar{\eta}_n, \bar{\rho}, y_n)\right), \\
		\bC &:= C(\bar{\eta}, \bar{\rho}) := \diag\left(c(\bar{\eta}_1, \bar{\rho}), c(\bar{\eta}_2, \bar{\rho}), \dotsc, c(\bar{\eta}_n, \bar{\rho})\right),
	\end{aligned}
\end{equation*}
and the $n$-vectors $\bmu := \mu(\bar{\eta}) := \left(h(\bar{\eta}_{1}), h(\bar{\eta}_{2}), \dotsc, h(\bar{\eta}_{n})\right)^{\top}$ and $\bar{\eta} := \left(\bar{\eta}_{1}, \bar{\eta}_{2}, \dotsc, \bar{\eta}_{n}\right)^{\top} = X_1 \bbeta_1 + X_2 \bbeta_2$ with $\bar{\eta}_{i} := \eta(\bar{\beta}, x_i)$. Using the log-link further guarantees $\rank(\bPsi) = n$ because
$
\psi(\bar{\eta}_{i}, \bar{\rho}, y_i) = \bar{\mu}_{i} \bar{\rho} (y_i + \bar{\rho}) / (\mu_{i} + \bar{\rho})^{2} > 0
$
since $\bar{\mu}_{i} > 0$, $\bar{\rho} > 0$ and $y_{i} \geq 0$ for all $i$. Moreover, define
\begin{equation*}
	\bar{t} :=  \bar{g} \xmat{\bkappa}{\ones} - \bar{g} \xmat{(y - \bmu)}{\bC} \bar{\eta} - (\bar{g}^2 \xmat{\bk}{\ones} + \bar{\varrho} \xmat{\bkappa}{\ones})\balpha,
\end{equation*}
with $n$-vectors
\begin{equation*}
	\bk := k(\bar{\eta}, \bar{\rho}, y) := (\bk_{1}, \bk_{2}, \dotsc, \bk_{n})^{\top}, \quad
	\bkappa := \kappa(\bar{\eta}, \bar{\rho}, y) := (\bkappa_{1}, \bkappa_{2}, \dotsc, \bkappa_{n})^{\top},
\end{equation*}
where $\bk_{i} := k(\bar{\eta}_{i}, \bar{\rho}, y_i)$, $\bkappa_{i} := \kappa(\bar{\eta}_{i}, \bar{\rho}, y_i)$ and $\ones := (1, \dotsc, 1)^{\top}$ is an $n$-vector filled with ones.
Notice the slight abuse in notation, where $\mu$, $k$ and $\kappa$ are vector-valued functions here, whereas they were scalar-valued in the sections before. Furthermore, let
\begin{equation*}
	\bepsilon := \frac{\bar{g}}{\bar{g}^2 \xmat{\bk}{\ones} + \bar{\varrho} \xmat{\bkappa}{\ones}}, \qquad \bq := \bC (y - \bmu).
\end{equation*}
Then, the solution to the linearized system of likelihood equations \eqref{eq:taylor} can be expressed in closed form as
\begin{align*}
		\tbeta_{1u} ={}& \bigg[
		\left( \frac{\xmat{\bX_1}{\bX_1}}{n} + \frac{\bar{g} \bepsilon}{n} \xmat{X_1}{\bq} \xmat{\bq}{X_1} \right)^{-1} \nonumber \\ 
		& + \bigg\{
			\left( \frac{\xmat{\bX_1}{\bX_1}}{n} + \frac{\bar{g} \bepsilon}{n} \xmat{X_1}{\bq} \xmat{\bq}{X_1} \right)^{-1}
			\left(\frac{\xmat{\bX_1}{\bX_2}}{n} + \frac{\bar{g} \bepsilon}{n} \xmat{X_1}{\bq} \xmat{\bq}{X_2} \right)
			\left(\frac{\bX_{2}^{\top} \bM_1 \bX_2}{n} \right)^{-1}  \nonumber \\
			&  \cdot \left( \frac{\xmat{\bX_2}{\bX_1}}{n} + \frac{\bar{g} \bepsilon}{n} \xmat{X_2}{\bq} \xmat{\bq}{X_1} \right)
			\left( \frac{\xmat{\bX_1}{\bX_1}}{n} + \frac{\bar{g} \bepsilon}{n} \xmat{X_1}{\bq} \xmat{\bq}{X_1} \right)^{-1}
			\bigg\}
		\bigg]
	\left(\frac{\xmat{\bX_1}{\bar{y}_0}}{n} - \frac{\bar{t} \bepsilon}{n} \xmat{X_1}{\bar{q}} \right) \nonumber \\
	& - \biggl[ \left( \frac{\xmat{\bX_1}{\bX_1}}{n} + \frac{\bar{g} \bepsilon}{n} \xmat{X_1}{\bq} \xmat{\bq}{X_1} \right)^{-1}
	\left(\frac{\xmat{\bX_1}{\bX_2}}{n} + \frac{\bar{g} \bepsilon}{n} \xmat{X_1}{\bq} \xmat{\bq}{X_2} \right)
	\left(\frac{\bX_{2}^{\top} \bM_1 \bX_2}{n} \right)^{-1} \nonumber \\
	& \cdot \left(\frac{\xmat{\bX_2}{\bar{y}_0}}{n} - \frac{\bar{t} \bepsilon}{n} \xmat{X_2}{\bar{q}} \right) \biggr], \\
		\tbeta_{2u} ={}& - \biggl[
		\left(\frac{\bX_{2}^{\top} \bM_1 \bX_2}{n} \right)^{-1}
		\left(\frac{\xmat{\bX_2}{\bX_1}}{n} + \frac{\bar{g} \bepsilon}{n} \xmat{X_2}{\bq} \xmat{\bq}{X_1} \right)
		\left(\frac{\xmat{\bX_1}{\bX_1}}{n} + \frac{\bar{g} \bepsilon}{n} \xmat{X_1}{\bq} \xmat{\bq}{X_1} \right)^{-1} \nonumber \\
		& \cdot \left(\frac{\xmat{\bX_1}{\bar{y}_0}}{n} - \frac{\bar{t} \bepsilon}{n} \xmat{X_1}{\bar{q}} \right)
		\biggr]
		+ \left(\frac{\bX_{2}^{\top} \bM_1 \bX_2}{n} \right)^{-1}
		\left(\frac{\xmat{\bX_2}{\bar{y}_0}}{n} - \frac{\bar{t} \bepsilon}{n} \xmat{X_2}{\bar{q}} \right), \\
	\talpha_{u} ={}& -\frac{\bar{t} + \bar{g} \xmat{(y - \bmu)}{\bC} (X_1 \tbeta_{1u} + X_2 \tbeta_{2u})}{\bar{g}^2 \xmat{\bk}{\ones} + \bar{\varrho} \xmat{\bkappa}{\ones}},
\end{align*}
where
\begin{align*}
	\bM_1 ={}& (I_n + \bar{g} \bepsilon \bPsi^{-1/2} \bq \bq^{\top} \bPsi^{-1/2})& \nonumber \\
	& - \left[ (\bX_1 + \bar{g} \bepsilon \bPsi^{-1/2} \bq \bq^{\top} X_1) (\xmat{\bX_{1}}{\bX_{1}} + \bar{g} \bepsilon \xmat{X_1}{\bq} \xmat{\bq}{X_1} )^{-1}  (\bX_{1}^{\top} + \bar{g} \bepsilon X_{1}^{\top} \bq \bq^{\top} \bPsi^{-1/2}) \right],
\end{align*}
is a symmetric matrix. In contrast to \citet{deluca2018glm}, $\bM_1$ is not idempotent anymore due to the rank-1 perturbation in $I_n + \bar{g} \bepsilon \bPsi^{-1/2} \bq \bq^{\top} \bPsi^{-1/2}$, which is a consequence of the additional dispersion parameter $\rho$ in the NB2 model compared to GLMs.

Likewise, consider the general one-step ML estimator for the $j$th model. Define the symmetric and idempotent $k_2 \times k_2$ matrix
\begin{equation*}
	\bP_j := \left(\frac{\bX_{2}^{\top} \bM_1 \bX_2}{n} \right)^{-1/2}  R_{j} \left( R_{j}^{\top} \left(\frac{\bX_{2}^{\top} \bM_1 \bX_2}{n} \right)^{-1} R_{j} \right)^{-1}  R_{j}^{\top} \left(\frac{\bX_{2}^{\top} \bM_1 \bX_2}{n} \right)^{-1/2},
\end{equation*}
the $k_1 \times k_2$ matrix
\begin{equation*}
	\bQ := \left(\frac{\xmat{\bX_1}{\bX_1}}{n} + \frac{\bar{g} \bepsilon}{n} \xmat{X_1}{\bq} \xmat{\bq}{X_1} \right)^{-1}  \left(\frac{\xmat{\bX_1}{\bX_2}}{n} + \frac{\bar{g} \bepsilon}{n} \xmat{X_1}{\bq} \xmat{\bq}{X_2} \right) \left(\frac{\bX_{2}^{\top} \bM_1 \bX_2}{n} \right)^{-1/2},
\end{equation*}
and the following transformation of the unrestricted one-step ML estimator $\tbeta_{2u}$
\begin{equation}\label{eq:vartheta}
	\tvartheta :=  \left(\frac{\bX_{2}^{\top} \bM_1 \bX_2}{n} \right)^{1/2} \tbeta_{2u}.
\end{equation}
Then, analogous to Proposition 1 of \citet{deluca2018glm}, I obtain the one-step ML estimator for the $j$th model in the following proposition.
\begin{proposition}[One-step ML estimators]\label{prop:onestepML} The one-step ML estimators of $\beta_1$, $\beta_2$ and $\alpha$ based on the $j$th model are
	\begin{align*}
		\tbeta_{1j} &= \tbeta_{1r} - \bQ \bW_{j} \tvartheta, \\
		\tbeta_{2j} &= \left(\frac{\bX_{2}^{\top} \bM_1 \bX_2}{n} \right)^{-1/2} \bW_j \tvartheta, \\
		\talpha_{j} &= -\frac{\bar{t} + \bar{g} \xmat{(y - \bmu)}{\bC} (X_1 \tbeta_{1j} + X_2 \tbeta_{2j})}{\bar{g}^2 \xmat{\bk}{\ones} + \bar{\varrho} \xmat{\bkappa}{\ones}},
	\end{align*}
	where
	\begin{align*}
		\tbeta_{1r} &= \left( \frac{\xmat{\bX_1}{\bX_1}}{n} + \frac{\bar{g} \bepsilon}{n} \xmat{X_1}{\bq}\xmat{\bq}{X_1} \right)^{-1} \left( \frac{\xmat{\bX_1}{\bar{y}_0}}{n} - \frac{\bar{t} \bepsilon}{n} \xmat{X_1}{\bq} \right),
	\end{align*}
	is the fully restricted one-step ML estimator of $\beta_1$ and $\bW_j = I_{k_{2}} - \bP_j$.
\end{proposition}

\subsection{Transformed model}\label{sec:trafoNB}

The WALS NB estimator relies on a preliminary transformation of the auxiliary regressors to reduce the computational burden, akin to WALS for the linear regression model \citep{magnus2010growth} and GLMs \citep{deluca2018glm}.

First, scale the focus regressors by defining
\begin{align}\label{eq:trafoZ1}
	\bZ_{1} := \bX_{1} \bDelta_{1}, && Z_{1} := X_1 \bDelta_1, && \gamma_1 := \bDelta_{1}^{-1} \beta_1,
\end{align}
with the $k_1 \times k_1$ diagonal matrix $\bDelta_1 := \diag\left(\xmat{\bX_1}{\bX_1} / n \right)^{-1/2}$ 
such that $\diag\left( \xmat{\bZ_1}{\bZ_1} / n \right) = (1, \dotsc, 1)$. The only purpose of the transformation is to improve the numerical accuracy by normalizing all regressors to be the same scale in $\bZ_1$ \citep[p.~5]{deluca2018glm}.
It further implies
\begin{align*}
	\bZ_1 \gamma_1 ={}& \bX_1 \beta_1,\\
	\beta_1 ={}& \bDelta_1 \gamma_1, \\
	\bM_{1} ={}& (I_n + \bar{g} \bepsilon \bPsi^{-1/2} \bq \bq^{\top} \bPsi^{-1/2})& \nonumber \\
	& - \left[ (\bZ_1 + \bar{g} \bepsilon \bPsi^{-1/2} \bq \bq^{\top} Z_1) (\xmat{\bZ_{1}}{\bZ_{1}} + \bar{g} \bepsilon \xmat{Z_1}{\bq} \xmat{\bq}{Z_1} )^{-1}  (\bZ_{1}^{\top} + \bar{g} \bepsilon Z_{1}^{\top} \bq \bq^{\top} \bPsi^{-1/2}) \right], \\
	\overset{\eqref{eq:trafoZ1}}{=} & (I_n + \bar{g} \bepsilon \bPsi^{-1/2} \bq \bq^{\top} \bPsi^{-1/2})& \nonumber \\
	& - \left[ (\bX_1 + \bar{g} \bepsilon \bPsi^{-1/2} \bq \bq^{\top} X_1) (\xmat{\bX_{1}}{\bX_{1}} + \bar{g} \bepsilon \xmat{X_1}{\bq} \xmat{\bq}{X_1} )^{-1}  (\bX_{1}^{\top} + \bar{g} \bepsilon X_{1}^{\top} \bq \bq^{\top} \bPsi^{-1/2} ) \right],
\end{align*}
so scaling by $\bDelta_1$ has no effect on $\bM_1$. Next, transform the auxiliary regressors by
\begin{align}\label{eq:trafoZ2}
	\bZ_{2} := \bX_{2} \bDelta_{2} \bXi^{-1/2}, && Z_{2} := X_{2} \bDelta_{2} \bXi^{-1/2}, && \gamma_{2} := \bXi^{1/2} \bDelta_{2}^{-1} \beta_{2},
\end{align}
where
\begin{align}\label{eq:xiZ}
	\bXi := \frac{\bDelta_2 \bX_{2}^{\top} \bM_{1} \bX_2 \bDelta_2}{n},
\end{align}
and I assumed $\bX_{2}^{\top} \bM_1 \bX_2$ to be positive definite so $\bXi^{1/2}$ exists.
Furthermore, the $k_2 \times k_2$ diagonal matrix $\bDelta_{2} := \diag\left(\bX_{2}^{\top} \bM_{1} \bX_2 / n \right)^{-1/2}$ is chosen such that $\diag(\bXi) = (1, \dotsc, 1)$. Unlike the matrix $\bDelta_1$, the transformation by $\bDelta_2$ serves the dual purpose of improving numerical accuracy and making the WALS NB estimator equivariant to scale transformations of the auxiliary regressors. Otherwise it would be only scale equivariant for the focus regressors \citep[p.~528]{deluca2011stata}.

Notice that combining \eqref{eq:trafoZ2} and \eqref{eq:xiZ} leads to
\begin{align}\label{eq:identity}
	\frac{\bZ_{2}^{\top} \bM_1 \bZ_2}{n} = I_{k_{2}}.
\end{align}
In contrast to \citet{deluca2018glm}, $\bM_1 \bZ_2 / \sqrt{n}$ is not semiorthogonal\footnote{Semiorthogonality is defined as $A A^{\top} = I$ or $A^{\top} A = I$ for a general (non-square) matrix $A$ \citep[p.~104]{zhang2017matrix}.} anymore, since $\bM_1$ is not idempotent.
The transformation further implies
\begin{align*}
	\bZ_2 \gamma_{2} &= \bX_{2} \beta_{2},\\
	\beta_{2} &= \bDelta_2 \bXi^{-1/2} \gamma_{2}.
\end{align*}

Using \eqref{eq:trafoZ1} and \eqref{eq:trafoZ2} I can show for the unrestricted model that
\begin{equation*}
	\eta = Z_1 \gamma_1 + Z_2 \gamma_2 = X_1 \beta_1 + X_2 \beta_2,
\end{equation*}
so the linear predictor stays the same for the unrestricted model. Therefore, all the quantities that only depend on $\bar{\alpha}$ and indirectly on $\bar{\beta}$ via
\begin{equation*}
	\bar{\eta} = X_1 \bbeta_1 + X_2 \bbeta_2 = Z_1 \bgamma_1 + Z_2 \bgamma_2,
\end{equation*}
where $\bgamma_{1} = \bDelta_{1}^{-1} \bar{\beta}_1$ and $\bgamma_{2} = \bXi^{1/2} \bDelta_{2}^{-1} \bar{\beta}_2$, remain the same (e.g.\ $\bar{\mu}$, $\bPsi$, $\bQ$, $\bar{g}$, $\bar{t}$, \ldots), as they do not depend on $\bar{\beta}$ directly.
Note that \citet[p.~6]{deluca2018glm} suggest using the fully iterated unrestricted ML estimates as starting values $\bar{\beta}$ and $\balpha$. In this case, the starting value of the dispersion parameter $\balpha_{Z}$ for the transformed regressors $Z$ is identical to $\balpha$ for the original regressors $X$ since the estimated conditional means are equal, i.e.\ $h(\eta(\bar{\beta}, x_i)) = h(\eta(\bgamma, z_i))$ for all $i$, where $z_{i}^{\top}$ is the $i$th row vector of $Z = (Z_1, Z_2)$.

\subsection{One-step ML estimation of transformed models}\label{sec:onesteptrafo}

It follows from \cref{prop:onestepML} using \eqref{eq:identity} that the one-step ML estimators for the $j$th transformed model are given by
\begin{equation}
	\begin{aligned}\label{eq:trafosol}
		\tgamma_{1j} &= \tgamma_{1r} - \bar{D} W_{j} \tgamma_{2u},\\
		\tgamma_{2j} &= \left( \frac{\bZ_{2}^{\top} \bM_1 \bZ_{2}}{n} \right)^{-1/2} \bW_{Z,j} \tvartheta_Z = W_{j} \tgamma_{2u}, \\
		\talpha_{j} &= -\frac{\bar{t} + \bar{g} \xmat{(y - \bmu)}{\bC} (Z_1 \tgamma_{1j} + Z_2 \tgamma_{2j})}{\bar{g}^2 \xmat{\bk}{\ones} + \bar{\varrho} \xmat{\bkappa}{\ones}},
	\end{aligned}	
\end{equation}
where the fully restricted and unrestricted estimators are
\begin{equation}
	\begin{aligned}\label{eq:trafosol2}
		\tgamma_{1r} ={} & \left( \frac{\xmat{\bZ_1}{\bZ_1}}{n} + \frac{\bar{g} \bepsilon}{n} \xmat{Z_1}{\bq}\xmat{\bq}{Z_1} \right)^{-1} \left( \frac{\xmat{\bZ_1}{\bar{y}_0}}{n} - \frac{\bar{t} \bepsilon}{n} \xmat{Z_1}{\bq} \right) = \bDelta_{1}^{-1} \tbeta_{1r},\\
		\tgamma_{1u} ={} & \tgamma_{1r} - \bQ_{Z} \tgamma_{2u} = \bDelta_{1}^{-1} \left[ \tbeta_{1r} - \bQ \tvartheta \right] = \bDelta_{1}^{-1} \tbeta_{1u}, \\
		\tgamma_{2u} ={} & -\left(\frac{\xmat{\bZ_2}{\bZ_1}}{n} + \frac{\bar{g} \bepsilon}{n} \xmat{Z_2}{\bq} \xmat{\bq}{Z_1} \right) \left(\frac{\xmat{\bZ_1}{\bZ_1}}{n} + \frac{\bar{g} \bepsilon}{n} \xmat{Z_1}{\bq} \xmat{\bq}{Z_1} \right)^{-1} \left(\frac{\xmat{\bZ_1}{\bar{y}_0}}{n} - \frac{\bar{t} \bepsilon}{n} \xmat{Z_1}{\bar{q}} \right) \\
		&+ \left(\frac{\xmat{\bZ_2}{\bar{y}_0}}{n} - \frac{\bar{t} \bepsilon}{n} \xmat{Z_2}{\bar{q}} \right) \\
		={} & \bXi^{1/2} \bDelta_{2}^{-1} \tbeta_{2u},
	\end{aligned}
\end{equation}
with
\begin{equation*}
	\begin{aligned}
	\bQ_{Z} 
	&= \left(\frac{\xmat{\bZ_1}{\bZ_1}}{n} + \frac{\bar{g} \bepsilon}{n} \xmat{Z_1}{\bq} \xmat{\bq}{Z_1} \right)^{-1}  \left(\frac{\xmat{\bZ_1}{\bZ_2}}{n} + \frac{\bar{g} \bepsilon}{n} \xmat{Z_1}{\bq} \xmat{\bq}{Z_2} \right) := \bar{D}.
	\end{aligned}
\end{equation*}
Exploiting $R_{j}^{\top}R_{j} = I_{r_j}$, the following terms simplify
\begin{equation*}
	\begin{aligned}
		\bP_{Z,j} &= R_{j} (\xmat{R_{j}}{R_{j}})^{-1} R_{j}^{\top} = R_{j}R_{j}^{\top} =: P_{j}, \\
		\bW_{Z,j} &= I_{k_2} - \bP_{Z,j} = I_{k_2} - P_{j} =: W_{j}.
	\end{aligned}
\end{equation*}
Analogous to $\tvartheta$, using \eqref{eq:identity} yields
\begin{equation*} 
	\tvartheta_{Z} = \left( \frac{\bZ_{2}^{\top} \bM_1 \bZ_{2}}{n} \right)^{1/2} \tgamma_{2u} = \tgamma_{2u}.
\end{equation*}

As a direct consequence of \eqref{eq:identity}, both $P_{j}$ and $W_{j}$ become nonrandom projection matrices that are different from $\bP_j$ and $\bW_j$ used for the estimation with the untransformed regressors. Furthermore, $W_{j}$ reduces to a diagonal matrix with $k_2 - r_j$ ones and $r_j$ zeros on its main diagonal. The $h$th diagonal element of $W_{j}$ is zero, when the $h$th component of $\gamma_{2}$ is constrained to be zero in the $j$th model. Otherwise, the $h$th component is one. Combining this observation with $\tgamma_{2j}$ from \eqref{eq:trafosol}, it follows that all models that include the $h$th column of $Z_2$ as regressor will have the same estimator for the $h$th component, namely the $h$th component of $\tgamma_{2u}$.

Note that the $j$th model for the transformed regressors is generally not equivalent to the $j$th model of the untransformed regressors because the restriction in \eqref{eq:systemBeta} differs. The exceptions are the unrestricted model $u$ and the fully restricted model $r$, where the restriction is irrelevant:
\begin{enumerate}
	\item For the unrestricted model, the restriction matrix is zero, i.e.\ $R_{u} = 0$.
	\item For the fully restricted model, the restriction matrix is the identity matrix, i.e.\ $R_{r} = I_{k_2}$. However, this is equivalent to estimating an unrestricted model containing only the focus regressors.
\end{enumerate}
This implies that $\tgamma_{2j} \neq \bXi^{1/2} \bDelta_{2}^{-1} \tbeta_{2j}$ for $j \notin \{u, r\}$ and $k_2 \geq 2$ auxiliary regressors. For $k_2 = 1$, there exist only two models: 1.\ the unrestricted and 2.\ the fully restricted model, so $j \in \{u, r\}$. The results are summarized in \cref{cor:equiv}.

\begin{corollary}\label{cor:equiv}
	Let $u$ and $r$ be the indices that denote the unrestricted and fully restricted estimators with $R_{u} = 0$ and $R_{r} = I_{k_2}$, respectively. Then, for $j \neq \{u, r\}$ and $k_2 \geq 2$:
	\begin{equation*}
		\tgamma_{1j} \neq \bDelta_{1}^{-1} \tbeta_{1j}, \qquad
		\tgamma_{2j} \neq \bXi^{1/2} \bDelta_{2}^{-1} \tbeta_{2j}.
	\end{equation*}
	For $k_{2} = 1$, either $j = u$ or $j = r$ holds, so the general relationships for the unrestricted and fully restricted estimators apply:
	\begin{equation*}
		\tgamma_{1u} = \bDelta_{1}^{-1} \tbeta_{1u}, \qquad
		\tgamma_{2u} = \bXi^{1/2} \bDelta_{2}^{-1} \tbeta_{2u},
	\end{equation*}
	and
	\begin{equation*}
		\tgamma_{1r} = \bDelta_{1}^{-1} \tbeta_{1r}, \qquad \tgamma_{2r} = 0.
	\end{equation*}
\end{corollary}

\section{WALS NB model averaging estimator}\label{sec:mavgNB}

Consider the model averaging estimators of $\gamma_1$, $\gamma_{2}$ and $\alpha$
\begin{align*}
	\hat{\gamma}_{1} = \sum_{j=1}^{2^{k_{2}}} \lambda_{j} \tgamma_{1j}, && \hat{\gamma}_{2} = \sum_{j=1}^{2^{k_{2}}} \lambda_{j} \tgamma_{2j}, && \hat{\alpha} =  \sum_{j=1}^{2^{k_{2}}} \lambda_{j} \talpha_{j},
\end{align*}
where $\lambda_j$ are data-dependent model weights satisfying the restrictions
\begin{align}\label{eq:maweights}
	0 \leq \lambda_j \leq 1, && \sum_{j=1}^{2^{k_2}} \lambda_j = 1, && \lambda_j = \lambda_j(\sqrt{n} \tgamma_{2u}).
\end{align}
Note that the regularity condition $\lambda_j = \lambda_j(\sqrt{n} \tgamma_{2u})$
is equivalent to the condition on the model weights used by \citet{hjort2003ma}.

From \eqref{eq:trafosol} I get
\begin{equation}
\begin{aligned}\label{eq:gammahat}
	\hat{\gamma}_{1} &= \tgamma_{1r} - \bar{D} W \tgamma_{2u} = \tgamma_{1r} - \bar{D} \hat{\gamma}_{2},\\
	\hat{\gamma}_{2} &= W \tgamma_{2u}, \\
	\hat{\alpha} &= -\frac{\bar{t} + \bar{g} \xmat{(y - \bmu)}{\bC} (Z_1 \hat{\gamma}_{1} + Z_2 \hat{\gamma}_{2})}{\bar{g}^2 \xmat{\bk}{\ones} + \bar{\varrho} \xmat{\bkappa}{\ones}},
\end{aligned}
\end{equation}
where $W = \sum_{j=1}^{2^{k_2}} \lambda_j W_{j}$ is a diagonal matrix with entries $w_{h} \in [0, 1]$, because $W_{j}$ is a diagonal matrix with entries $w_{j,h} \in \{0, 1\}$, $h = 1, 2, \dotsc, k_2$ (notice the slight abuse of notation: $h$ is used as an index here and does not refer to the inverse link). Next, I can transform $\hat{\alpha}$ to an estimate for $\rho$ by applying the inverse of the log-link, i.e.\ $\hat{\rho} = \exp(\hat{\alpha})$. 
Furthermore, using $\hat{\gamma}_1$ and $\hat{\gamma}_2$, the WALS estimators of the original parameters $\beta_1$ and $\beta_{2}$ are given by
\begin{align*}
	\hat{\beta}_{1} = \bDelta_1 \hat{\gamma}_{1}, && \hat{\beta}_{2} = \bDelta_2 \bXi^{-1/2} \hat{\gamma}_{2}.
\end{align*}

The final step in completing the WALS NB model averaging estimator is to estimate the model weights $\lambda_j$. However, notice that both $\hat{\gamma}_1$ and $\hat{\alpha}$ can be expressed as functions of $\hat{\gamma}_2$. Therefore, it is sufficient to find an expression for $\hat{\gamma}_2$ instead of directly estimating the weights $\lambda_j$. Similar to \citet[p.~6]{deluca2018glm}, I construct $\hat{\gamma}_{2}$ as a Bayesian shrinkage estimator by exploiting the approximate normality and independence of $\tgamma_{2u}$ under the local misspecification framework \citep[see e.g.][]{hjort2003ma}. First, let the auxiliary parameters be $\beta_{2} = \delta / \sqrt{n}$, where $\delta$ is an unknown constant vector that represents the departure of the DGP from the unrestricted model. Then, if the fully iterated ML estimator of the unrestricted model is used as starting values $\bar{\beta}_1$, $\bar{\beta}_2$ and $\bar{\alpha}$ and mild regularity conditions are assumed, I can show that
\begin{equation}\label{eq:approxGamma2}
 \sqrt{n} \tgamma_{2u} \approx \normal( \sqrt{n} \gamma_{2n}, I_{k_2}) = \normal(d, I_{k_2}),
\end{equation}
in large samples, where $\gamma_{2n} = d / \sqrt{n}$, $d$ = $\Xi^{1/2} \Delta_{2}^{-1} \delta$, $\Xi = \plim \bXi$ and $\Delta_{2} = \plim \bDelta_{2}$ (see the \refSup\ for more details). Further, consider $\hat{\gamma}_2$ from \eqref{eq:gammahat} and assume analogously to \citet[p.~6]{deluca2018glm} that each diagonal element $w_h, h = 1, 2, \dotsc, k_2$, of $W$ only depends on the $h$th component $\sqrt{n} \tgamma_{2u,h}$ of $\sqrt{n} \tilde{\gamma}_{2u}$. Then, \eqref{eq:approxGamma2} implies that the components of $\hat{\gamma}_{2}$ are also approximately independent. This assumption further simplifies the estimation problem by reducing the $k_2$-dimensional problem of estimating $\hat{\gamma}_{2}$ to $k_2$ times a one-dimensional problem of estimating each element of $\hat{\gamma}_{2}$. Moreover, $\hat{\gamma}_{2,h}$ is a shrunken version of $\tgamma_{2u,h}$ because $0 \leq w_h \leq 1$, therefore, $\hat{\gamma}_{2}$ is a shrinkage estimator of $\gamma_{2n}$.

The previous two observations suggest that the Bayesian posterior mean is a suitable shrinkage estimator for $\sqrt{n} \gamma_{2n,h}$. Thus, the $h$th component of the WALS NB estimator $\hat{\gamma}_{2}$ follows as
\begin{align}\label{eq:gammaest}
	\hat{\gamma}_{2,h} = \frac{\E(\sqrt{n} \gamma_{2n,h} | \sqrt{n} \tgamma_{2u,h})}{\sqrt{n}} = \frac{\E(d_h | \sqrt{n} \tgamma_{2u,h})}{\sqrt{n}}, \quad d_h \sim f,
\end{align}
where $\sqrt{n} \tgamma_{2u,h} \approx \normal(d_h, 1)$ with prior mean $d_h$, which is the $h$th element of $d$ and is assumed to have a symmetric and unimodal prior $f$ (see section 9 of \citet{magnus2016wals} for more details on the prior and the estimation). Notice again that $\hat{\gamma}_{2,h}$ lies between 0 and the `observed data' $\tgamma_{2u,h}$.

\citet{magnus2016wals} require the desirable properties of robustness\footnote{A prior $\pi(\gamma)$ is robust if the posterior mean $m(x)$ based on $\pi$ satisfies $x - m(x) \rightarrow 0$ as $x \rightarrow \infty$.}, neutrality\footnote{A prior $\pi(\gamma)$ is neutral if the prior median of $\gamma$ is zero and the prior median of $|\gamma|$ is one.} and minimax regret\footnote{Regret is defined as difference between risk and the infimum of risk, where risk is defined as expected squared loss.} for the prior $f$, which further motivates the use of the Bayesian posterior mean as the shrinkage estimator in $\hat{\gamma}_{2,h}$. The reflected Weibull, under suitable parameter values, is a prior that fulfills all the properties mentioned above. In contrast, the Laplace prior is neutral but not robust \citep[p.~132]{magnus2016wals}. However, it admits a closed-form expression for the posterior mean in \eqref{eq:gammaest} \citep[see e.g.\ Theorem 1 in][]{magnus2010growth} and therefore calculating the posterior mean under the Laplace prior is computationally less complex than under the reflected Weibull, which requires numerical integration.

\section{Performance metrics}
\label{sec:score}

In order to compare the performance of WALS NB with other methods, I first need to define performance metrics. The classical performance measure for regression is the RMSE, which is given by
\begin{equation}
	\mathrm{RMSE} = \sqrt{\frac{1}{n} \sumin (\hat{\mu}_i - y_i)^2},
\end{equation}
where $\hat{\mu}_i$ is the predicted mean for observation $i$. However, I would like to evaluate the fit of an entire distribution and not only the expectation. Traditional measures used in machine learning such as (R)MSE only focus on point predictions, i.e.\ the conditional expectation of the fitted distribution, in relation to the observed values and do not make judgment on other aspects of the fitted distribution. \citet{czado2009predictive} recommend scoring rules for evaluation of count data models, which have also been used in \citet{kolassa2016count}. WALS NB and all other methods considered in this paper fit an entire (conditional) distribution for each individual that allows probabilistic predictions/forecasts, which is exactly the scenario for which scoring rules provide quality assessment \citep[p.~359]{gneiting2007scores}. For count data, a probabilistic forecast is a predictive probability distribution $\hat{P}$ on the set of nonnegative integers $\mathbb{N}_{0}$ \citep[p.~1254]{czado2009predictive}.

Following \citet[p.~1256]{czado2009predictive}, I take scoring rules to be penalties I wish to minimize. Specifically, the penalty $s(\hat{P}, y)$ is incurred when the forecaster quotes predictive distribution $\hat{P}$ and count $y$ is realized. Moreover, let $s(\hat{P}, Q)$ denote the expected value of $s(\hat{P}, \cdot)$ under distribution $Q$
\begin{equation*}
	s(\hat{P}, Q) = \int s(\hat{P}, y) dQ(y).
\end{equation*}
In practice, the average over suitable pairs $(\hat{P}, y)$ is used:
\begin{equation}\label{eq:scoreavg}
	S := \frac{1}{n}\sum_{i = 1}^{n} s(\hat{P}_i, y_i),
\end{equation}
where $\hat{P}_i$ refers to the $i$th predictive distribution and $y_i$ the $i$th observed count. In the simulation experiment and empirical application of \Cref{sec:sim,sec:cvexpvar}, respectively, scores will always refer to a suitable average.

Suppose the forecaster has predictive distribution $Q$ available. Then the forecaster has no incentive to predict any $\hat{P} \neq Q$ and is encouraged to quote her true belief, $\hat{P} = Q$, if the scoring rule is \textit{strictly proper}. Strict propriety is defined by
\begin{equation*}\label{eq:strictproper}
	s(Q, Q) \leq s(\hat{P}, Q),
\end{equation*}
with equality if and only if $\hat{P} = Q$, and encourages honest quotes (\citealp[p.~1256]{czado2009predictive}; \citealp[p.~360]{gneiting2007scores}). If $s(Q, Q) \leq s(\hat{P}, Q)$ for all $\hat{P}$ and $Q$, then the scoring rule is only proper. Since only strict propriety ensures that both calibration (consistency with actual realizations) and sharpness (concentration of the predictive distribution) of the predictive distribution are addressed \citep{winkler1996scores}, I exclusively use strictly proper scoring rules.

\citet[p.~1256~f.]{czado2009predictive} propose a number of strictly proper scoring rules for count data. It is a priori unclear which scoring rule to use unless there is a unique and clearly defined underlying decision problem. Since probabilistic forecasts often have many uses, it is appropriate to use a variety of scores to take advantage of their differing emphases \citep[p.~1257]{czado2009predictive}. In this paper, I use the logarithmic (log), Brier and spherical score, which I briefly summarize here: Let $\hat{p}_y := \hat{P}(Y = y)$ denote the probability mass at count $y$ (for continuous distributions it is the density at $y$), then the \textit{log score} is defined as
\begin{equation}\label{eq:logscore}
		\logs(\hat{P}, y) = -\log(\hat{p}_y).
\end{equation}
The sum of log scores corresponds to the negative log-likelihood. Further define
\begin{equation}
	||\hat{p}||^2 = \sum_{r = 0}^{\infty} \hat{p}_{r}^2, \label{eq:infsumsq}
\end{equation}
where the infinite sum may be truncated if no closed-form expression exists. The \textit{quadratic score}, also called \textit{Brier score}, is then
	\begin{equation}\label{eq:brier}
		\qs(\hat{P}, y) = -2\hat{p}_y + ||\hat{p}||^2.
	\end{equation}
The \textit{spherical score} uses the same components differently:
\begin{equation}
	\sphs(\hat{P},y) = - \frac{\hat{p}_y}{||\hat{p}||}.
\end{equation}

\section{Simulation experiment}\label{sec:sim}

The aim is to compare the performance of WALS NB with the traditional ML estimator of the NB2 regression model in a controlled environment. The DGP is inspired by the local misspecification framework so I can assess the influence of varying numbers of focus and auxiliary regressors.

The dependent count variable is sampled from an NB2 using a log-link, i.e.\
\begin{equation}\label{eq:dgpSim}
	\begin{aligned}
		y_i | x_i &\sim \mathrm{NB2}(\mu_i, \rho),\\
		\mu_i &= \exp(\alpha + x_{i} ^{\top} \beta), \\
		x_i \overset{\mathrm{i.i.d.}}&{\sim}  \normal(0, \Sigma_{k}),
	\end{aligned}
\end{equation}
for $i = 1, 2, \dotsc, n$, where $x_i = (x_{i1}^{\top}, x_{i2}^{\top})^{\top}$ is a random vector of dimension $k = k_1 + k_2$  composed of $k_1$ focus regressors $x_{i1}$ and $k_2$ auxiliary regressors $x_{i2}$. Analogously, the coefficient vector is separated into two parts: $\beta = (\beta_{1}^{\top}, \beta_{2}^{\top})^{\top}$. Inspired by the simulation experiments in \citet{zhang2019inference} and \citet{deluca2023interval}, who compare confidence and prediction intervals of model averaging methods for the linear regression model, I choose the regressors to be multivariate normal because it allows me to analyze the effect of the correlation between the regressors on the performance of the methods. For simplicity, I specify each element of $x_i$ to have variance $1$ and pairwise correlation $b$, i.e.\
\[
	\Sigma_{k} = \begin{pmatrix}
		1		& b 		& \cdots 	& b \\
		b		& 1 		& \cdots	& b \\ 
		\vdots 	& \vdots 	& \ddots 	& \vdots \\
		b	 	& b			& \cdots 	& 1
	\end{pmatrix}.
\]
The same offset $\alpha = \log(3)$ is used in all experiments such that the DGP produces reasonable counts, see \autoref{fig:histSimRun1} for a visualization of a training set from a specific run.

\begin{figure}
	\centering
	\includegraphics[width=0.75\textwidth]{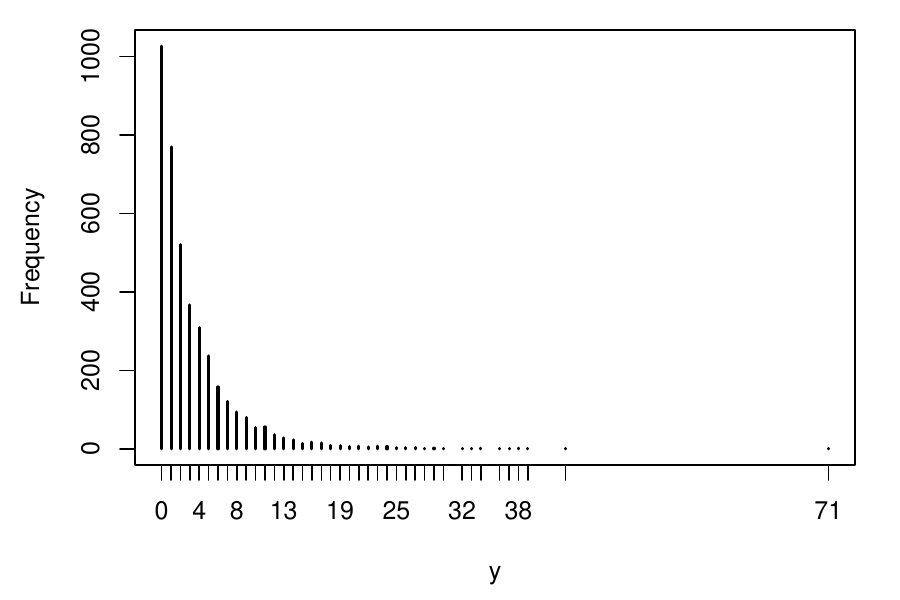}
	\caption{Visualization of count variable $y$ in the training set of the first simulation run of setting $k_1 = 10, k_2 = 100, \rho = 1, b = 0$ with $n = 4000$ observations.}\label{fig:histSimRun1}
\end{figure}

Moreover, the regression coefficients are generated as follows: Define the vectors $\bar{\beta}_{1} := (\bar{\beta}_{1,1}, \bar{\beta}_{1,2}, \dotsc, \bar{\beta}_{1,10})^{\top}$ and $\bar{\beta}_{2} := (\bar{\beta}_{2,1}, \bar{\beta}_{2,2}, \dotsc, \bar{\beta}_{2,100})^{\top}$. Then, the maximum number of regression coefficients $k_1 = 10$ and $k_2 = 100$ are randomly sampled \textit{once} according to the following rules:
\begin{itemize}
	\item $\bar{\beta}_{1}$: Each element $\bar{\beta}_{1,j}, j = 1, 2, \dotsc, 10$, is drawn independently with 50\% chance from $U(-0.25,-0.1)$ or from $U(0.1, 0.25)$, so that both positive and negative coefficients are present.
	\item $\bar{\beta}_{2}$: Each element is drawn independently from a uniform, i.e.\
	\[\bar{\beta}_{2,m} \sim U(-0.01, 0.01), \quad m = 1, 2, \dotsc, 100.\]
\end{itemize}
The simulations then only take the first $k_1$ and $k_2$ values from these vectors as regression coefficients $\beta_1$ and $\beta_2$. For example, in the setting $k_1 = 5, k_2 = 10$, $\beta_1 = (\bar{\beta}_{1,1}, \bar{\beta}_{1,2}, \dotsc, \bar{\beta}_{1,5})^{\top}$ and $\beta_{2} = (\bar{\beta}_{2,1}, \bar{\beta}_{2,2}, \dotsc, \bar{\beta}_{2,10})^{\top}$. Hence, the magnitude of the elements in $\beta_2$ is much smaller than in $\beta_1$, therefore the regressors $x_{i2}$ are considered auxiliary regressors and the main variation is driven by $x_{i1}$. \autoref{tab:barbeta1} and \ref{tab:barbeta2} in \Cref{sec:appSim} show the entries of $\bar{\beta}_1$ and $\bar{\beta}_2$, respectively.

\begin{table}
	\caption{Simulation design -- Parameters and choice of values}\label{tab:simPars}
	\begin{center}
		\begin{threeparttable}
		\begin{small}
			\begin{tabular}{llr}
    \toprule\toprule 
    Parameter & Description & Values \\ 
    \midrule
    $k_1$ & \# focus regressors & 1, 5, 10 \\ 
    $k_2$ & \# auxiliary regressors & 1, 5, 10, 20, 50, 100 \\ 
    $n$ & \# training observations & 500, 1000, 2000, 3000, 4000 \\ 
    $b$ & correlation coefficient of regressors & 0, 0.3, 0.5, 0.9 \\
    $\alpha$ & constant & $\log(3)$ \\ 
    $\rho$ & dispersion parameter & 1, 1.5, 2\\ 
    $R$ & \# runs & 300 \\ 
    \bottomrule
\end{tabular}
			\end{small}
		\begin{tablenotes}
			\footnotesize
			\item[--] Each simulation scenario is a unique combination of the parameter values, which produces 1080 scenarios in total. 
		\end{tablenotes}
		\end{threeparttable}
	\end{center}
	\end{table}

All values of the parameters used in the experiment are summarized in \autoref{tab:simPars}. A total of 1080 scenarios consisting of all combinations of the parameters are simulated for $R = 300$ runs each.

I compare six different procedures that are named according to the pattern `method-specification'. The two methods are called `walsNB', which estimates the NB2 regression model using WALS NB, and `ML', which uses maximum likelihood. For WALS NB procedures,	 the Weibull prior is used as it theoretically provides the best tradeoff between robustness and regret, for more details see \citet[p.~130~ff.]{magnus2016wals}. The results for other priors are expected to be quite similar as WALS for the linear regression model has empirically shown to be relatively insensitive to the choice of the prior \citep{deluca2022sampling}.\footnote{I also conducted the simulation experiment using the Laplace prior and the results are similar to the ones using the Weibull prior.} 

The procedures considered are
\begin{enumerate}
	\item walsNB-dgp: Emulates the DGP \eqref{eq:dgpSim} by including $x_{i1}$ and a constant as focus regressors and $x_{i2}$ as auxiliary regressors.
	\item walsNB-aux: Includes only a constant as focus regressor and the `true' regressors $x_i$ as auxiliary.
	\item ML-U: Includes a constant and the `true' regressors $x_i$.
	\item ML-focus: Only includes the focus regressors $x_{i1}$ and a constant.
	\item ML-AC: Only includes the auxiliary regressors $x_{i2}$ and a constant.
	\item oracle: The true model of the DGP \eqref{eq:dgpSim} that is not estimated.
\end{enumerate}

The second WALS NB specification, walsNB-aux, is included to analyze the extent to which prior information about the focus regressors in walsNB-dgp affects performance. Ideally, including $x_{i1}$ as focus regressors in the procedure should improve performance as they are the covariates that dominate and should therefore be included in all submodels of WALS NB. However, in walsNB-aux their coefficients are also subjected to the regularization of the Bayesian estimation step, which may improve performance. Thus, a priori it is unclear which model will dominate.

All ML specifications are estimated using a log-link for the mean parameter, while the dispersion parameter is estimated directly without a link (default setting). Moreover, I increase the maximum number of iterations for both the alternation process between IRLS and ML estimation of $\rho$ and the IRLS algorithm itself from the default setting of 25 to 2500 to increase the odds for convergence. The remaining settings, e.g.\ convergence criteria, are left at their default values.

Moreover, the WALS NB specifications use a $\log$-link for the mean and dispersion parameter following the DGP \eqref{eq:dgpSim} and are initialized using the ML estimates of the unrestricted model, which are given by the ML-U procedure (using the increased maximum number of iterations as described above). This initialization is recommended by \citet[p.~4]{deluca2018glm} as it produces lower RMSE for the WALS GLM estimator in their Monte Carlo simulations \citep[see Table 3 in][p.~11]{deluca2018glm} compared to using the estimates of the fully restricted model as starting values. It further ensures that \eqref{eq:approxGamma2} approximately holds for the one-step estimators of the auxiliary regression coefficients, which I exploit in the Bayesian estimation step to reduce the $k_2$-dimensional posterior mean estimation to $k_2$ one-dimensional problems (see \Cref{sec:mavgNB}).

Finally, the Weibull prior for all WALS specifications uses the parameters recommended in \citet[p.~132]{magnus2016wals}, which are minimax regret solutions for the normal location problem. 

In order to compare the performance of the procedures, I follow the benchmark experiments framework of \citet[p.~681~f.]{hothorn2005design} and more specifically the `Simulation Problem'. The simulation is structured to emulate the typical use of the methods: For each scenario and run, a training sample of size $n$ is generated, where all procedures are applied and performance criteria are computed on an independently generated validation set that is fixed in size to $n_e = 4000$ to avoid any variation due to its size. I do not employ hypothesis tests to check if the performance differences are significant because the simulation experiment itself is already computationally intensive due to the large number of parameter settings.

I consider the RMSE as a classical precision measure for regression and additionally log, Brier and spherical scores to assess the distributional fit as described in \Cref{sec:score}. For the scoring rules, the average is taken over the validation sample as in \eqref{eq:scoreavg}. Further, I truncate the infinite sum in $|| \hat{p} ||$ from \eqref{eq:infsumsq} used in the Brier and spherical score at the count $r = 150$ because the response typically does not exceed 150 and it would be meaningless to extrapolate beyond the observed data. See \autoref{fig:histSimRun1} for a visualization of the training data of a single simulation run of the setting $k_1 = 10, k_2 = 100, \rho = 1, b = 0$ (this setting should maximize the range of $y$, since $k_1 = 10$ and $k_2 = 100$ allow for the largest possible means $\mu_i$ and $\rho = 1$ maximizes variance), where the response only ranges between 0 and 71.

In the following, only the scenarios highlighting the differences between the procedures are discussed. For more results, see the \refSup. 

\subsection{Varying the number of regressors}

First, I analyze how the procedures behave when the number of regressors is varied. In the following plots, the points represent the mean validation metric over all successful runs of the experiment, i.e.\ $R = 300$ if the method never fails to converge. The total number of failed runs is given below the corresponding points in the plots and the shaded area displays the interquartile range (the box of a boxplot) of the validation metric.

\begin{figure}[tbh]
	\centering
	\includegraphics[width=0.9\textwidth]{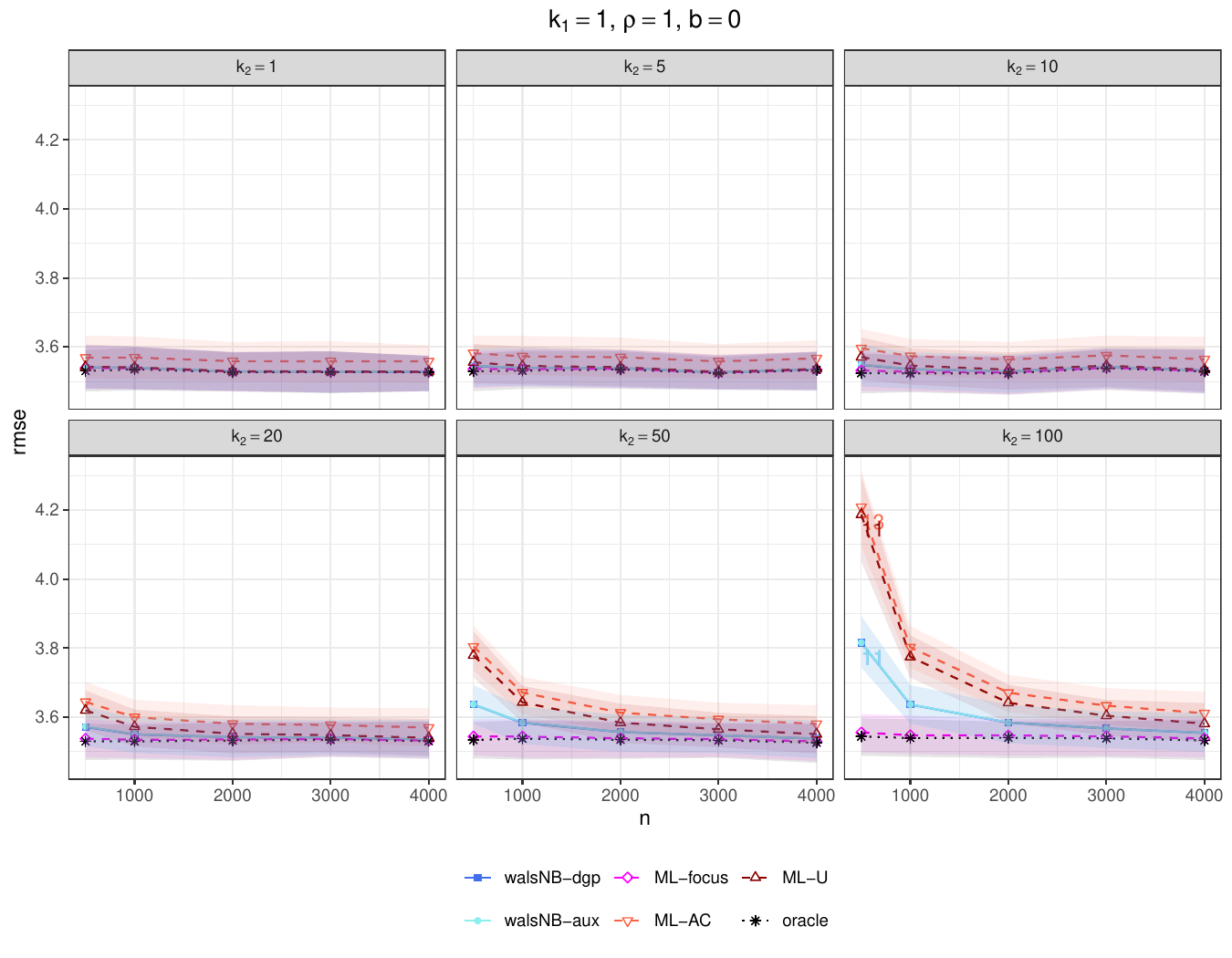}
	\caption{Mean validation RMSE and quartiles varying $n$ and $k_2$}\label{fig:simRMSEk2}
	\justifying
	\footnotesize
	\noindent
	The remaining parameters are fixed at $k_1 = 1, \rho = 1$ and $b = 0$. The shaded areas show the interquartile range. The number below a point indicates how often the method failed to converge in this particular setting.
\end{figure}

\autoref{fig:simRMSEk2} shows that walsNB-aux performs similarly to walsNB-dgp in terms of mean validation RMSE when we vary $k_2$ with fixed $k_1 = 1$, $\rho = 1$ and $b = 0$, because most of the regressors are auxiliary and the former includes all regressors as auxiliary. Moreover, both WALS NB specifications outperform ML-U on average when $k_2 \geq 20$ and $n$ is small ($n \leq 2000$). In fact, WALS NB specifications show lower mean validation RMSE than all ML specifications in these scenarios, except for ML-focus that contains only the focus regressors. For $k_2 = 100$ and $n < 2000$ the `typical' performance of walsNB-dgp and walsNB-aux is also better than ML-U and ML-AC as their interquartile ranges do not overlap. On the other hand, when $k_2$ is small and/or $n$ is large, their interquartile ranges are similar. The largest difference in mean RMSE is observed at $n = 500$, $k_1 = 1$, $k_2 = 100$ where walsNB-aux exhibits around 8.9\% lower mean RMSE than ML-U. ML-focus performs the best in all scenarios, especially when $k_2$ is large and $n$ small.

Therefore, if we know the focus regressors, then ML-focus yields the best fit in very sparse situations with few observations. Otherwise, walsNB-dgp and walsNB-aux are better than using all regressors in the large regression model ML-U. The outperformance in walsNB-dgp and walsNB-aux compared to ML-U is likely due to the reduced variance thanks to the Bayesian regularization step, which typically reduces variance and leads to lower RMSE via the bias-variance trade-off. In reality it is unlikely that we can exactly identify which regressors are the focus regressors, so walsNB-aux offers a great alternative that does not require variable selection. In all scenarios, it performs at least as well as ML-U but better when the data is sparse and few observations are available. For large $n$ or small $k_2$, all procedures fit the data equally well as their mean validation RMSE converges to the RMSE of the oracle.

\begin{figure}[tbh]
	\centering
	\includegraphics[width=0.9\textwidth]{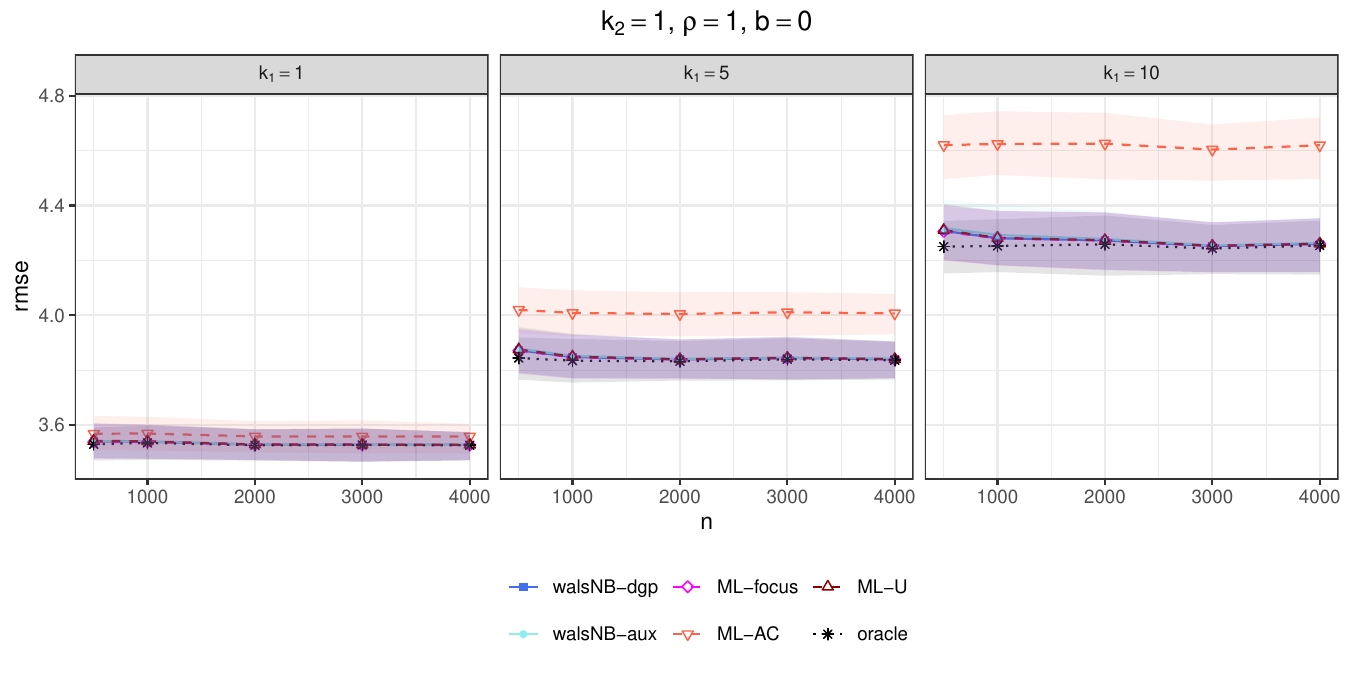}
	\caption{Mean validation RMSE and quartiles varying $n$ and $k_1$}\label{fig:simRMSEk1}
		\justifying \footnotesize \noindent The remaining parameters are fixed at $k_2 = 1, \rho = 1$ and $b = 0$. The shaded areas show the interquartile range.
\end{figure}

When I vary $k_1$ with fixed $k_2 = 1$, $\rho = 1$ and $b = 0$ in \autoref{fig:simRMSEk1}, the picture changes. In all scenarios, ML-AC returns the highest RMSE and the remaining specifications perform similarly as their interquartile ranges overlap. Increasing $k_1$ shifts the `RMSE-curve' up for all procedures, including the oracle, while retaining their relative order. This behavior is explained by the form of the variance of the NB2 distribution in \eqref{eq:nb2meanvar}. The more focus regressors with large coefficients are included, the more likely it is that the conditional expectation $\mu_i$ is large, which increases the variation of the response $y_i$ since the conditional variance is monotonically increasing in $\mu_i$. Thus, even if we could exactly estimate the true $\beta_1$ and $\beta_2$, the RMSE would still increase due to the increased conditional variance, which is demonstrated by the behavior of the oracle.

\begin{figure}[tbh]
	\centering
	\includegraphics[width=0.9\textwidth]{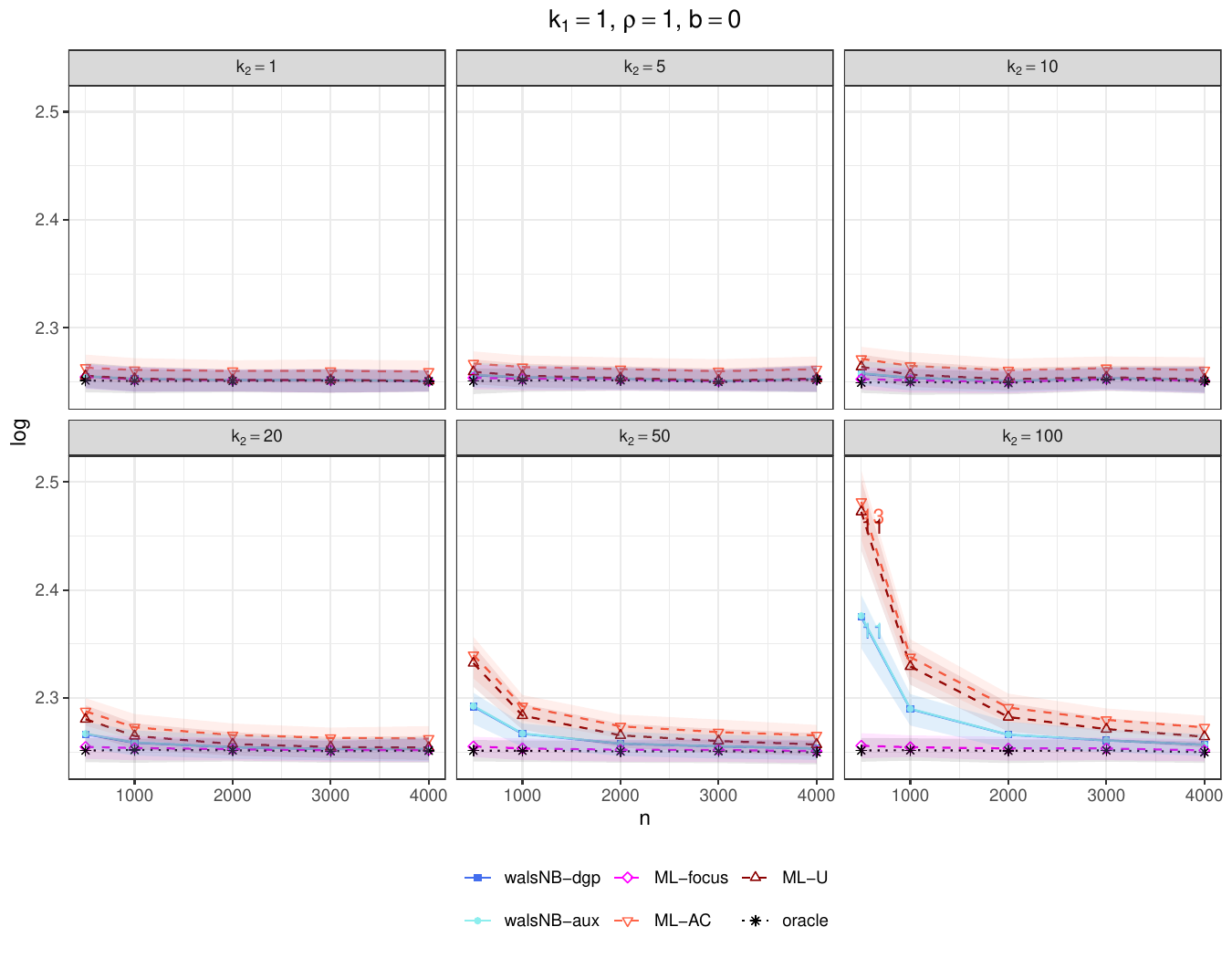}
	\caption{Mean validation log score and quartiles varying $n$ and $k_2$}\label{fig:simLogk2}
	\justifying \footnotesize \noindent
	The remaining parameters are fixed at $k_1 = 1, \rho = 1$ and $b = 0$. The shaded areas show the interquartile range. The number below a point indicates how often the method failed to converge in this particular setting.
\end{figure}

The same patterns hold for the validation log score. Firstly, \autoref{fig:simLogk2} shows that WALS NB specifications generally perform better than ML specifications in terms of log score, when the number of auxiliary regressors is high compared to the number of focus regressors and few observations are available. The exception is again ML-focus, which performs the best across all scenarios. The largest difference in mean log score between walsNB-aux and ML-U is realized at $k_1 = 1$, $k_2 = 100$ and $n = 500$ where the mean log score of walsNB-aux is around 3.9\% lower. Moreover, the typical performance of walsNB-aux is also better than ML-U in this scenario as their interquartile ranges do not overlap. For large $n$ the distributional fit of all models is similar because the mean log scores converge to the log score of the oracle.

\begin{figure}[tbh]
	\centering
	\includegraphics[width=0.9\textwidth]{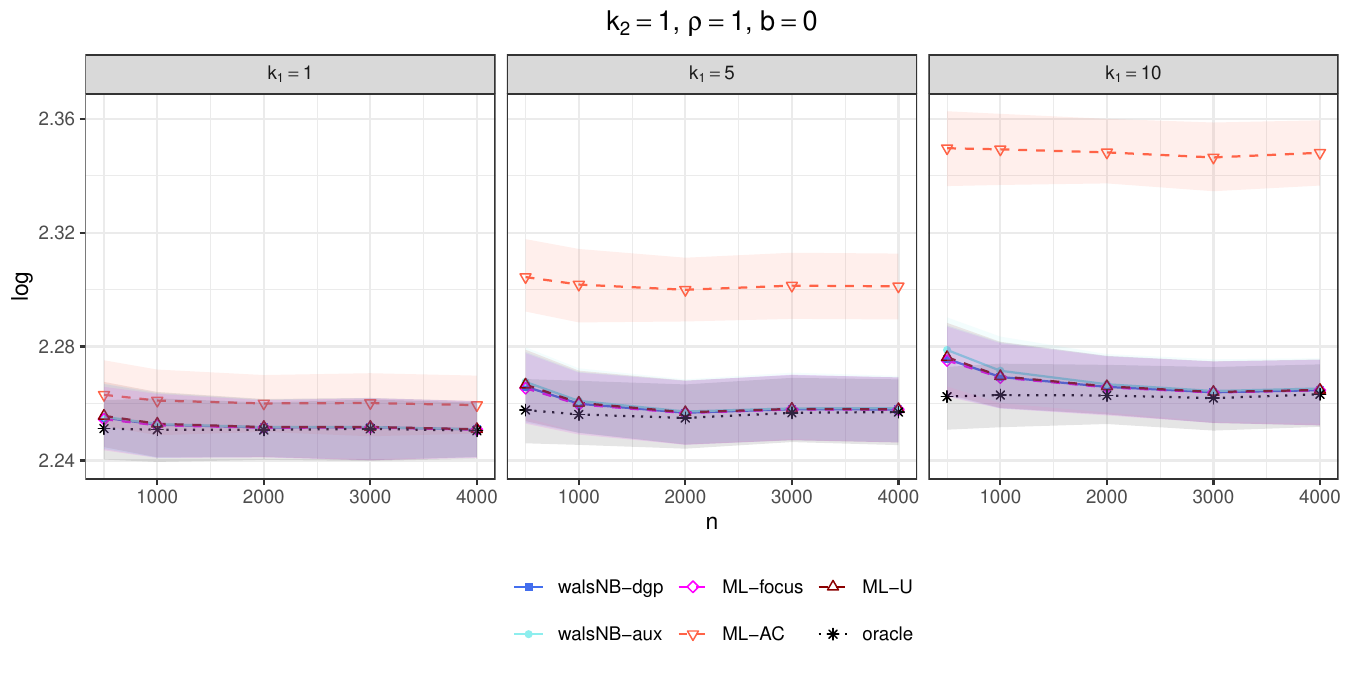}
	\caption{Mean validation log score and quartiles varying $n$ and $k_1$}\label{fig:simLogk1}
	\justifying \footnotesize \noindent
	The remaining parameters are fixed at $k_2 = 1, \rho = 1$ and $b = 0$. The shaded areas show the interquartile range.
\end{figure}

Secondly, similar to the results for RMSE in \autoref{fig:simRMSEk1}, I find a small upwards shift of the mean validation log scores in \autoref{fig:simLogk1} when increasing $k_1$ given $k_2 = 1$.
As expected, the distributional fit of ML-AC, which only includes the auxiliary regressors, is the worst among the procedures when $k_1 \geq 5$. Finally, the interquartile range of all models except ML-AC overlap, so their performance in terms of log score typically does not differ. The relative ranking of the procedures for the Brier and spherical score is the same as for the log score, so their results are only shown in the \refSup.

In summary, the WALS NB specifications generally outperform ML-U in terms of RMSE and log score when the number of auxiliary regressors is very large relative to the number of focus regressors and when the number of observations is small. This is in line with the results from \citet{abadie2019regularized} for the pretest estimator, which is the predecessor of the WALS estimator: The authors consider a `Spike and Normal' process for noisy estimates $\hat{\mathcal{X}}_1, \hat{\mathcal{X}}_2, \dotsc, \hat{\mathcal{X}}_k$ of e.g.\ the coefficients from a linear regression model \citep[p.~746]{abadie2019regularized}: The estimates $\hat{\mathcal{X}}_j$ are assumed to follow $\hat{\mathcal{X}}_j \sim \normal(m_j, s_j^2)$ for $j = 1, 2, \dotsc, k$, e.g.\ $\hat{\mathcal{X}}_j$ are elements of the ordinary least squares (OLS) estimator in a linear regression model with homoskedastic normal error terms. In this setup, the mean $m_j$ can be regarded as the true value of the regression coefficient that is estimated as $\hat{\mathcal{X}}_j$. The idea is that regularized estimators such as lasso, ridge, and pretest modify the OLS estimator $\hat{\mathcal{X}}_j$. Furthermore, the mean $m_j$ is set to zero (spike) with a fixed probability $p$, and with probability $1-p$ the coefficient follows $m_j \sim \normal(m_0, s_{0}^2)$ for all $j$. Under this setting, the authors show that the pretest estimator exhibits smaller integrated risk (integrated expected squared error over the space of distributions of the data distribution, see \citet[p.~745~f.]{abadie2019regularized} for details) than lasso and ridge, when the process is very sparse, i.e.\ $p$ is high and $m_0$ is large, so many coefficients are set to zero and the non-zero coefficients are far away from zero. The results further agree with the Monte Carlo simulations of \citet{deluca2023interval} for WALS in the linear regression model: The authors find that the ratio of the MSE of OLS relative to the MSE of WALS increases when the number of auxiliary regressors $k_2$ becomes larger. Moreover, for all $k_2$, the ratio decreases when the sample size $n$ increases. Both observations are in line with \autoref{fig:simRMSEk2}, where walsNB-dgp dominates in terms of (R)MSE compared to the unrestricted estimator ML-U when $n$ is small and $k_2$ is large.

\subsection{Varying $\rho$ and $b$}

I fixed $k_1 = k_2 = 5$ so that varying the correlation $b$ affects the correlation within focus and auxiliary regressors, as well as the correlation between focus and auxiliary regressors. In contrast, if I had set $k_1 = k_2 = 1$, only the correlation between focus and auxiliary regressors would be modified.

\begin{figure}[tbh]
	\centering
	\includegraphics[width=0.87\textwidth]{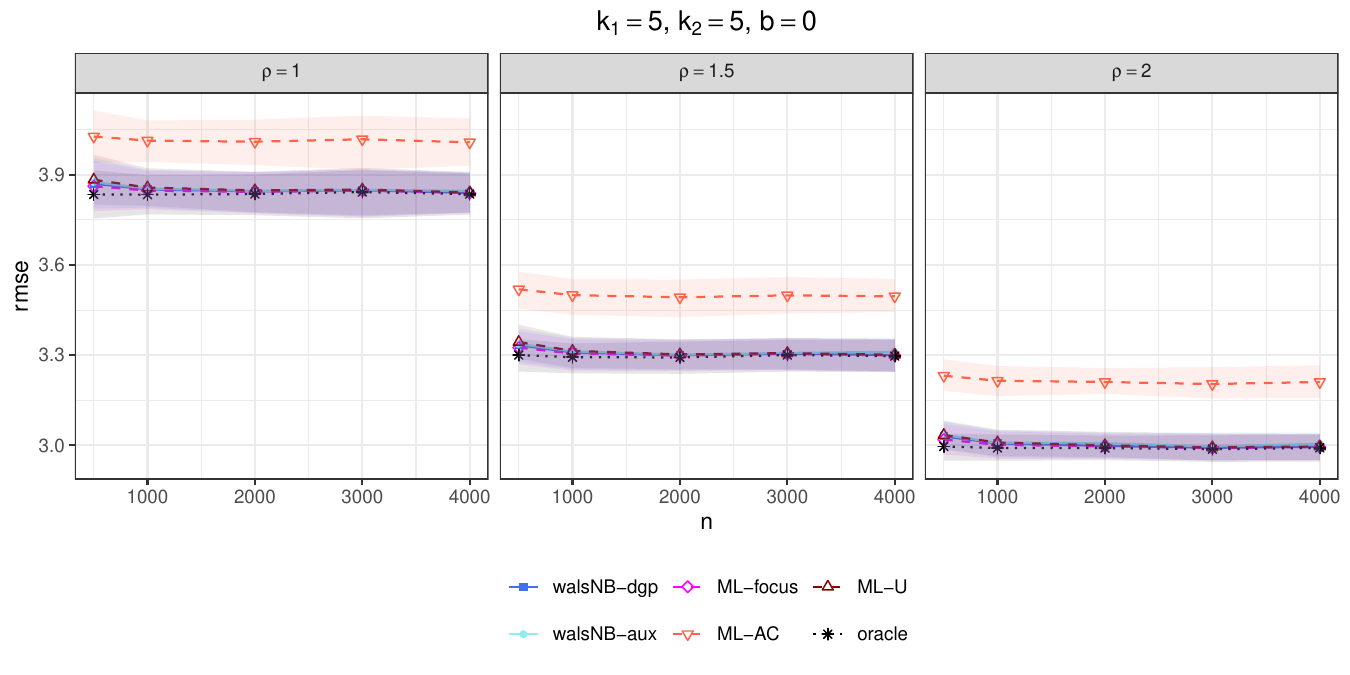}
	\caption{Mean validation RMSE and quartiles varying $n$ and $\rho$}\label{fig:simRMSErho}
	\justifying \footnotesize \noindent The remaining parameters are fixed at $b = 0$ and $k_1 = k_2 = 5$. The shaded areas show the interquartile range.
\end{figure}

\autoref{fig:simRMSErho} shows that for fixed $b = 0$ and $k_1 = k_2 = 5$, all procedures yield similar mean validation RMSE across all $\rho$, except for ML-AC, which exhibits much higher values compared to the other methods.
Note that the mean validation RMSE generally decreases for all procedures, even the oracle, when $\rho$ increases. This is due to the fact that higher $\rho$ leads to less overdispersion, i.e.\ lower conditional variance, resulting in lower RMSE for all procedures.

\begin{figure}[tbp]
	\centering
	\includegraphics[width=0.87\textwidth]{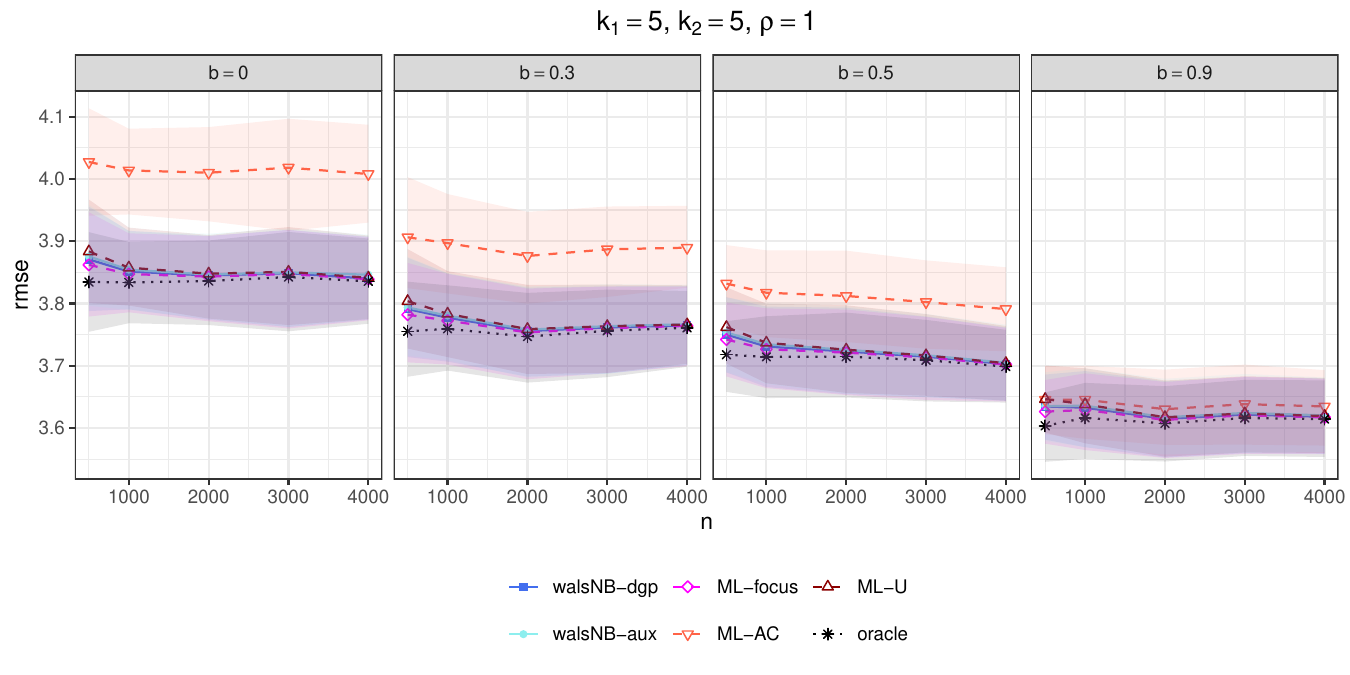}
	\caption{Mean validation RMSE and quartiles varying $n$ and $b$}\label{fig:simRMSEcorr}
	\justifying \footnotesize \noindent The remaining parameters are fixed at $\rho = 1$ and $k_1 = k_2 = 5$. The shaded areas show the interquartile range.
\end{figure}

Increasing the correlation $b$ between all regressors for fixed $\rho = 1$ and $k_1 = k_2 = 5$ in \autoref{fig:simRMSEcorr}, the mean validation RMSE shifts down for all procedures, especially for ML-AC as it only includes the auxiliary regressors and a constant. The larger the correlation, the better it can compensate the lack of focus regressors. Generally, the choice of regressors matters less for prediction when the regressors are highly correlated as each of them will contain similar information for the prediction task. The remaining procedures perform very similarly when increasing $b$ and converge to the mean validation RMSE of the oracle for large $n$. Except for ML-AC in the cases with $b < 0.9$, the typical RMSE of the procedures are comparable as the interquartile ranges overlap and have similar widths in all scenarios.

The results for the mean validation log score when varying $\rho$ and $b$ with fixed $k_1 = k_2 = 5$ in \autoref{fig:simLogRho} and \autoref{fig:simLogCorr} are qualitatively the same as for the mean validation RMSE. Interestingly, I also observe a downward shift in the mean validation log score for all procedures and $n$ when I increase $\rho$.
The argument used to explain the downward shift for the mean validation RMSE, namely that the variance around the conditional mean is lower the higher $\rho$, does not hold anymore since less overdispersion does not necessarily lead to lower log scores. Intuitively, less overdispersion leads to less variation around the conditional mean that could allow for a more precise estimation of the conditional mean, resulting in an improved distributional fit and, hence, a lower log score.

Finally, \autoref{fig:simLogCorr} shows the mean validation log scores varying $n$ and the correlation $b$. Similar to the RMSE, the mean validation log scores generally decrease across all $n$, when $b$ is increased. The reduction is especially large for ML-AC due to the same reasons as for the RMSE in \autoref{fig:simRMSEcorr}. The remaining procedures perform very similarly: Their mean validation log scores are similar and converge to the oracle when $n$ is large and their interquartile ranges overlap.

The relative ranking of the procedures for the Brier and spherical score is similar to that for the log score, so their plots are only shown in the \refSup.

\begin{figure}[tbh]
	\centering
	\includegraphics[width=0.87\textwidth]{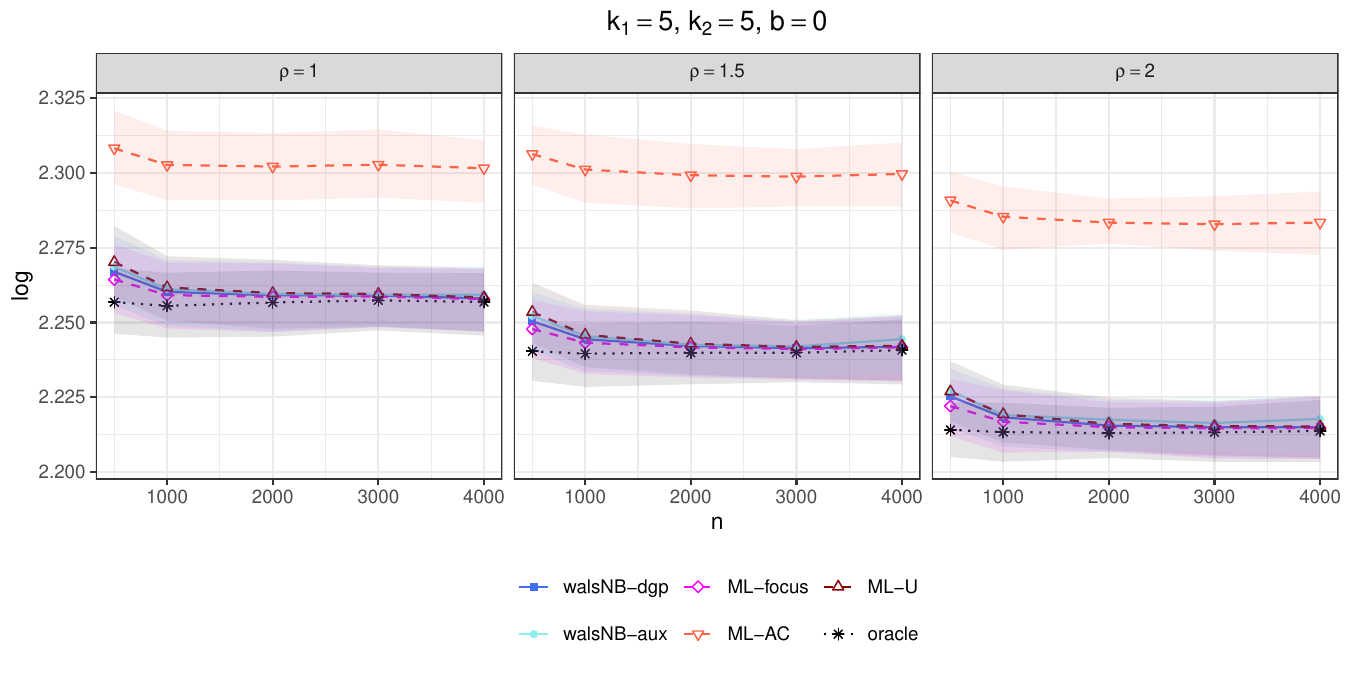}
	\caption{Mean validation log score and quartiles varying $n$ and $\rho$}\label{fig:simLogRho}
	\justifying \footnotesize \noindent The remaining parameters are fixed at $b = 0$ and $k_1 = k_2 = 5$. The shaded areas show the interquartile range.
\end{figure}

\begin{figure}[tbh]
	\centering
	\includegraphics[width=0.87\textwidth]{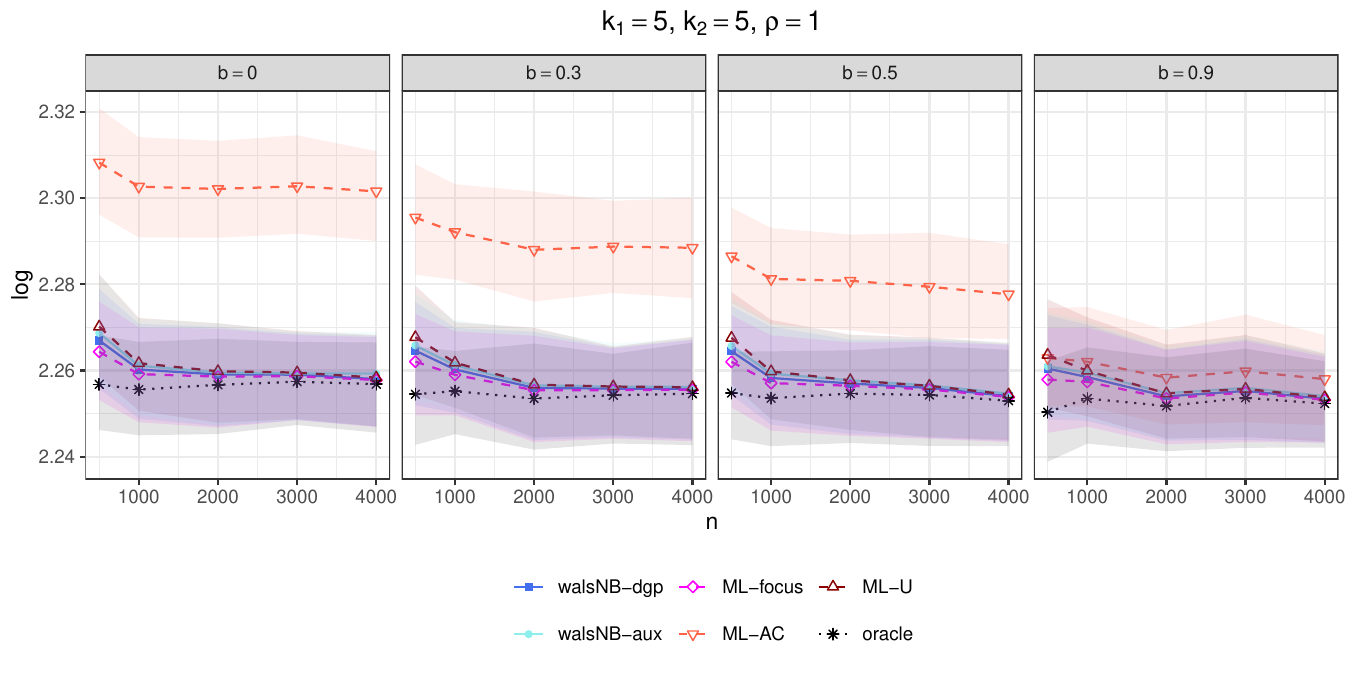}
	\caption{Mean validation log score and quartiles varying $n$ and $b$}\label{fig:simLogCorr}
	\justifying \footnotesize \noindent The remaining parameters are fixed at $\rho = 1$ and $k_1 = k_2 = 5$. The shaded areas show the interquartile range.
\end{figure}

\clearpage

\section{Empirical illustration}\label{sec:cvexpvar}

The aim of the empirical illustration is to compare the predictive performance of WALS NB with ML and lasso estimation of the NB2 regression model on real data, and to check whether the observations from the simulation experiment translate to a real-world application.

I use the cross-sectional data set called `DoctorVisits', which derives from the 1977-1978 Australian Health Survey and was analyzed in \citet{cameron1986doctor} and \citet{mullahy1997hetero}. The dataset contains $n = 5190$ observations from individuals over 18 years of age on twelve variables, including the response \texttt{visits}, which describes the number of doctor visits in a two-week period before the interview. It further provides explanatory variables such as income and age, as well as health-related variables like recent illnesses and health insurance coverage. The data are available via the \proglang{R} package \pkg{AER} \citep{kleiber2008aer} as \texttt{DoctorVisits} based on the original from the Journal of Applied Econometrics Data Archive.\footnote{\url{https://www.journaldata.zbw.eu/dataset/heterogeneity-excess-zeros-and-the-structure-of-count-data-models}}

\begin{figure}[tbh]
	\centering
	\includegraphics[width=0.75\textwidth]{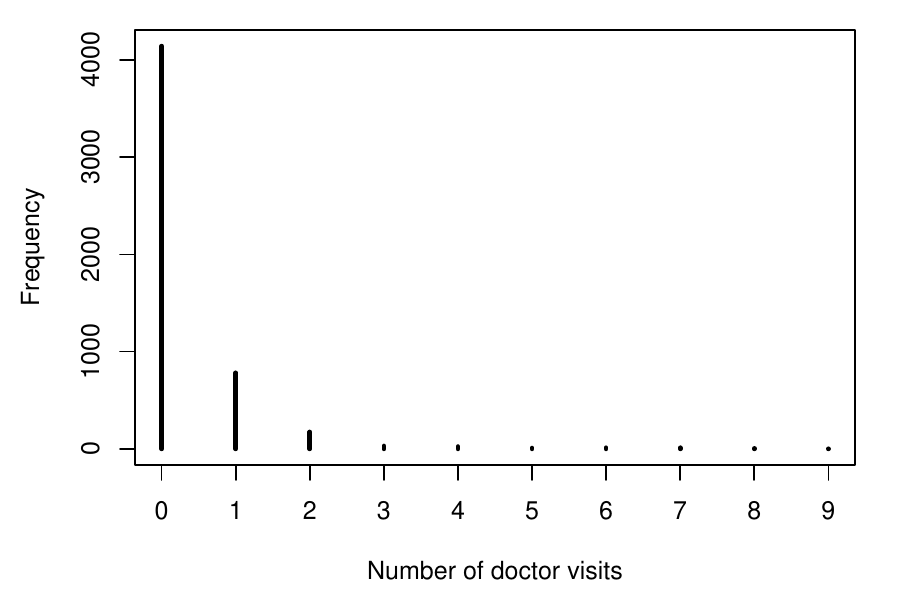}
	\caption{Visualization of \texttt{visits} in DoctorVisits}\label{fig:histDV}
\end{figure}

\autoref{tab:descDV} and \ref{tab:sumstatsDV} in \Cref{sec:appcvexpvar} provide a description and summary statistics of the variables in the DoctorVisits dataset. Further, \autoref{fig:histDV} shows a visualization of the response \texttt{visits}, which clearly exhibits overdispersion and will be modeled using regression models for count data. For the computation of the Brier and spherical score, I truncate the infinite sum in $||\hat{p}||$ from \eqref{eq:infsumsq} at the largest observed count of the dataset, which is 9.

Inspired by the applications of \citet[p.~2~f.]{rupp2012mlatom} and \citet[p.~164~f.]{faber2020qml} in quantum chemistry, I apply $K$-fold cross-validation (CV) to produce `learning curves' that allow me to compare the performance of the procedures for different sizes of the training set. \autoref{alg:cvlearn} illustrates the process for generating a $K$-fold cross-validated learning curve for any evaluation metric and procedure.

\begin{algorithm}[tbh]
	\caption{$K$-fold cross-validated learning curves}\label{alg:cvlearn}
	\begin{enumerate}
		\item Randomly split dataset $\mathcal{D} := \{ (y_{i}, x_{i}) \}_{i = 1, 2, \dotsc, n}$ into $K$ parts $\mathcal{D}_1, \mathcal{D}_2, \dotsc, \mathcal{D}_K$ of roughly the same size (see implementation of \texttt{cv()} of \pkg{mboost} \citep{mboost2014} for more details, size of last partition will be smaller if $n/K$ is not an integer). Then, the training and validation set $\mathcal{T}_k$ and $\mathcal{V}_k$ for each fold $k = 1, 2, \dotsc, K$ are defined as
		$$
			\mathcal{T}_k := \{ \mathcal{D}_{j} : j \neq k \}, \quad \mathcal{V}_{k} := \mathcal{D}_k.
		$$
		Further, let $\tau_{k}: \{1, 2, \dotsc, |\mathcal{T}_k| \} \rightarrow \{1, 2, \dotsc, n\}$ be an indexing function that maps the index of the observations of $\mathcal{T}_k$ to the original dataset $\mathcal{D}$.

		\item Specify the grid for the number of training observations $t = (t_1, t_2, \dotsc, t_L)$ with $t_L \leq t_{max}$, where $t_{max} = |\mathcal{D}| - \max_{k} | \mathcal{D}_k |$ is the largest possible size of the training set.

		\item For $l = 1, 2, \dotsc, L$:
		\begin{enumerate}
			\item For procedure $m = 1, 2, \dotsc, M$:
			\begin{enumerate}
				\item For $k = 1, 2, \dotsc, K$:
				\begin{enumerate}
					\item Fit and tune procedure $m$ on data $\mathcal{T}_{l,k} := \{(y_{\tau_{k}(r)}, x_{\tau_{k}(r)})\}_{r = 1, 2, \dotsc, t_{l}}$.
					\item Compute validation metric $\hat{s}_{l,m,k}$ on $\mathcal{V}_k$.
				\end{enumerate}
				\item Output the cross-validated metric for training size $t_l$: $\hat{s}_{l,m} = \frac{1}{K} \sum_{k = 1}^{K} \hat{s}_{l,m,k}$.
			\end{enumerate}
		\end{enumerate}
		\item The learning curve for each $m = 1, 2, \dotsc, M$ plots $\hat{s}_{l,m}$ against $t_{l}$ for $l = 1, 2, \dotsc, L$.
	\end{enumerate}
\end{algorithm}
Note that only the training sets $\mathcal{T}_{l,k}$ vary in size but the validation sets $\mathcal{V}_{k}$ remain the same. Similar to \citet[p.~398]{meek2002learncurve}, the training sets $\mathcal{T}_{l, k}$ are nested, i.e.\ $\mathcal{T}_{l, k} \subset \mathcal{T}_{l + 1, k}$ for $l = 1,2 \dotsc, L-1$. For all experiments below, I use $K = 10$ folds.

I compare procedures that differ in the estimator and specification of the mean, where the choices for the latter are inspired by the applications in \citet[p.~46~ff.]{cameron1986doctor}. The different combinations of estimator and specification are named following the pattern: `estimator-specification'. Again, `walsNB' and `ML' represent the WALS NB and ML estimator, respectively, while `lasso' estimates the NB2 regression model using the lasso estimator of \citet{wang2016pencount} (see \Cref{sec:software} for details). I consider a total of six estimator-specification combinations:
\begin{enumerate}
	\item walsNB-main: Includes all covariates linearly as auxiliary regressors and only a constant as focus regressor.
	\item walsNB-main-focus: Includes all covariates and a constant linearly as focus regressors, a quadratic term for age and two-way interactions between health and gender, health and age, health and income, and finally gender and illness as auxiliary regressors.
	\item walsNB-int: Includes all regressors of walsNB-main-focus (including interactions) as auxiliary and only a constant as focus regressor.
	\item ML-main: Includes all covariates linearly and a constant and hence uses the same regressors as walsNB-main.
	\item ML-int: Includes all regressors of ML-main but adds a quadratic term for age and two-way interactions between health and gender, health and age, health and income, and finally gender and illness. Counterpart of walsNB-main-focus and walsNB-int.
	\item lasso-int: Includes all covariates linearly, a quadratic term for age and two-way interactions between health and gender, health and age, health and income, and finally gender and illness \textit{in the fitting process}. Depending on the choice of the regularization parameter, not all the aforementioned regressors have to be included in the final model. In contrast, a constant is always included.
\end{enumerate}

All procedures are fitted using a log-link for the mean parameter. WALS NB procedures use the Laplace prior because the Weibull led to numerical instabilities in some small subsamples resulting from the numerical integration procedure required for computing the posterior mean of the auxiliary regression coefficients in \eqref{eq:gammaest}. The parameters of the Laplace prior are taken from \citet[p.~132]{magnus2016wals}, which are minimax regret solutions for the normal location problem. The remaining settings for WALS NB and ML specifications are retained from the simulation experiment of \Cref{sec:sim}.
Notably, all WALS NB specifications use the unrestricted ML estimator for NB2 as starting values for the regression coefficients and the dispersion parameter. By unrestricted, I refer to the unrestricted model given the covariates that are included in the specification, i.e.\ ML-main for walsNB-main and ML-int for, both, walsNB-main-focus and walsNB-int.

The lasso specification `lasso-int' performs tuning (maximizing 10-fold CV log-likelihood) in the training set $\mathcal{T}_{l,k}$ of each fold $k$ (and each training set size $t_l$), as recommended by \citet[p.~679]{hothorn2005design} who include tuning and final model fit in the estimation process. This is sensible, as tuning of the regularization parameter is key to the performance of lasso. Different values of the regularization parameter correspond to different levels of regularization and the regressors included in the model may also differ.
Moreover, the method also standardizes the regressors in the training set of each fold to have zero mean and unit variance before estimation (i.e., it uses the estimated means and variances of the regressors in the subsample and not over the entire dataset).

In \autoref{fig:DVlearnrmse} and \autoref{tab:DVlearnrmse} we observe that all WALS NB specifications except for walsNB-main-focus outperform the ML specifications in terms of RMSE for all numbers of training observations. The differences are particularly large for small training sets, e.g.\ for $t_{l} = 500$ the CV RMSE of walsNB-int is almost 19\% smaller than ML-int. Except for $t_{l} < 1500$, walsNB-int and walsNB-main also outperform the lasso specification lasso-int. Further, note that walsNB-int outperforms walsNB-main-focus although the only difference between the two procedures is that the latter specifies some of the covariates as focus regressors.
This observation seems to contradict the results of the simulation experiment, where walsNB-dgp and walsNB-aux perform very similarly even though the latter considers all covariates as auxiliary regressors and the former considers part of them as focus regressors. However, walsNB-dgp chooses the same focus regressors as the DGP of the simulation, which is unlikely in empirical applications.

\begin{table}[tbh]
	\caption{10-fold CV RMSE varying $t_l$, DoctorVisits}\label{tab:DVlearnrmse}
	\begin{center}
		\begin{threeparttable}
		\begin{footnotesize}
			
\begin{tabular}{@{\extracolsep{5pt}} lccccccccc} 
\\[-1.8ex]\hline 
\hline \\[-1.8ex] 
 Training obs. $t_{l}$ & 500 & 1000 & 1500 & 2000 & 2500 & 3000 & 3500 & 4000 & 4671 \\ 
\hline \\[-1.8ex] 
walsNB-main & $0.838$ & $0.886$ & $0.807$ & $0.756$ & $0.747$ & $0.746$ & $0.750$ & $0.756$ & $0.760$ \\ 
walsNB-main-focus & $1.012$ & $1.068$ & $0.893$ & $0.789$ & $0.768$ & $0.765$ & $0.767$ & $0.774$ & $0.780$ \\ 
walsNB-int & $0.837$ & $0.864$ & $0.797$ & $0.753$ & $0.745$ & $0.742$ & $0.743$ & $0.746$ & $0.752$ \\ 
ML-main & $1.013$ & $0.950$ & $0.873$ & $0.794$ & $0.771$ & $0.766$ & $0.774$ & $0.781$ & $0.785$ \\ 
ML-int & $1.032$ & $0.956$ & $0.885$ & $0.792$ & $0.769$ & $0.762$ & $0.764$ & $0.770$ & $0.776$ \\ 
lasso-int & $0.791$ & $0.826$ & $0.831$ & $0.772$ & $0.761$ & $0.753$ & $0.759$ & $0.764$ & $0.776$ \\ 
\hline \\[-1.8ex] 
\end{tabular} 

			\end{footnotesize}
		\begin{tablenotes}
			\footnotesize
			\item[--] All figures rounded to three decimal places.
		\end{tablenotes}
		\end{threeparttable}
	\end{center}
	\end{table}

\begin{figure}[tbh]
	\centering
	\includegraphics[width=0.81\textwidth]{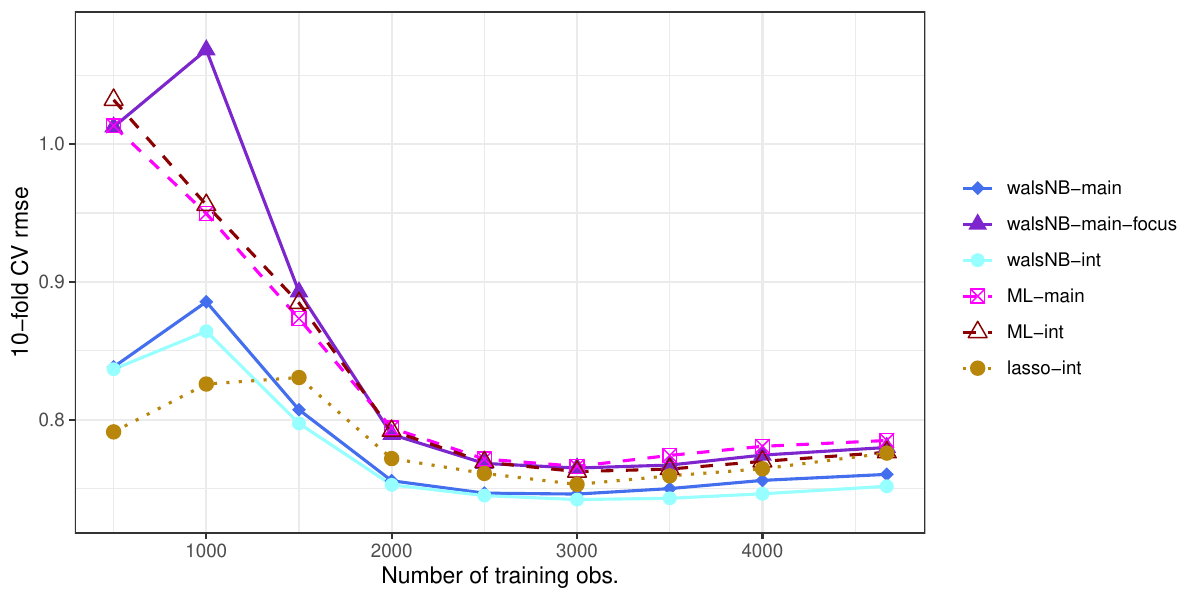}
	\caption{10-fold CV RMSE varying $t_l$, DoctorVisits}\label{fig:DVlearnrmse}
\end{figure}

The CV log scores in \autoref{fig:DVlearnlog} and \autoref{tab:DVlearnlog} show that all procedures perform similarly in terms of the distributional fit. Moreover, the curves decrease as I increase the number of training observations and flatten at about 2000 observations. This shows that the methods are able to `learn' more (i.e.\ improve the fit in terms of log score), when more training observations are available but stop `learning' at some point, i.e.\ when the curves flatten.

The other metrics for distributional fit, Brier and spherical score, show qualitatively similar results but the curves are flatter, hence the results are only shown in the \refSup. This further underlines that the distributional fit of the methods does not improve drastically when the dataset becomes larger. 

Note that WALS NB specifications are computationally less demanding than lasso, while performing similarly in terms of CV RMSE and log score. They do not require any tuning unlike lasso, which performs an `internal' 10-fold CV to choose the optimal regularization parameter. Consequently, the fitting times of WALS NB are typically shorter than those of lasso and competitive with the ML specifications. Of course, one may change the parameters of the fitting algorithm of lasso to improve the computing time. However, it should not result in better performance metrics as the current setup already favors lasso: It allows many iterations in the fitting algorithm and, thus, a high chance for convergence.

\begin{table}[tbh]
	\caption{10-fold CV log score varying $t_l$, DoctorVisits}\label{tab:DVlearnlog}
	\begin{center}
		\begin{threeparttable}
		\begin{footnotesize}
			
\begin{tabular}{@{\extracolsep{5pt}} lccccccccc} 
\\[-1.8ex]\hline 
\hline \\[-1.8ex] 
 Training obs. $t_{l}$ & 500 & 1000 & 1500 & 2000 & 2500 & 3000 & 3500 & 4000 & 4671 \\ 
\hline \\[-1.8ex] 
walsNB-main & $0.639$ & $0.628$ & $0.623$ & $0.622$ & $0.622$ & $0.622$ & $0.621$ & $0.621$ & $0.620$ \\ 
walsNB-main-focus & $0.642$ & $0.629$ & $0.624$ & $0.622$ & $0.621$ & $0.621$ & $0.621$ & $0.620$ & $0.620$ \\ 
walsNB-int & $0.642$ & $0.629$ & $0.626$ & $0.625$ & $0.623$ & $0.623$ & $0.623$ & $0.622$ & $0.621$ \\ 
ML-main & $0.640$ & $0.627$ & $0.622$ & $0.621$ & $0.620$ & $0.621$ & $0.620$ & $0.620$ & $0.620$ \\ 
ML-int & $0.644$ & $0.630$ & $0.626$ & $0.624$ & $0.622$ & $0.622$ & $0.622$ & $0.621$ & $0.621$ \\ 
lasso-int & $0.637$ & $0.626$ & $0.623$ & $0.622$ & $0.622$ & $0.622$ & $0.621$ & $0.620$ & $0.620$ \\ 
\hline \\[-1.8ex] 
\end{tabular} 

			\end{footnotesize}
		\begin{tablenotes}
			\footnotesize
			\item[--] All figures rounded to three decimal places.
		\end{tablenotes}
		\end{threeparttable}
	\end{center}
\end{table}

\begin{figure}[tbh]
	\centering
	\includegraphics[width=0.81\textwidth]{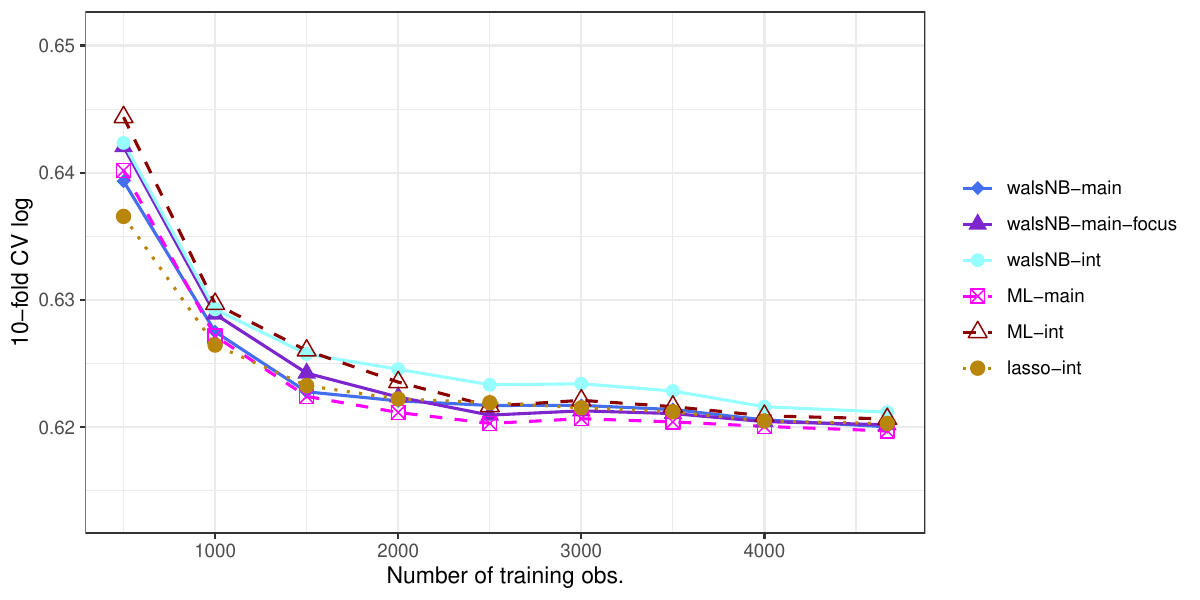}
	\caption{10-fold CV log score varying $t_l$, DoctorVisits}\label{fig:DVlearnlog}
\end{figure}

In conclusion, all WALS NB specifications, except for walsNB-main-focus, perform better than the ML specifications in terms of RMSE, while metrics for the distributional fit such as log, Brier and spherical score are similar or minimally worse than the ML specifications. Moreover, the RMSE is similar or slightly lower than for lasso at large $t_{l}$, while demanding less computational resources as WALS NB does not require tuning by CV.

\section{Conclusion}

This paper extends the WALS approach to NB2 regression models (WALS NB) for count data based on WALS GLM of \citet{deluca2018glm} and compares its predictive performance to the traditional ML and lasso estimator in simulated and real count datasets using the classical measure RMSE and strictly proper scoring rules.

In the simulation experiment, WALS NB outperforms the ML estimator in very sparse situations, i.e.\ where the number of auxiliary regressors is large and the number of training observations is small. When increasing the number of training observations, WALS NB and the unrestricted ML estimator converge in all performance metrics. Interestingly, whether WALS NB includes all regressors as auxiliary regressors or parts of them as focus does not change the results substantially. This shows that specifying all regressors as auxiliary is a reasonable choice if no prior information is available on the importance of the individual regressors. Moreover, it highlights that the regularized Bayesian estimation of the coefficients of the auxiliary regressors is key for the performance of WALS NB.

The empirical illustration emphasizes the results found in the simulation experiment: For small training sets, WALS NB using all covariates as auxiliary regressors outperforms all ML specifications in terms of RMSE while yielding a comparable distributional fit measured by strictly proper scores.
Only the lasso estimator yields lower RMSE for small training sets. However, it is more computationally demanding than WALS NB due to the additional 10-fold CV that is run for determining the optimal regularization parameter. This highlights an important advantage of WALS compared to other model averaging techniques: low computational costs. Moreover, WALS NB using all covariates as auxiliary regressors outperformed all other specifications of WALS NB. Thus, if only the predictive power is of concern, WALS NB is a viable alternative to established estimation methods for the NB2 regression model that is easy to specify (choose all regressors as auxiliary), regularized and computationally efficient.

For future research, it would be interesting to generalize WALS to hurdle or zero-inflation models to handle count data with excess zeros. Thus far, WALS has been limited to univariate response variables, therefore extending the methodology to multivariate outcomes would allow a larger variety of applications, such as joint modeling of related count processes. Lastly, an investigation of the large sample properties of WALS could improve our understanding of statistical inference after model averaging.

\section*{Acknowledgements}
\addcontentsline{toc}{section}{Acknowledgements}

The scientific computing center sciCORE (\url{https://scicore.unibas.ch/}) at the University of Basel provided me with valuable computing resources. I would also like to thank Christian Kleiber for our helpful discussions.

\stoplist[main]{lof}
\stoplist[main]{lot}


\startlist[appendix]{lof}
\startlist[appendix]{lot}
\clearpage
\addappheadtotoc

\section*{Appendix}\label{sec:app}
\appendix


\renewcommand{\thesubsection}{\Alph{subsection}}
\renewcommand\thefigure{\Alph{subsection}.\arabic{figure}}
\renewcommand\thetable{\Alph{subsection}.\arabic{table}}

\renewcommand{\theequation}{\Alph{subsection}.\arabic{equation}}
\renewcommand{\thetheorem}{\Alph{subsection}.\arabic{theorem}}
\renewcommand{\theass}{\Alph{subsection}.\arabic{ass}}
\counterwithin{equation}{subsection} 
\counterwithin{theorem}{subsection} 
\counterwithin{ass}{subsection} 

\setcounter{figure}{0}
\setcounter{table}{0}
\pagestyle{appendix}

\subsection{Proofs}\label{sec:proofs}

\begin{proof}[\textbf{Proof of \cref{prop:onestepML}}]

I start by rewriting the equation system from \eqref{eq:taylor} using the data transformations in \eqref{eq:trafos}. Notice that
\begin{align*}
	\bX_{p}^{\top} \bV (y - \bmu) &= \sumin \bv_i (y_i - \bmu_i)x_{ip}, \\
	\bX_{p}^{\top} \bX_{q} &= \sumin \bpsi_i x_{ip} x_{iq}^{\top},
\end{align*}
for  $p, q = 1, 2$, so the first equation of \eqref{eq:taylor} can be expressed as
\begin{align*}
	0 ={}& \bX_{1}^{\top} \Big(\bX_1 \bbeta_1 + \bX_2 \bbeta_2 + \bPsi^{-1/2} \bV(y - \bmu) - \bar{g} \bPsi^{-1/2} \bC (y - \bmu) \balpha - \bX_1 \beta_1 - \bX_2 \beta_2 \\
	&+ \bar{g}\bPsi^{-1/2}\bC(y - \bmu) \alpha \Big).
\end{align*}
Using $\bar{y}_0$ from \eqref{eq:trafos}, the expression can be written more compactly as
\begin{equation}\label{eq:taylorx1}
	0 = \bX_{1}^{\top} \left(\bar{y}_0 - \bX_1 \beta_1 - \bX_2 \beta_2 + \bar{g} \bPsi^{-1/2} \bC (y - \bmu) \alpha \right).
\end{equation}
Following the same steps, the second equation of \eqref{eq:taylor} becomes
\begin{equation}\label{eq:taylorx2}
	0 = \bX_{2}^{\top} \left(\bar{y}_0 - \bX_1 \beta_1 - \bX_2 \beta_2 + \bar{g} \bPsi^{-1/2} \bC (y - \bmu) \alpha \right) - R_{j} \nu_{j}.
\end{equation}
Analogously, the third equation in \eqref{eq:taylor} can be expressed as
\begin{equation}\label{eq:tayloralpha0}
\begin{aligned}
	0 &= \bar{g} \xmat{\bkappa}{\ones}  - \bar{g} \xmat{(y - \bmu)}{\bC} \bar{\eta} + \bar{g} \xmat{(y - \bmu)}{\bC} (X_1 \beta_1 + X_2 \beta_2) + (\bar{g}^2 \xmat{\bk}{\ones} + \bar{\varrho} \xmat{\bkappa}{\ones}) (\alpha - \balpha), \\
	  &= \bar{t} + \bar{g} \xmat{(y - \bmu)}{\bC} (X_1 \beta_1 + X_2 \beta_2) + (\bar{g}^2 \xmat{\bk}{\ones} + \bar{\varrho} \xmat{\bkappa}{\ones}) \alpha. 
\end{aligned}
\end{equation}
Assuming $\bar{g}^2 \xmat{\bk}{\ones} + \bar{\varrho} \xmat{\bkappa}{\ones} \neq 0$,
I can solve \eqref{eq:tayloralpha0} for $\alpha$:
\begin{align}\label{eq:tayloralpha1}
	\alpha = -\frac{\bar{t} + \bar{g} \xmat{(y - \bmu)}{\bC} (X_1 \beta_1 + X_2 \beta_2)}{\bar{g}^2 \xmat{\bk}{\ones} + \bar{\varrho} \xmat{\bkappa}{\ones}}.
\end{align}

Let us combine \eqref{eq:taylorx1} and \eqref{eq:taylorx2} into a larger matrix equation. First, move some terms so they become
\begin{align*}
	\xmat{\bX_1}{\bX_1}\beta_1 + \xmat{\bX_1}{\bX_2}\beta_2 &= \xmat{\bX_1}{\bar{y}_0} + \xmat{\bX_1}{\bPsi^{-1/2}}\bC (y - \bmu) \bar{g} \alpha, \\
	\xmat{\bX_2}{\bX_1}\beta_1 + \xmat{\bX_2}{\bX_2}\beta_2 &= \xmat{\bX_2}{\bar{y}_0} + \xmat{\bX_2}{\bPsi^{-1/2}}\bC (y - \bmu) \bar{g} \alpha  - R_j \nu_{j}.
\end{align*}
Using $\bPsi^{-1/2} \bX_p = X_p, p = 1, 2$, collect both equations to
\begin{equation*}
	\begin{pmatrix}
		\xmat{\bX_{1}}{\bX_{1}} & \xmat{\bX_{1}}{\bX_{2}} \\
		\xmat{\bX_{2}}{\bX_{1}} & \xmat{\bX_{2}}{\bX_{2}} \\		
	\end{pmatrix} \begin{pmatrix}
		\beta_{1} \\
		\beta_{2}
	\end{pmatrix}
	=
	\begin{pmatrix}
		\xmat{\bX_1}{\bar{y}_0} \\
		\xmat{\bX_2}{\bar{y}_0}
	\end{pmatrix} + \begin{pmatrix}
		\xmat{X_1}{\bC}(y - \bmu) \\
		\xmat{X_2}{\bC}(y - \bmu) \\		
	\end{pmatrix} \bar{g} \alpha - \begin{pmatrix}
		0 \\
		R_j
	\end{pmatrix} \nu_{j}.
\end{equation*}	
Inserting \eqref{eq:tayloralpha1} and rearranging yields
\begin{equation*}\label{eq:systemBeta1.1}
\begin{aligned}
	\begin{pmatrix}
		\xmat{\bX_{1}}{\bX_{1}} & \xmat{\bX_{1}}{\bX_{2}} \\
		\xmat{\bX_{2}}{\bX_{1}} & \xmat{\bX_{2}}{\bX_{2}} \\		
	\end{pmatrix} \begin{pmatrix}
		\beta_{1} \\
		\beta_{2}
	\end{pmatrix}
	={}
	& \begin{pmatrix}
		\xmat{\bX_1}{\bar{y}_0} \\
		\xmat{\bX_2}{\bar{y}_0}
	\end{pmatrix} - \frac{\bar{t} \bar{g}}{\bar{g}^2 \xmat{\bk}{\ones} + \bar{\varrho} \xmat{\bkappa}{\ones}}\begin{pmatrix}
		\xmat{X_1}{\bC}(y - \bmu) \\
		\xmat{X_2}{\bC}(y - \bmu) \\		
	\end{pmatrix}  \\
	& - \frac{\bar{g}^2}{\bar{g}^2 \xmat{\bk}{\ones} + \bar{\varrho} \xmat{\bkappa}{\ones}}\begin{pmatrix}
		\xmat{X_1}{\bC}(y - \bmu) (y - \bmu)^{\top} \bC X_1 \beta_1 \\
		\xmat{X_2}{\bC}(y - \bmu) (y - \bmu)^{\top} \bC X_1 \beta_1 \\		
	\end{pmatrix}  \\
	& - \frac{\bar{g}^2}{\bar{g}^2 \xmat{\bk}{\ones} + \bar{\varrho} \xmat{\bkappa}{\ones}}\begin{pmatrix}
		\xmat{X_1}{\bC}(y - \bmu) (y - \bmu)^{\top} \bC X_2 \beta_2 \\
		\xmat{X_2}{\bC}(y - \bmu) (y - \bmu)^{\top} \bC X_2 \beta_2 \\		
	\end{pmatrix} \\
	& - \begin{pmatrix}
		0 \\
		R_j
	\end{pmatrix} \nu_{j},
\end{aligned}
\end{equation*}
which can be further rewritten, using $\bepsilon$ and $\bq$ as defined in \Cref{sec:onestepNB}, to
\begin{equation}\label{eq:system1}
\begin{aligned}
	 \bar{A} \begin{pmatrix}
		\beta_{1} \\
		\beta_{2}
	\end{pmatrix} &= \begin{pmatrix}
		\xmat{\bX_1}{\bar{y}_0} \\
		\xmat{\bX_2}{\bar{y}_0}
	\end{pmatrix} - \bar{t} \bepsilon \begin{pmatrix}
		\xmat{X_1}{\bq} \\
		\xmat{X_2}{\bq} \\	
		\end{pmatrix} - \begin{pmatrix}
		0 \\
		R_j
	\end{pmatrix} \nu_{j},
\end{aligned}
\end{equation}
with
\[
\bar{A} := \begin{pmatrix}
	\xmat{\bX_{1}}{\bX_{1}} + \bar{g} \bepsilon \xmat{X_1}{\bq} \xmat{\bq}{X_1} 
	& \xmat{\bX_{1}}{\bX_{2}} + \bar{g} \bepsilon \xmat{X_1}{\bq} \xmat{\bq}{X_2} \\
	\xmat{\bX_{2}}{\bX_{1}} + \bar{g} \bepsilon \xmat{X_2}{\bq} \xmat{\bq}{X_1} 
	& \xmat{\bX_{2}}{\bX_{2}} + \bar{g} \bepsilon \xmat{X_2}{\bq} \xmat{\bq}{X_2}\\		
\end{pmatrix}.
\]
Then, consider the partitioned inverse
\[
\bar{A}^{-1} = \begin{pmatrix}
	\bar{A}^{11} & \bar{A}^{12} \\
	\bar{A}^{21} & \bar{A}^{22}
\end{pmatrix},
\]
with elements
\begin{align}
	\bar{A}^{11} ={}& (\xmat{\bX_1}{\bX_1} + \bar{g} \bepsilon \xmat{X_1}{\bq} \xmat{\bq}{X_1})^{-1} \nonumber \\ 
	& + \bigg[ (\xmat{\bX_1}{\bX_1} + \bar{g} \bepsilon \xmat{X_1}{\bq} \xmat{\bq}{X_1})^{-1} (\xmat{\bX_1}{\bX_2} + \bar{g} \bepsilon \xmat{X_1}{\bq} \xmat{\bq}{X_2}) (\bX_{2}^{\top} \bM_1 \bX_{2})^{-1} \nonumber \\
	&  \cdot (\xmat{\bX_2}{\bX_1} + \bar{g} \bepsilon \xmat{X_2}{\bq} \xmat{\bq}{X_1}) (\xmat{\bX_1}{\bX_1} + \bar{g} \bepsilon \xmat{X_1}{\bq} \xmat{\bq}{X_1})^{-1} \bigg], \label{eq:invA11} \\
	\bar{A}^{12} ={}& -(\xmat{\bX_1}{\bX_1} + \bar{g} \bepsilon \xmat{X_1}{\bq} \xmat{\bq}{X_1})^{-1} (\xmat{\bX_1}{\bX_2} + \bar{g} \bepsilon \xmat{X_1}{\bq} \xmat{\bq}{X_2}) (\bX_{2}^{\top} \bM_1 \bX_2)^{-1} = \bar{A}^{21}{}^{\top}, \label{eq:invA12} \\
	\bar{A}^{22} ={}& (\bX_{2}^{\top} \bM_1 \bX_2)^{-1}. \label{eq:invA22}
\end{align}
It is assumed that $\bX_{2}^{\top} \bM_1 \bX_2$ is positive definite so all elements of the partitioned inverse of $\bar{A}^{-1}$ exist.

Using the Sherman-Morrison-Woodbury formula, I can rewrite the following inverse if $(1 + \bar{g} \bepsilon \xmat{\bq}{X_1} (\xmat{\bX_1}{\bX_1})^{-1} \xmat{X_1}{\bq}) \neq 0$ as
\begin{align}\label{eq:smwX1X1}
	 (\xmat{\bX_1}{\bX_1} + \bar{g} \bepsilon \xmat{X_1}{\bq} \xmat{\bq}{X_1})^{-1} = (\xmat{\bX_1}{\bX_1})^{-1} - \frac{(\xmat{\bX_1}{\bX_1})^{-1} (\bar{g} \bepsilon \xmat{X_1}{\bq} \xmat{\bq}{X_1}) (\xmat{\bX_1}{\bX_1})^{-1}}{1 + \bar{g} \bepsilon \xmat{\bq}{X_1} (\xmat{\bX_1}{\bX_1})^{-1} \xmat{X_1}{\bq}}.
\end{align}
Therefore, $(\xmat{\bX_1}{\bX_1} + \bar{g} \bepsilon \xmat{X_1}{\bq} \xmat{\bq}{X_1})^{-1}$ exists if $(1 + \bar{g} \bepsilon \xmat{\bq}{X_1} (\xmat{\bX_1}{\bX_1})^{-1} \xmat{X_1}{\bq}) \neq 0$ and $\xmat{\bX_1}{\bX_1}$ is invertible. The latter is easily shown because $\rank(\bX_1) = \rank(\bPsi^{-1/2} X_1) = \rank(X_1) = k_{1}$ (assumed $X_1$ to have full column rank) and  $\rank(\bPsi^{-1/2}) = n$, otherwise I would not be able to compute $\bPsi^{-1/2}$. Thus, $\bX_1$ has full column rank and $\rank(\bX_{1}^{\top} \bX_1) = k_1$. 

Let $\tbeta_{1u}$ and $\tbeta_{2u}$ denote the solution of the unrestricted model, then plugging $R_{u} = 0$ into \eqref{eq:system1} yields the unrestricted equation system
\begin{align}
	\bar{A} \begin{pmatrix}
		\tbeta_{1u} \\
		\tbeta_{2u}
	\end{pmatrix} = \begin{pmatrix}
		\xmat{\bX_1}{\bar{y}_0} \\
		\xmat{\bX_2}{\bar{y}_0}
	\end{pmatrix} - \bar{t} \bepsilon \begin{pmatrix}
		\xmat{X_1}{\bq} \\
		\xmat{X_2}{\bq} \\	
		\end{pmatrix}. \label{eq:systemu}
\end{align}
Further, let $\tbeta_{1j}$ and $\tbeta_{2j}$ denote the solution of the $j$th model. Using \eqref{eq:system1} then yields
\begin{align}
	\bar{A} \begin{pmatrix}
		\tbeta_{1j} \\
		\tbeta_{2j}
	\end{pmatrix} &= \begin{pmatrix}
		\xmat{\bX_1}{\bar{y}_0} \\
		\xmat{\bX_2}{\bar{y}_0}
	\end{pmatrix} - \bar{t} \bepsilon \begin{pmatrix}
		\xmat{X_1}{\bq} \\
		\xmat{X_2}{\bq} \\	
		\end{pmatrix} - \begin{pmatrix}
		0 \\
		R_j
	\end{pmatrix} \tilde{\nu}_{j}. \label{eq:systemj}
\end{align}

Combining \eqref{eq:systemu} and \eqref{eq:systemj}, I can find an explicit solution for $\nu_j$ as they imply
\begin{align*}
	\bar{A} \begin{pmatrix}
		\tbeta_{1j} \\
		\tbeta_{2j}
	\end{pmatrix} = \bar{A} \begin{pmatrix}
		\tbeta_{1u} \\
		\tbeta_{2u}
	\end{pmatrix} - \begin{pmatrix}
		0 \\
		R_j
	\end{pmatrix} \tilde{\nu}_{j},
\end{align*}
then multiply with $\bar{A}^{-1}$ so
\begin{align}\label{eq:systemj2}
	\begin{pmatrix}
		\tbeta_{1j} \\
		\tbeta_{2j}
	\end{pmatrix} = \begin{pmatrix}
		\tbeta_{1u} \\
		\tbeta_{2u}
	\end{pmatrix} - \begin{pmatrix}
		\bar{A}^{11} & \bar{A}^{12} \\
		\bar{A}^{21} & \bar{A}^{22}
	\end{pmatrix} \begin{pmatrix}
		0 \\
		R_j
	\end{pmatrix} \tilde{\nu}_{j}.
\end{align}
Multiply both sides by $\begin{pmatrix}
0 & R_{j}^{\top}
\end{pmatrix}$ and note that by \eqref{eq:optNB} $R_{j}^{\top} \tbeta_{2j} = 0$, then
\begin{align}
	\begin{pmatrix}
		0 & R_{j}^{\top}
	\end{pmatrix} \begin{pmatrix}
		\tbeta_{1j} \\
		\tbeta_{2j}
	\end{pmatrix} &= \begin{pmatrix}
		0 & R_{j}^{\top}
	\end{pmatrix} \begin{pmatrix}
		\tbeta_{1u} \\
		\tbeta_{2u}
	\end{pmatrix} - \begin{pmatrix}
		0 & R_{j}^{\top}
	\end{pmatrix} \begin{pmatrix}
		\bar{A}^{11} & \bar{A}^{12} \\
		\bar{A}^{21} & \bar{A}^{22}
	\end{pmatrix} \begin{pmatrix}
		0 \\
		R_j
	\end{pmatrix} \tilde{\nu}_{j} \nonumber \\
	0 &= R_{j}^{\top} \tbeta_{2u} - R_{j}^{\top} \bar{A}^{22} R_{j} \tilde{\nu}_j \nonumber \\
	\rightarrow \tilde{\nu}_{j} &= (R_{j}^{\top} \bar{A}^{22} R_{j})^{-1} R_{j}^{\top} \tbeta_{2u}, \label{eq:nu}
\end{align} 
assuming the inverse $(R_{j}^{\top} \bar{A}^{22} R_{j})^{-1}$ exists.
Plug \eqref{eq:nu} into \eqref{eq:systemj2} for
\begin{align*}
	\tbeta_{1j} &= \tbeta_{1u} - \bar{A}^{12}R_{j} (R_{j}^{\top} \bar{A}^{22} R_{j})^{-1} R_{j}^{\top} \tbeta_{2u}, \\
	\tbeta_{2j} &= \tbeta_{2u} - \bar{A}^{22}R_{j} (R_{j}^{\top} \bar{A}^{22} R_{j})^{-1} R_{j}^{\top} \tbeta_{2u}.
\end{align*}
Now, insert \eqref{eq:invA12}, \eqref{eq:invA22} and introduce $n$ so
\begin{equation}\label{eq:b1j}
\begin{aligned}
	\tbeta_{1j} ={}& \tbeta_{1u} + \bigg[ \left(\frac{\xmat{\bX_1}{\bX_1}}{n} + \frac{\bar{g} \bepsilon}{n} \xmat{X_1}{\bq} \xmat{\bq}{X_1} \right)^{-1}  \left(\frac{\xmat{\bX_1}{\bX_2}}{n} + \frac{\bar{g} \bepsilon}{n} \xmat{X_1}{\bq} \xmat{\bq}{X_2} \right) \left(\frac{\bX_{2}^{\top} \bM_1 \bX_2}{n} \right)^{-1/2}  \\
	& \cdot \left(\frac{\bX_{2}^{\top} \bM_1 \bX_2}{n} \right)^{-1/2}  R_{j} \left( R_{j}^{\top} \left(\frac{\bX_{2}^{\top} \bM_1 \bX_2}{n} \right)^{-1} R_{j} \right)^{-1}  R_{j}^{\top} \left(\frac{\bX_{2}^{\top} \bM_1 \bX_2}{n} \right)^{-1/2} \\
	& \cdot \left(\frac{\bX_{2}^{\top} \bM_1 \bX_2}{n} \right)^{1/2} \tbeta_{2u} \bigg]  \\
	={}&  \tbeta_{1u} + \bQ \bP_{j} \tvartheta.
\end{aligned}
\end{equation}
Moreover, introduce $n$ also for $\tbeta_{2j}$ which yields
\begin{align*}
	\tbeta_{2j} ={}& \tbeta_{2u} - \bigg[ \left(\frac{\bX_{2}^{\top} \bM_1 \bX_2}{n} \right)^{-1}  R_{j} \left( R_{j}^{\top} \left(\frac{\bX_{2}^{\top} \bM_1 \bX_2}{n} \right)^{-1} R_{j} \right)^{-1}  R_{j}^{\top} \left(\frac{\bX_{2}^{\top} \bM_1 \bX_2}{n} \right)^{-1/2} \\
	& \cdot \left(\frac{\bX_{2}^{\top} \bM_1 \bX_2}{n} \right)^{1/2} \tbeta_{2u} \bigg] \\
	={}& \tbeta_{2u} - \left(\frac{\bX_{2}^{\top} \bM_1 \bX_2}{n} \right)^{-1/2}  \bP_j \left(\frac{\bX_{2}^{\top} \bM_1 \bX_2}{n} \right)^{1/2} \tbeta_{2u} \\	
	={}& \left(\frac{\bX_{2}^{\top} \bM_1 \bX_2}{n} \right)^{-1/2} (I_{k_2} - \bP_j ) \tvartheta \\
	={}&   \left(\frac{\bX_{2}^{\top} \bM_1 \bX_2}{n} \right)^{-1/2} \bW_j \tvartheta.
\end{align*}

In order to obtain the estimator of the fully restricted model $\tbeta_{1r}$, I set $R_{r} = I_{k_{2}}$. Combined with \eqref{eq:b1j} this leads to
\begin{align}
	\tbeta_{1r} &= \tbeta_{1u} + \bQ \tvartheta \label{eq:b1rb1u} \\
	\leftrightarrow  \tbeta_{1u} &= \tbeta_{1r} - \bQ \tvartheta. \label{eq:b1ub1r}
\end{align}
Insert \eqref{eq:b1ub1r} in \eqref{eq:b1j} to get
\begin{equation*}
	\tbeta_{1j} = \bbeta_{1r} - \bQ \bW_{j} \tvartheta.
\end{equation*}

What remains to be derived are $\tbeta_{1u}$ and $\tbeta_{2u}$. First, multiply \eqref{eq:systemu} with $\bar{A}^{-1}$ for
\begin{align*}
	\begin{pmatrix}
		\tbeta_{1u} \\
		\tbeta_{2u}
	\end{pmatrix} = \begin{pmatrix}
		\bar{A}^{11} & \bar{A}^{12} \\
		\bar{A}^{21} & \bar{A}^{22}
	\end{pmatrix} \begin{pmatrix}
		\xmat{\bX_1}{\bar{y}_0} \\
		\xmat{\bX_2}{\bar{y}_0}
	\end{pmatrix} - \bar{t} \bepsilon \begin{pmatrix}
		\bar{A}^{11} & \bar{A}^{12} \\
		\bar{A}^{21} & \bar{A}^{22}
	\end{pmatrix} \begin{pmatrix}
		\xmat{X_1}{\bq} \\
		\xmat{X_2}{\bq} \\	
		\end{pmatrix}.
\end{align*}
Inserting \eqref{eq:invA11}, \eqref{eq:invA12} and \eqref{eq:invA22} results in 
\begin{equation}\label{eq:b1u}
\begin{aligned}
	\tbeta_{1u} ={}& \bigg[ (\xmat{\bX_1}{\bX_1} + \bar{g} \bepsilon \xmat{X_1}{\bq} \xmat{\bq}{X_1})^{-1} \\ 
	& + \bigg\{ (\xmat{\bX_1}{\bX_1} + \bar{g} \bepsilon \xmat{X_1}{\bq} \xmat{\bq}{X_1})^{-1} (\xmat{\bX_1}{\bX_2} + \bar{g} \bepsilon \xmat{X_1}{\bq} \xmat{\bq}{X_2}) (\bX_{2}^{\top} \bM_1 \bX_{2}^{\top})^{-1} \\
	&  (\xmat{\bX_2}{\bX_1} + \bar{g} \bepsilon \xmat{X_2}{\bq} \xmat{\bq}{X_1}) (\xmat{\bX_1}{\bX_1} + \bar{g} \bepsilon \xmat{X_1}{\bq} \xmat{\bq}{X_1})^{-1} \bigg\} \bigg] (\bX_{1}^{\top} \bar{y}_0 - \bar{t} \bepsilon \xmat{X_1}{\bq}) \\
	& - (\xmat{\bX_1}{\bX_1} + \bar{g} \bepsilon \xmat{X_1}{\bq} \xmat{\bq}{X_1})^{-1} (\xmat{\bX_1}{\bX_2} + \bar{g} \bepsilon \xmat{X_1}{\bq} \xmat{\bq}{X_2}) (\bX_{2}^{\top} \bM_1 \bX_{2}^{\top})^{-1} (\xmat{\bX_2}{\bar{y}_0} - \bar{t} \bepsilon \xmat{X_2}{\bq}).
\end{aligned}
\end{equation}
and
\begin{equation}\label{eq:b2u}
\begin{aligned}
	\tbeta_{2u} ={}& -(\bX_{2}^{\top} \bM_1 \bX_2)^{-1} (\xmat{\bX_2}{\bX_1} + \bar{g} \bepsilon \xmat{X_2}{\bq} \xmat{\bq}{X_1}) (\xmat{\bX_1}{\bX_1} + \bar{g} \bepsilon \xmat{X_1}{\bq} \xmat{\bq}{X_1})^{-1} (\xmat{\bX_1}{\bar{y}_0} - \bar{t} \bepsilon \xmat{X_1}{\bar{q}}) \\
	& + (\bX_{2}^{\top} \bM_1 \bX_2)^{-1} (\xmat{\bX_2}{\bar{y}_0} - \bar{t} \bepsilon \xmat{X_2}{\bar{q}}). 
\end{aligned}
\end{equation}
Plugging $\tbeta_{1u}$ and $\tbeta_{2u}$ into \eqref{eq:b1rb1u} and exploiting \eqref{eq:vartheta} yields
\begin{equation}
\begin{aligned}
	\tbeta_{1r} &= (\xmat{\bX_1}{\bX_1} + \bar{g} \bepsilon \xmat{X_1}{\bq}\xmat{\bq}{X_1})^{-1} (\xmat{\bX_1}{\bar{y}_0} - \bar{t} \bepsilon \xmat{X_1}{\bq}) \\
	&= \left( \frac{\xmat{\bX_1}{\bX_1}}{n} + \frac{\bar{g} \bepsilon}{n} \xmat{X_1}{\bq}\xmat{\bq}{X_1} \right)^{-1} \left( \frac{\xmat{\bX_1}{\bar{y}_0}}{n} - \frac{\bar{t} \bepsilon}{n} \xmat{X_1}{\bq} \right).
\end{aligned}
\end{equation}
Finally, insert $\tbeta_{1j}$ and $\tbeta_{2j}$ in \eqref{eq:tayloralpha1} for the estimator of the dispersion parameter
\begin{align*}
	\talpha_{j} = -\frac{\bar{t} + \bar{g} \xmat{(y - \bmu)}{\bC} (X_1 \tbeta_{1j} + X_2 \tbeta_{2j})}{\bar{g}^2 \xmat{\bk}{\ones} + \bar{\varrho} \xmat{\bkappa}{\ones}}.
\end{align*}
The proposition collects the solutions for $\tbeta_{1j}$, $\tbeta_{2j}$, $\talpha_j$ and $\tbeta_{1r}$.
\end{proof}
%

\begin{proof}[\textbf{Proof of \cref{cor:equiv}}]
	I prove the first statement in \cref{cor:equiv} by contradiction. Assume for $j \neq \{u, r\}$ that the $j$th model for the transformed regressors $Z$ and untransformed regressors $X$ are equivalent so we can transform the estimators for the auxiliary variables into each other, i.e.\ $\tgamma_{2j} = \bXi^{1/2} \bDelta_{2}^{-1} \tbeta_{2j}$. The assumption requires the restrictions in the estimation for the transformed and untransformed regressors (see \eqref{eq:taylor}) to be equivalent, i.e.\
	\begin{equation}\label{eq:assequiv}
		R_{j}^{\top} \tgamma_{2j} = R_{j}^{\top} \tbeta_{2j}. \tag{$*$}
	\end{equation}

	Without loss of generality, I analyze the special case $k_2 = 2$, where $\tgamma_{2j} = (\tgamma_{2j,1}, \tgamma_{2j,2})^{\top}$ and $\tbeta_{2j} = (\tbeta_{2j,1}, \tbeta_{2j,2})^{\top}$, and the $j$th model is assumed to set  $\tgamma_{2j,2} = \tbeta_{2j,2} = 0$ via the restriction matrix $R_{j}^{\top} = (0, 1)$. Define the elements of $\bXi^{1/2}$ to be
	\begin{equation*}
		\bXi^{1/2} = \begin{pmatrix}
			\xi_{11} & \xi_{12} \\
			\xi_{21} & \xi_{22}
		\end{pmatrix},
	\end{equation*}
	and $\bDelta_{2}^{-1}$ reduces to a $2 \times 2$ diagonal matrix $\bDelta_{2}^{-1} = \diag(\bDelta_{2}^{11}, \bDelta_{2}^{22})$.

	Under the assumption $\tgamma_{2j} = \bXi^{1/2} \bDelta_{2}^{-1} \tbeta_{2j}$, the restriction in the estimation of the $j$th model for the transformed regressors is
	\begin{align*}
		R_{j}^{\top} \tgamma_{2j} &= R_{j}^{\top} \bXi^{1/2} \bDelta_{2}^{-1} \tbeta_{2j} \\
		&= \begin{pmatrix} 0 & 1 \end{pmatrix} 
		\begin{pmatrix}
			\xi_{11} & \xi_{12} \\
			\xi_{21} & \xi_{22}
		\end{pmatrix}
		\begin{pmatrix}
			\bDelta_{2}^{11} \tbeta_{2j,1} \\
			\bDelta_{2}^{22} \tbeta_{2j,2}
		\end{pmatrix} \\
		&= \xi_{21} \bDelta_{2}^{11} \tbeta_{2j,1} + \xi_{22} \bDelta_{2}^{22} \tbeta_{2j,2} = 0,
	\end{align*}
	while the restriction in the estimation of the $j$th model for the untransformed regressors is
	\begin{align*}
		R_{j}^{\top} \tbeta_{2j} = \tbeta_{2j,2} = 0.
	\end{align*}
	Therefore, generally $R_{j}^{\top} \tgamma_{2j} \neq R_{j}^{\top} \tbeta_{2j}$ so the restrictions in the estimator using the transformed and untransformed regressors differ, which contradicts $\eqref{eq:assequiv}$. Consequently, $\tgamma_{2j} = \bXi^{1/2} \bDelta_{2}^{-1} \tbeta_{2j}$ cannot generally hold for $k_2 > 2$ as it does not even hold for $k_2 = 2$. By extension, $\tgamma_{1j} \neq \bDelta_{1}^{-1} \tbeta_{1j}$ because the restrictions in the estimation differ, which finishes the proof of the first part of the corollary.

	For $k_2 = 1$, there exist only two models: 1.\ the unrestricted and 2.\ the fully restricted model so $j \in \{u, r\}$. In this case, the general results $\tgamma_{1u} = \bDelta^{-1} \tbeta_{1u}$, $\tgamma_{2u} = \bXi^{1/2} \bDelta_{2}^{-1} \tbeta_{2u}$ and $\tgamma_{1r} = \bDelta_{1}^{-1} \tbeta_{1r}$ from \eqref{eq:trafosol2} apply. Finally, the fully restricted model fulfills $R_{r}^{\top}\tgamma_{2r} = \tgamma_{2r} = 0$, which finishes the proof of the second part of the corollary.
\end{proof}

\setcounter{figure}{0} 
\setcounter{table}{0}

\subsection{Software and implementation}\label{sec:software}

The simulation experiment of \Cref{sec:sim} is performed on the scientific computing center sciCORE at the University of Basel (\url{https://scicore.unibas.ch/}) with \textsf{R} version 4.3.0 \citep{R2023}, while the empirical illustration of \Cref{sec:cvexpvar} is computed on a local machine running \textsf{R} version 4.3.1. Models estimated by WALS are fitted using the newly developed \textsf{R} package \pkg{WALS} version 0.2.4 \citep{wals} available from the Comprehensive \textsf{R} Archive Network (CRAN, \url{https://cran.r-project.org/package=WALS}). \pkg{WALS} is based on the \textsf{MATLAB} code version 2.0 for WALS in the linear regression model by \citet{magnus2016wals} which can be downloaded from \url{https://www.janmagnus.nl/items/WALS.pdf}. The dependencies of \pkg{WALS} along with the particular versions used are: \pkg{Formula} version 1.2-5 \citep{formula2010}, \pkg{MASS} version 7.3-60 \citep{mass2002} and \pkg{Rdpack} version 2.5 \citep{rdpack2023}. Standard NB2 regressions use \texttt{glm.nb()} from \pkg{MASS}. The function uses an algorithm that alternates between fitting the coefficients $\beta$ of the NB2 regression for fixed dispersion parameter $\rho$ using IRLS (iteratively reweighted least squares) and then maximizing the log-likelihood with respect to $\rho$ given $\beta$. Moreover, training and validation splits for $K$-fold CV in \Cref{sec:cvexpvar} are generated using the function \texttt{cv()} of \pkg{mboost} version 2.9-8 \citep{mboost2014} and the computations are parallelized over the number of training observations $t_{l}$ using \pkg{parallel}. The lasso estimation of `lasso-int' in \Cref{sec:cvexpvar} is performed using \texttt{cv.glmregNB()} from \pkg{mpath} version 0.4-2.23 \citep{mpath2023}. 

Finally, the integral in \eqref{eq:gammaest} is evaluated numerically for all priors except for the Laplace prior. Numerical integration is performed using the \texttt{integrate()} function of \pkg{stats}, which uses an adaptive quadrature method with the basic step being a Gauss-Kronrod quadrature (for more details see the code documentation).

\subsubsection*{Computational details}

The lasso estimator of the NB2 regression model maximizes the following penalized objective function from \citet[p.~2687~f.]{wang2016pencount}:
\begin{equation*}
	\max_{\beta, \rho} L(\beta, \rho) = \max_{\beta, \rho} \left\{ \ell(\beta, \rho) - n \cdot d \sum_{j = 1}^{p} | \beta_{j} | \right\},
\end{equation*}
where $\beta = (\beta_0, \beta_{1}, \dotsc, \beta_{p})^{\top}$, $\beta_0$ is the constant and $d \geq 0$ is the regularization parameter. Moreover, $\ell(\beta, \rho)$ is the log-likelihood from \eqref{eq:nb2loglik} and the penalty term is scaled by the sample size $n$ such that the penalty does not vanish when the number of observations becomes large. Notice that neither the constant $\beta_0$ nor the dispersion parameter $\rho$ are regularized, but only the `true' regression coefficients $\beta_{j}$, $j > 0$, are regularized.

The optimal regularization parameter $d$ for `lasso-int' in \Cref{sec:cvexpvar} is determined by maximizing the 10-fold CV log-likelihood, where candidates $\{d_{1}, d_{2}, \dotsc, d_{H}\}$ are generated by first running the method on the entire training set (ignoring the folds) and setting $d_{H}$ such that an intercept only model is estimated. Then, the minimum value is set as $d_{1} = a \cdot d_{H}$, $0 < a < 1$. The remaining values are determined by an evenly spaced grid on the log scale, i.e.\ between $\log(d_{1})$ and $\log(d_{H})$.

The implementation alternates between $P$ iterations of 1.\ the coordinate descent algorithm described in \citet[p.~2690~f.]{wang2016pencount} to estimate the regression coefficients $\beta$ given a value of the dispersion parameter $\rho$ and 2.\ maximizing the log-likelihood with respect to $\rho$ given the estimate of $\beta$. The latter uses \texttt{theta.ml()} from \pkg{MASS} \citep{mass2002} but limits its number of iterations to 10. The alternation process is stopped when certain convergence criteria are met. For the estimation of `lasso-int' in \Cref{sec:cvexpvar}, the maximum allowed number of iterations (until convergence) for the coordinate descent algorithm is increased to 2500 from the default setting of 1000. Furthermore, I also increase the maximum number of alternations between coordinate descent and ML estimation of the dispersion parameter $\rho$ from 10 to 1000. The remaining settings are left at their default values, see the documentation of \pkg{mpath} \citep{mpath2023} for more details.

A caveat for the results of `lasso-int' on the DoctorVisits dataset is that the ML estimation of $\rho$ often reaches its internal iteration limit within the alternation process between coordinate descent, for estimating the regression coefficients, and ML estimation of $\rho$ using \texttt{theta.ml()} from \pkg{MASS}. Ideally, I could increase the number of iterations used in \texttt{theta.ml()}, however, the implementation of lasso in \pkg{mpath} uses a fixed number of 10 iterations. I compensate for this by allowing the alternation process to run a maximum of 1000 iterations (instead of the default 10) until convergence before it is forced to terminate. The idea is that, even if \texttt{theta.ml()} does not converge in an iteration of the alternation process, the alternation process can go through many iterations, where \texttt{theta.ml()} is run in each round.

Tables are generated using \pkg{xtable} version 1.8-4 \citep{xtable2019} and \pkg{stargazer} version 5.2.3 \citep{stargazer2022}, and some plots use \pkg{ggplot2} version 3.4.4 \citep{ggplot2016} with themes from \pkg{ggthemes} version 4.2.4 \citep{ggthemes2021}. \LaTeX\ expressions are inserted with \pkg{latex2exp} version 0.9.6 \citep{latex2exp2022}. Finally, results are partly processed with \pkg{abind} version 1.4-5 \citep{abind2016} before plotting.

\setcounter{figure}{0} %
\setcounter{table}{0}

\subsection{Additional tables for the simulation experiment}\label{sec:appSim}

\begin{table}[h]
	\caption{Values of $\bar{\beta}_1$}\label{tab:barbeta1}
	\begin{center}
		\begin{threeparttable}
		\begin{footnotesize}
			\begin{tabular}{lrrrrrrrrrr}
  \toprule
 \midrule
 & 1 & 2 & 3 & 4 & 5 & 6 & 7 & 8 & 9 & 10 \\ 
  \midrule
  & -0.1518 & -0.197 & 0.1401 & 0.1328 & -0.155 & 0.1775 & 0.1403 & 0.1272 & 0.1778 & -0.1602 \\ 
   \bottomrule
\end{tabular}

			\end{footnotesize}
		\begin{tablenotes}
			\footnotesize
			\item[--] All figures rounded to four decimal places. 
		\end{tablenotes}
		\end{threeparttable}
	\end{center}
	\end{table}

\begin{landscape}
\begin{table}
	\caption{Values of $\bar{\beta}_2$}\label{tab:barbeta2}
	\begin{center}
		\begin{threeparttable}
		\begin{footnotesize}
			\begin{tabular}{lrrrrrrrrrr}
  \toprule
 \midrule
 & 1 & 2 & 3 & 4 & 5 & 6 & 7 & 8 & 9 & 10 \\ 
  \midrule
0 & $ -8 \times 10^{-4 }$ & 0.0089 & 0.0052 & 0.0087 & $ -6 \times 10^{-4 }$ & 0.0021 & $ -3 \times 10^{-4 }$ & -0.0078 & -0.005 & 0 \\ 
  10 & -0.0025 & 0.0087 & $ 5 \times 10^{-4 }$ & -0.0037 & -0.0044 & 0.0058 & 0.004 & -0.0067 & -0.0087 & 0.0051 \\ 
  20 & 0.0024 & -0.0066 & -0.0088 & -0.0078 & -0.0024 & -0.0066 & -0.004 & -0.0062 & -0.0049 & -0.0064 \\ 
  30 & $ -5 \times 10^{-4 }$ & 0.0054 & -0.0094 & $ 5 \times 10^{-4 }$ & 0.0076 & -0.0025 & -0.009 & -0.0072 & -0.0036 & -0.0069 \\ 
  40 & -0.0074 & -0.0056 & -0.0055 & -0.0074 & 0.0096 & -0.0035 & $ 1 \times 10^{-4 }$ & 0.0036 & -0.008 & -0.0076 \\ 
  50 & -0.009 & 0.0086 & 0.0035 & -0.0081 & $ -1 \times 10^{-4 }$ & $ -8 \times 10^{-4 }$ & -0.0025 & 0.0098 & -0.0065 & 0.0063 \\ 
  60 & -0.0086 & -0.002 & -0.0072 & -0.0061 & 0.0068 & 0.0044 & -0.0047 & $ -1 \times 10^{-4 }$ & -0.0083 & -0.0029 \\ 
  70 & 0.0094 & 0.0025 & 0.0033 & -0.0038 & -0.0019 & 0.0099 & 0.0071 & 0.0091 & 0.0062 & 0.0056 \\ 
  80 & -0.0046 & 0.0052 & 0.0097 & -0.0041 & -0.002 & 0.0062 & -0.0085 & -0.0027 & -0.0011 & -0.0069 \\ 
  90 & 0.0016 & 0.0094 & 0.0098 & -0.0065 & $ 8 \times 10^{-4 }$ & -0.0023 & 0.0035 & -0.0046 & $ -6 \times 10^{-4 }$ & -0.0066 \\ 
   \bottomrule
\end{tabular}

			\end{footnotesize}
		\begin{tablenotes}
			\footnotesize
			\item[--] To get element $\bar{\beta}_{2,z}$, construct $z$ as the sum of the row and the column.
			\item[--] Example: $\bar{\beta}_{2,23} = -0.0088$.
			\item[--] All figures rounded to four decimal places. 
		\end{tablenotes}
		\end{threeparttable}
	\end{center}
\end{table}
\end{landscape}

\clearpage

\setcounter{figure}{0} 
\setcounter{table}{0}

\subsection{Additional tables for the empirical illustration}\label{sec:appcvexpvar}

\begin{table}[h]
	\caption{Variable descriptions for DoctorVisits}\label{tab:descDV}
	\begin{center}
		\begin{threeparttable}
		\begin{small}
			
\begin{tabular}{@{\extracolsep{5pt}}lp{10cm}} 
\\[-1.8ex]\hline 
\hline \\[-1.8ex] 
Variable & \multicolumn{1}{l}{Description}  \\ 
\hline \\[-1.8ex] 
visits & \# of doctor visits in past two weeks.\\ 
genderfemale & = 1 if individual is female.\\ 
& Omitted reference category: male.\\
age & Age in years divided by 100.\\ 
income & Annual income in tens of thousands of dollars.\\ 
illness & \# of illnesses in past two weeks.\\ 
reduced & \# of days of reduced activity in past two weeks due to illness or injury.\\ 
health & General health questionnaire score using Goldberg's method.\\ 
privateyes & = 1 if individual has private health insurance  \\ 
& Omitted reference category: individual has no private health insurance.\\
freepooryes & = 1 if individual has free government health insurance due to low income.\\ 
freerepatyes & = 1 if individual has free government health insurance due to old age, disability or veteran status.\\ 
& Omitted reference category: individual has no free government health insurance.\\
nchronicyes & = 1 if individual has a chronic condition which does not limit activity.\\
lchronicyes & = 1 if individual has a chronic condition which limits activity.\\
& Omitted reference category: individual has no chronic condition.\\
\hline \\[-1.8ex] 
\end{tabular} 

			\end{small}
			\begin{tablenotes}
				\footnotesize
				\item[--] Reproduced based on the documentation of \texttt{DoctorVisits} in \pkg{AER} \citep{kleiber2008aer}.
			\end{tablenotes}
		\end{threeparttable}
	\end{center}
	\end{table}

\begin{table}[h]
	\caption{Summary statistics for DoctorVisits}\label{tab:sumstatsDV}
	\begin{center}
		\begin{threeparttable}
		\begin{footnotesize}
			
\begin{tabular}{@{\extracolsep{5pt}}lccccccc} 
\\[-1.8ex]\hline 
\hline \\[-1.8ex] 
Statistic & \multicolumn{1}{c}{Mean} & \multicolumn{1}{c}{St. Dev.} & \multicolumn{1}{c}{Min} & \multicolumn{1}{c}{Pctl(25)} & \multicolumn{1}{c}{Median} & \multicolumn{1}{c}{Pctl(75)} & \multicolumn{1}{c}{Max} \\ 
\hline \\[-1.8ex] 
visits & 0.302 & 0.798 & 0 & 0 & 0 & 0 & 9 \\ 
health & 1.218 & 2.124 & 0 & 0 & 0 & 2 & 12 \\ 
genderfemale & 0.521 & 0.500 & 0 & 0 & 1 & 1 & 1 \\ 
age & 0.406 & 0.205 & 0.190 & 0.220 & 0.320 & 0.620 & 0.720 \\ 
income & 0.583 & 0.369 & 0 & 0.250 & 0.550 & 0.900 & 1.500 \\ 
illness & 1.432 & 1.384 & 0 & 0 & 1 & 2 & 5 \\ 
reduced & 0.862 & 2.888 & 0 & 0 & 0 & 0 & 14 \\ 
privateyes & 0.443 & 0.497 & 0 & 0 & 0 & 1 & 1 \\ 
freepooryes & 0.043 & 0.202 & 0 & 0 & 0 & 0 & 1 \\ 
freerepatyes & 0.210 & 0.407 & 0 & 0 & 0 & 0 & 1 \\ 
nchronicyes & 0.403 & 0.491 & 0 & 0 & 0 & 1 & 1 \\ 
lchronicyes & 0.117 & 0.321 & 0 & 0 & 0 & 0 & 1 \\ 
$\text{age}^2$ & 0.207 & 0.186 & 0.036 & 0.048 & 0.102 & 0.384 & 0.518 \\ 
health$\times$genderfemale & 0.695 & 1.752 & 0 & 0 & 0 & 0 & 12 \\ 
health$\times$age & 0.503 & 1.013 & 0 & 0 & 0 & 0.570 & 8.640 \\ 
health$\times$income & 0.643 & 1.378 & 0 & 0 & 0 & 0.750 & 16.500 \\ 
genderfemale$\times$illness & 0.839 & 1.312 & 0 & 0 & 0 & 1 & 5 \\ 
\hline \\[-1.8ex] 
\end{tabular} 

			\end{footnotesize}
		\begin{tablenotes}
			\footnotesize
			\item[--] St.\ Dev.\ : Standard deviation, Pctl(25): 25\% quantile, Pctl(75): 75\% quantile.
			\item[--] $\times$ indicates interaction between two variables.
			\item[--] Number of observations $N = 5190$.
			\item[--] All figures rounded to three decimal places.
		\end{tablenotes}
		\end{threeparttable}
	\end{center}
	\end{table}

\clearpage


\section*{Supplementary materials}\label{sec:sup}
\addcontentsline{toc}{section}{Supplementary materials}

\renewcommand{\thesubsection}{S\arabic{subsection}}
\renewcommand\thefigure{S\arabic{subsection}.\arabic{figure}}
\renewcommand\thetable{S\arabic{subsection}.\arabic{table}}

\renewcommand{\theequation}{S\arabic{subsection}.\arabic{equation}}
\renewcommand{\thetheorem}{S\arabic{subsection}.\arabic{theorem}}
\renewcommand{\theass}{S\arabic{subsection}.\arabic{ass}}
\counterwithin{equation}{subsection} 
\counterwithin{theorem}{subsection} 
\counterwithin{ass}{subsection} 

\setcounter{subsection}{0}
\setcounter{figure}{0}
\setcounter{table}{0}
\pagestyle{supplement}

\subsection{Asymptotic distribution of one-step ML estimators}\label{sec:asymp}

Following \citet[p.~4]{deluca2018glm}, I use the local misspecification framework with true auxiliary parameters set to
\begin{equation}\label{eq:misspec}
	\beta_{2} = \frac{\delta}{\sqrt{n}},
\end{equation}
so the true parameter vector $\varphi = (\beta_{1}^{\top}, \beta_{2}^{\top}, \alpha)^{\top}$ converges to $\varphi_{0} = (\beta_{1}, 0, \alpha)$ for $n \rightarrow \infty$. Note that the asymptotics in this paper all refer to the case of $n \rightarrow \infty$.

I start the derivation by noting the asymptotic distribution of the fully iterated ML estimator of the unrestricted model $\check{\beta}_{1u}$, $\check{\beta}_{2u}$ and $\check{\alpha}_{u}$. Under the usual ML regularity conditions \citep[see e.g.][]{crowder1976ml}, it holds
\begin{align*}
	\sqrt{n} \begin{pmatrix}
		\check{\beta}_{1u} - \beta_{1} \\
		\check{\beta}_{2u} - \beta_{2} \\
		\check{\alpha}_{u} - \alpha
	\end{pmatrix} \stackrel{d}{\longrightarrow} 
	\normal (0, \Omega), \quad
	\Omega = \begin{pmatrix}
		\Omega_{1 1} & \Omega_{1 2} & \Omega_{1 \alpha}  \\
		\Omega_{2 1} & \Omega_{2 2} & \Omega_{2 \alpha} \\	
		\Omega_{\alpha 1} & \Omega_{\alpha 2} & \Omega_{\alpha \alpha}
\end{pmatrix}	 
:= 
\begin{pmatrix}
	\info_{1 1} & \info_{1 2} & \info_{1 \alpha} \\
	\info_{2 1} & \info_{2 2} & \info_{2 \alpha} \\
	\info_{\alpha 1} & \info_{\alpha 2} & \info_{\alpha \alpha} 
\end{pmatrix}^{-1}
= 
\info^{-1},
\end{align*}
where $\info := \info(\varphi_{0})$ is the information matrix evaluated at $\varphi_{0}$. For the remainder, I will restrict the analysis to the one-step ML estimators for the regression coefficients. The asymptotic distribution including the dispersion coefficient can be derived analogously.

Let $\bar{\varphi} = (\bar{\beta}_{1}^{\top}, \bar{\beta}_{2}^{\top}, \bar{\alpha})^{\top}$ collect the starting values and assume $\bar{\varphi} - \varphi = O_{p}(1/\sqrt{n})$, then the unrestricted one-step ML estimator has the same asymptotic distribution as the fully iterated ML estimator \citep[see e.g.\ Theorem 3.5 in][]{newey1994large}, i.e.\
\begin{equation}\label{eq:asympdistNB2}
	\sqrt{n} \begin{pmatrix}
		\tbeta_{1u} - \beta_{1} \\
		\tbeta_{2u} - \beta_{2} 
	\end{pmatrix} \overset{d}{\longrightarrow} 
	\normal(0, \Omega_{S}), \quad
	\Omega_{S}
	=
	\begin{pmatrix}
		\Omega_{11} & \Omega_{12} \\
		\Omega_{21} & \Omega_{22}
	\end{pmatrix}.
\end{equation}
Using $\info_{p \alpha} = 0 = \info_{\alpha p}^{\top}$, $p = 1,2$ \citep[p.~211]{lawless1987nb2}, the elements of $\Omega_{S}$ can be expressed as
\begin{equation}\label{eq:covmat}
	\begin{aligned}
		\Omega_{11} &= \info_{11}^{-1} + \info_{11}^{-1} \info_{12} (\info_{22} - \info_{21} \info_{11}^{-1} \info_{12})^{-1} \info_{21} \info_{11}^{-1}, \\
		\Omega_{12} &= - \info_{11}^{-1} \info_{12} (\info_{22} - \info_{21} \info_{11}^{-1} \info_{12})^{-1} = \Omega_{21}^{\top}, \\
		\Omega_{22} &= (\info_{22} - \info_{21} \info_{11}^{-1} \info_{12})^{-1}.
	\end{aligned}
	\end{equation}
If the DGP is included in the set of models considered for averaging, using the fully iterated ML estimator of the unrestricted model as starting values ensures, under the usual ML regularity conditions, that $\bar{\varphi} - \varphi = O_{p}(1/\sqrt{n})$ \citep[p.~4]{deluca2018glm}.

The following proposition utilizes the aforementioned ingredients and provides the asymptotic distribution of the one-step ML estimators for the $j$th model:

\newpage
\begin{proposition}[Asymptotic distribution of one-step ML estimators]\label{prop:onestepdist} In addition to \eqref{eq:misspec}, assume that the usual ML regularity conditions hold \citep[see e.g.][]{crowder1976ml}. If $\bar{\varphi} - \varphi = O_{p}(1/\sqrt{n})$, then as $n \rightarrow \infty$,
	\begin{equation*}
		\sqrt{n} \begin{pmatrix}
			\tbeta_{1j} - \beta_{1} \\
			\tbeta_{2j} - \beta_{2}
		\end{pmatrix} \overset{d}{\longrightarrow} \normal
		\left( \begin{pmatrix}
			\mQ \mP_j \Omega_{22}^{-1/2} \delta \\
			-\Omega_{22}^{1/2} \mP_j \Omega_{22}^{-1/2} \delta
		\end{pmatrix},
		\begin{pmatrix}
			\info_{11}^{-1} + \mQ \mW_j \mQ^{\top} & -\mQ \mW_j \Omega_{22}^{1/2} \\
			-\Omega_{22}^{1/2} \mW_j \mQ^{\top}	& \Omega_{22}^{1/2} \mW_j \Omega_{22}^{1/2}
		\end{pmatrix}
		\right),
	\end{equation*}
	where $\info_{pq}$ denotes the $pq$th submatrix of $\info(\varphi_0)$.	Further, $\Omega_{22} = (\info_{22} - \info_{21} \info_{11}^{-1} \info_{12})^{-1}$, $\mQ = \info_{11}^{-1} \info_{12} \Omega_{22}^{1/2}$, $\mP_j = \Omega_{22}^{1/2} R_j (R_{j}^{\top} \Omega_{22} R_j)^{-1} R_{j}^{\top} \Omega_{22}^{1/2}$ and $\mW_j = I_{k_2} - \mP_j$.
\end{proposition}

\begin{proof}[\textbf{Proof of \cref{prop:onestepdist}}]
	Let me first consider the asymptotic distribution of $\tvartheta$. Equation \eqref{eq:vartheta} implies
	\begin{equation*}\label{eq:vartheta2}
		\sqrt{n}(\tvartheta - \vartheta) = \left( \frac{\bX_{2}^{\top} \bM_{1} \bX_2}{n} \right)^{1/2} \sqrt{n} (\tbeta_{2u} - \beta_{2}) + \left[\left( \frac{\bX_{2}^{\top} \bM_{1} \bX_2}{n} \right)^{1/2} -\Omega_{22}^{-1/2} \right] \delta,
	\end{equation*}
	where $\vartheta :=\Omega_{22}^{-1/2} \beta_{2}$. Next, I analyze the probability limit of $\bX_{2}^{\top} \bM_{1} \bX_2 / n$. As $n \rightarrow \infty$
	\begin{align*}
		\plim \frac{\bX_{2}^{\top} \bM_1 \bX_2}{n}
		={}& \plim \Biggl( \frac{\bH_{22}}{n} - \frac{1}{n} \frac{\bH_{2\alpha}\bH_{\alpha 2}}{\bH_{\alpha\alpha}} \\
		&-\Biggl\{ 
			\left( \frac{\bH_{21}}{n} - \frac{1}{n} \frac{\bH_{2\alpha}\bH_{\alpha 1}}{\bH_{\alpha\alpha}} \right)
			\left(\frac{\bH_{11}}{n} - \frac{1}{n} \frac{\bH_{1\alpha}\bH_{\alpha 1}}{\bH_{\alpha\alpha}}  \right)^{-1} \\
			&\cdot \left( \frac{\bH_{12}}{n} - \frac{1}{n} \frac{\bH_{1\alpha}\bH_{\alpha 2}}{\bH_{\alpha\alpha}}  \right)
		\Biggr\} \Biggr) \\
		={}&
		\info_{22} - \frac{\info_{2 \alpha} \info_{\alpha 2}}{\info_{\alpha \alpha}} \\
		&- \left(\info_{21} - \frac{\info_{2 \alpha} \info_{\alpha 1}}{\info_{\alpha \alpha}}\right) \left( \info_{11} - \frac{\info_{1 \alpha} \info_{\alpha 1}}{\info_{\alpha \alpha}} \right)^{-1} \left( \info_{12} - \frac{\info_{1 \alpha} \info_{\alpha 2}}{\info_{\alpha \alpha}} \right). 
	\end{align*}
	\citet[p.~211]{lawless1987nb2} derived in equation (2.7b), that $\info_{p \alpha} = 0$, $p = 1,2$. Combined with \eqref{eq:covmat} this yields
	\begin{align}\label{eq:plimxmx}
		\plim \frac{\bX_{2}^{\top} \bM_1 \bX_2}{n} = \Omega_{22}^{-1}.
	\end{align}
	Thus, using \eqref{eq:asympdistNB2} I get
	\begin{equation}\label{eq:asympvartheta}
		\sqrt{n} (\tvartheta - \vartheta) \overset{d}{\longrightarrow} \normal(0, I_{k_2}).
	\end{equation}

	Now, I show the asymptotic distribution for all restricted estimators. First, I need the asymptotic distribution of the fully restricted estimator. From \eqref{eq:b1rb1u} I deduce
	\begin{equation*}\label{eq:b1r}
		\sqrt{n} (\tbeta_{1r} - \beta_{1}) = \sqrt{n} ( \tbeta_{1u} - \beta_{1} ) + \bQ \sqrt{n} ( \tvartheta - \vartheta ) + \bQ \Omega_{22}^{-1/2} \delta.
	\end{equation*}
	Since
	\begin{equation*}\label{eq:cQ}
		\plim \bQ = \info_{11}^{-1} \info_{12} \Omega_{22}^{1/2} =: \mQ,
	\end{equation*}
	which is the same as in Proposition 2 of \citet[p.~4]{deluca2018glm} for WALS GLM because $\alpha$ is asympotically independent of $\beta_1$ and $\beta_2$ in the NB2 model, it follows that
	\begin{equation}\label{eq:asympb1rvn}
		\sqrt{n} (\tbeta_{1r} - \beta_{1})
		\overset{d}{\longrightarrow}
		\normal(
			\info_{11}^{-1} \info_{12} \delta,
			\info_{11}^{-1}
		).
	\end{equation}
	This is equivalent to the asymptotic distribution of the fully restricted one-step ML estimator of the WALS GLM model in equation (A.4) of \citet[p.~14]{deluca2018glm}. Furthermore, $\tbeta_{1r}$ and $\tvartheta$ are also asymptotically independent as their joint asymptotic distribution is a multivariate normal with covariance $\Omega_{12} \Omega_{22}^{-1/2} + \info_{11}^{-1} \info_{12} \Omega_{22}^{1/2} \overset{\eqref{eq:covmat}}{=} 0$.

	Finally, I have all ingredients to derive the asymptotic distribution of the general one-step ML estimator. \cref{prop:onestepML} implies
	\begin{equation*}
		\sqrt{n} (\tbeta_{1j} - \beta_{1}) = \bQ \bP_j \Omega_{22}^{-1/2} \delta + ( \sqrt{n} (\tbeta_{1r} - \beta_{1}) - \bQ \Omega_{22}^{-1/2} \delta  ) - \bQ \bW_j \sqrt{n} (\tvartheta - \vartheta),
	\end{equation*}
	and
	\begin{equation*}
		\sqrt{n} (\tbeta_{2j} - \beta_{2}) = \left[ \left( \frac{\bX_{2}^{\top} \bM_1 \bX_2}{n} \right)^{-1/2} \bW_j \Omega_{22}^{-1/2} - I_{k_2} \right] \delta + \left( \frac{\bX_{2}^{\top} \bM_1 \bX_2}{n} \right)^{-1/2} \bW_j \sqrt{n} (\tvartheta - \vartheta).
	\end{equation*}
	Together with \eqref{eq:asympvartheta} and \eqref{eq:asympb1rvn}, the fact that $\tbeta_{1r}$ and $\tvartheta$ are asymptotically independent, and the probability limits
	\begin{align*}
		\plim \bP_j &= \Omega_{22}^{1/2} R_j (R_{j}^{\top} \Omega_{22} R_j)^{-1} R_{j}^{\top} \Omega_{22}^{1/2} =: \mP_j, \\
		\plim \bW_j &= I_{k_2} - \mP_j =: \mW_j,
	\end{align*}
	they imply the joint asymptotic distribution of $\tbeta_{1j}$ and $\tbeta_{2j}$ from the proposition.
\end{proof}

\subsection{Asymptotic distribution of one-step ML estimators in transformed models}\label{sec:asympNB2Z}

In order to establish the asymptotic distribution of the one-step ML estimators in the transformed models, it is necessary to determine the probability limits of the matrices involved in the transformation from $X$ to $Z$. Firstly, $\bDelta_1$ converges in probability to a constant matrix $\Delta_1$, i.e.\
\begin{equation*}\label{eq:plimdelta1}
	\plim \bDelta_{1} =: \Delta_1,
\end{equation*}
because $\plim \bX_{1}^{\top} \bX_{1} /n = \info_{11}$ is constant and $\bDelta_{1} = \diag\left(\bX_{1}^{\top} \bX_{1} / n \right)^{-1/2}$.
By the same line of reasoning,
\begin{equation*}\label{eq:plimdelta2}
	\plim \bDelta_{2} =: \Delta_{2},
\end{equation*}
which is constant because $\bDelta_{2} = \diag\left(\bX_{2}^{\top} \bM_{1} \bX_2 / n \right)^{-1/2}$ and $\plim = \bX_{2}^{\top} \bM_{1} \bX_2  /n  \overset{\eqref{eq:plimxmx}}{=} \Omega_{22}^{-1}$ is constant. Thus, the probability limit of $\bXi$ follows as
\begin{equation*}\label{eq:plimxi}
	\plim \bXi = \plim \frac{\bDelta_2 \bX_{2}^{\top} \bM_{1} \bX_2 \bDelta_2}{n} \overset{\eqref{eq:plimxmx}}{=} \Delta_2 \Omega_{22}^{-1} \Delta_2 =: \Xi.
\end{equation*}

Similar to \citet[p.~5]{deluca2018glm}, let $\gamma_{10} = \Delta^{-1} \beta_{1}$ and $\gamma_{2n} = \Xi^{1/2} \Delta_{2}^{-1} \beta_{2}$, then \cref{prop:onestepdist} implies
\begin{equation}
	\sqrt{n}\begin{pmatrix}
		\tgamma_{1j} - \gamma_{10}\\
		\tgamma_{2j} - \gamma_{2n}
	\end{pmatrix} \overset{d}{\longrightarrow} \normal
	\left(
		\begin{pmatrix}
		\mD P_{j} d \\
		-P_{j} d
	\end{pmatrix},
	\begin{pmatrix}
		\infoz_{11}^{-1} + \mD W_{j} \mD^{\top} & -\mD W_{j} \\
		-W_{j} \mD^{\top}							& W_{j} 
	\end{pmatrix}
	\right),
\end{equation}
where $d = \Xi^{1/2} \Delta_{2}^{-1} \delta$, $\mD = \plim \bar{D} = \infoz_{11}^{-1} \infoz_{12}$, $\infoz_{11} = \plim \bZ_{1}^{\top}\bZ_{1} / n = \Delta_{1} \info_{11} \Delta_{1}$ and $\infoz_{12} = \plim \bZ_{1}^{\top} \bZ_{2} / n = \Delta_{1} \info_{12} \Delta_{2} \Xi^{-1/2}$.
Therefore, I can approximate the distribution of $\sqrt{n}\tgamma_{2u}$ in large samples with
\begin{equation*}
	\sqrt{n} \tgamma_{2u} \approx \normal( \sqrt{n} \gamma_{2n}, I_{k_2}) = \normal(d, I_{k_2}).
\end{equation*}
%

\subsection{Additional results for the simulation experiment}
\setcounter{figure}{0} 
\setcounter{table}{0}

\subsubsection{Results for alternative scoring rules}

\begin{figure}[h]
	\centering
	\includegraphics[width=0.9\textwidth]{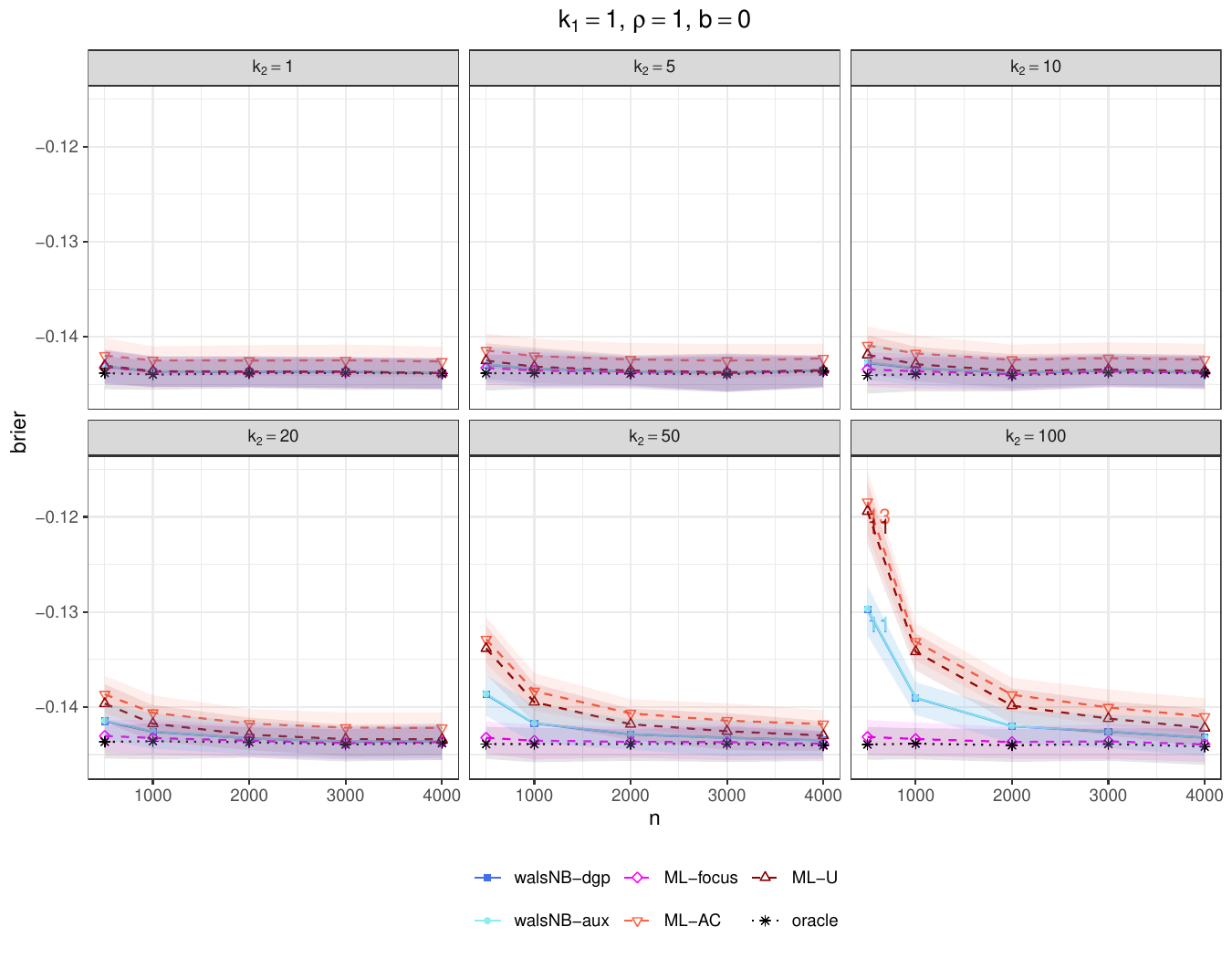}
	\caption{Mean validation Brier score and quartiles varying $n$ and $k_2$}\label{fig:simBrierk2}
	\justifying \footnotesize \noindent The remaining parameters are fixed at $k_1 = 1, \rho = 1$ and $b = 0$. Each point represents the mean over all successful runs of the experiment, i.e.\ over $R = 300$ when it never fails to converge. The shaded areas show the interquartile range. The number below a point indicates how often the method failed to converge in this particular setting.
\end{figure}

\begin{figure}[h]
	\centering
	\includegraphics[width=0.9\textwidth]{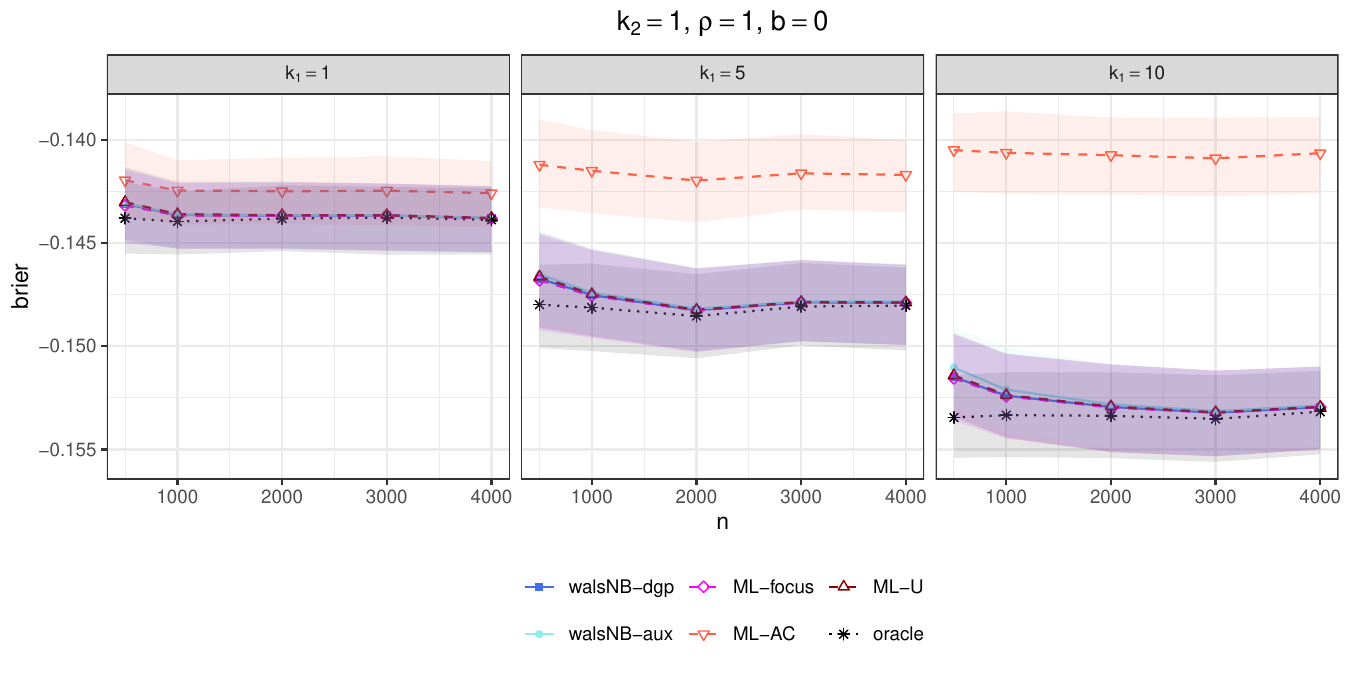}
	\caption{Mean validation Brier score and quartiles varying $n$ and $k_1$}\label{fig:simBrierk1}
	\justifying \footnotesize \noindent The remaining parameters are fixed at $k_2 = 1, \rho = 1$ and $b = 0$. Each point represents the mean over all successful runs of the experiment, i.e.\ over $R = 300$ when it never fails to converge. The shaded areas show the interquartile range.
\end{figure}

\begin{figure}[h]
	\centering
	\includegraphics[width=0.9\textwidth]{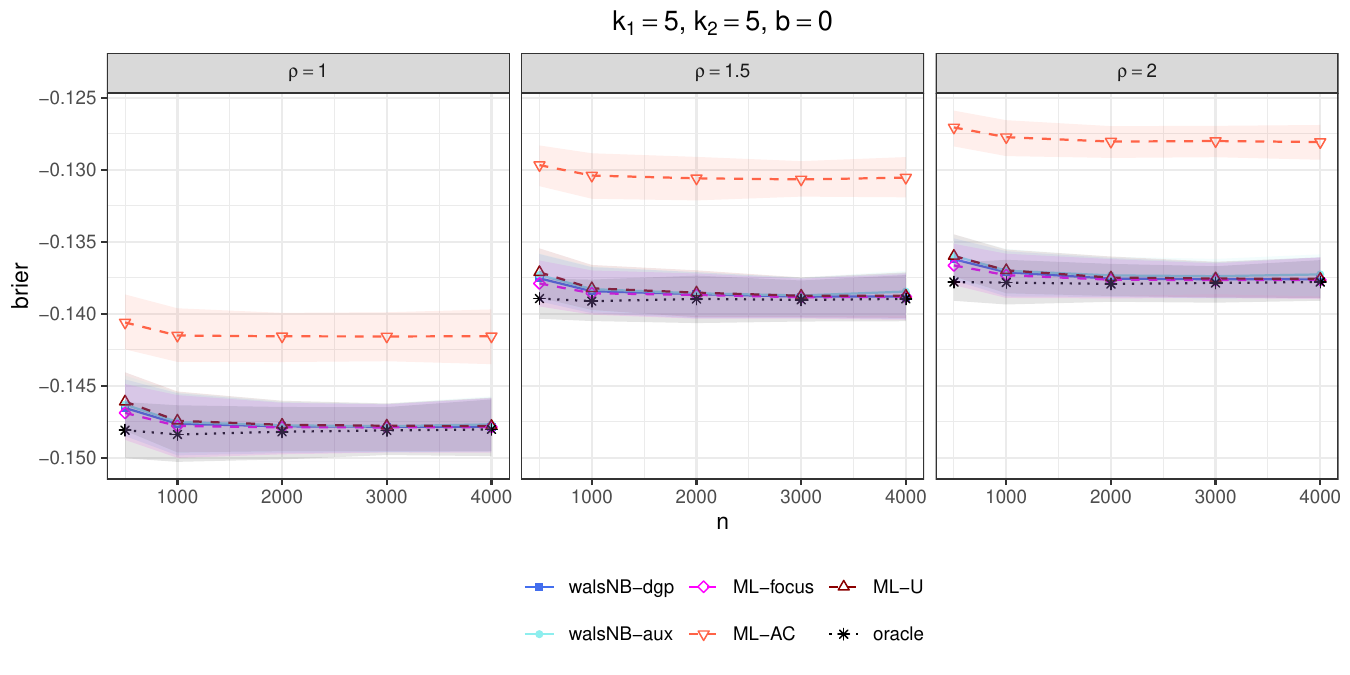}
	\caption{Mean validation Brier score and quartiles varying $n$ and $\rho$}\label{fig:simBrierTheta}
	\justifying \footnotesize \noindent The remaining parameters are fixed at $b = 0$ and $k_1 = k_2 = 5$. Each point represents the mean over all successful runs of the experiment, i.e.\ over $R = 300$ when it never fails to converge. The shaded areas show the interquartile range.
\end{figure}

\begin{figure}[h]
	\centering
	\includegraphics[width=0.9\textwidth]{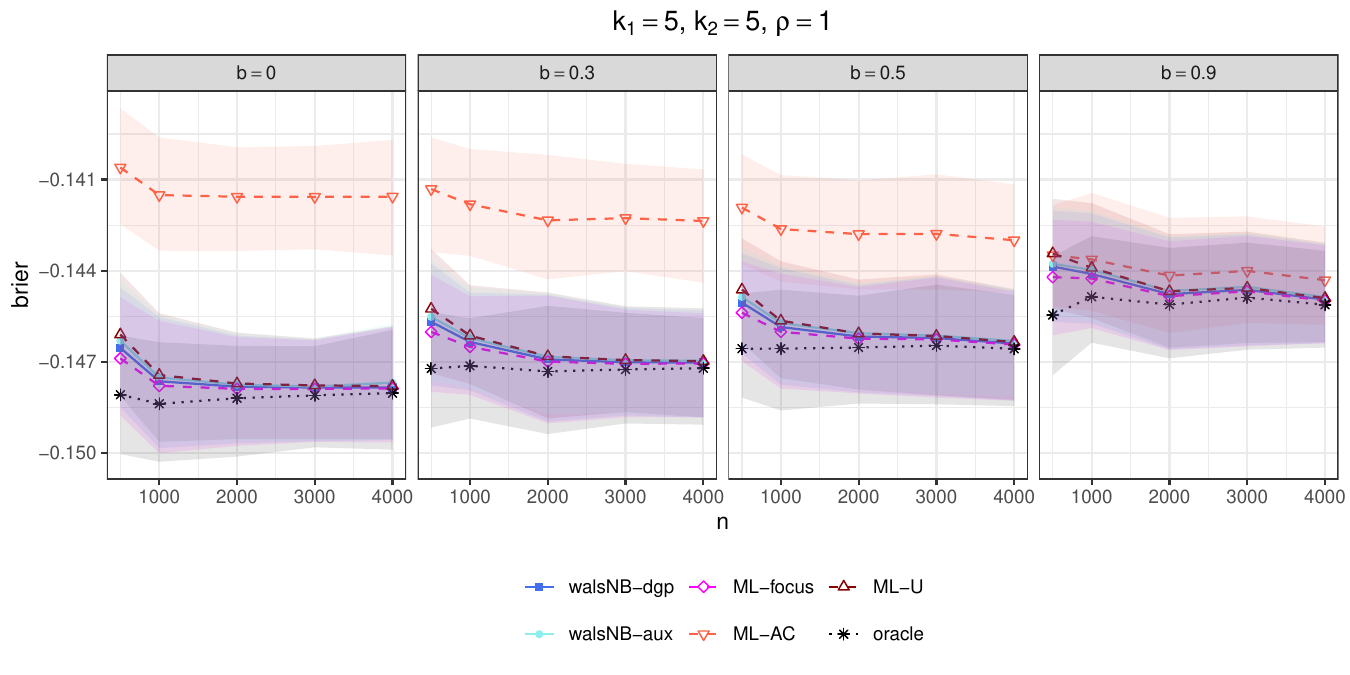}
	\caption{Mean validation Brier score and quartiles varying $n$ and $b$}\label{fig:simBrierCorr}
	\justifying \footnotesize \noindent The remaining parameters are fixed at $\rho = 1$ and $k_1 = k_2 = 5$. Each point represents the mean over all successful runs of the experiment, i.e.\ over $R = 300$ when it never fails to converge. The shaded areas show the interquartile range.
\end{figure}

\begin{figure}[h]
	\centering
	\includegraphics[width=0.9\textwidth]{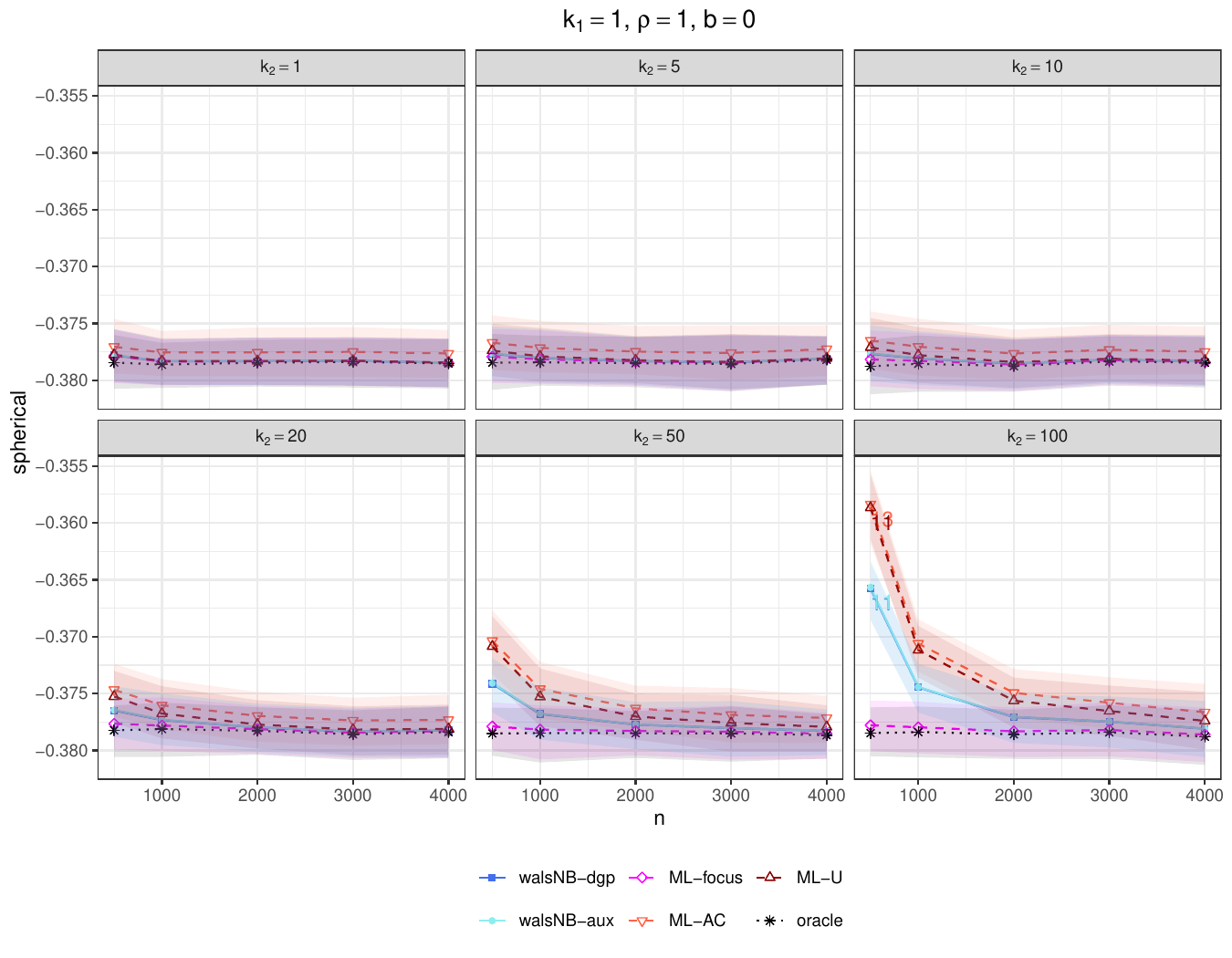}
	\caption{Mean validation spherical score and quartiles varying $n$ and $k_2$}\label{fig:simSphk2}
	\justifying \footnotesize \noindent The remaining parameters are fixed at $k_1 = 1, \rho = 1$ and $b = 0$. Each point represents the mean over all successful runs of the experiment, i.e.\ over $R = 300$ when it never fails to converge. The shaded areas show the interquartile range. The number below a point indicates how often the method failed to converge in this particular setting.
\end{figure}

\begin{figure}[h]
	\centering
	\includegraphics[width=0.9\textwidth]{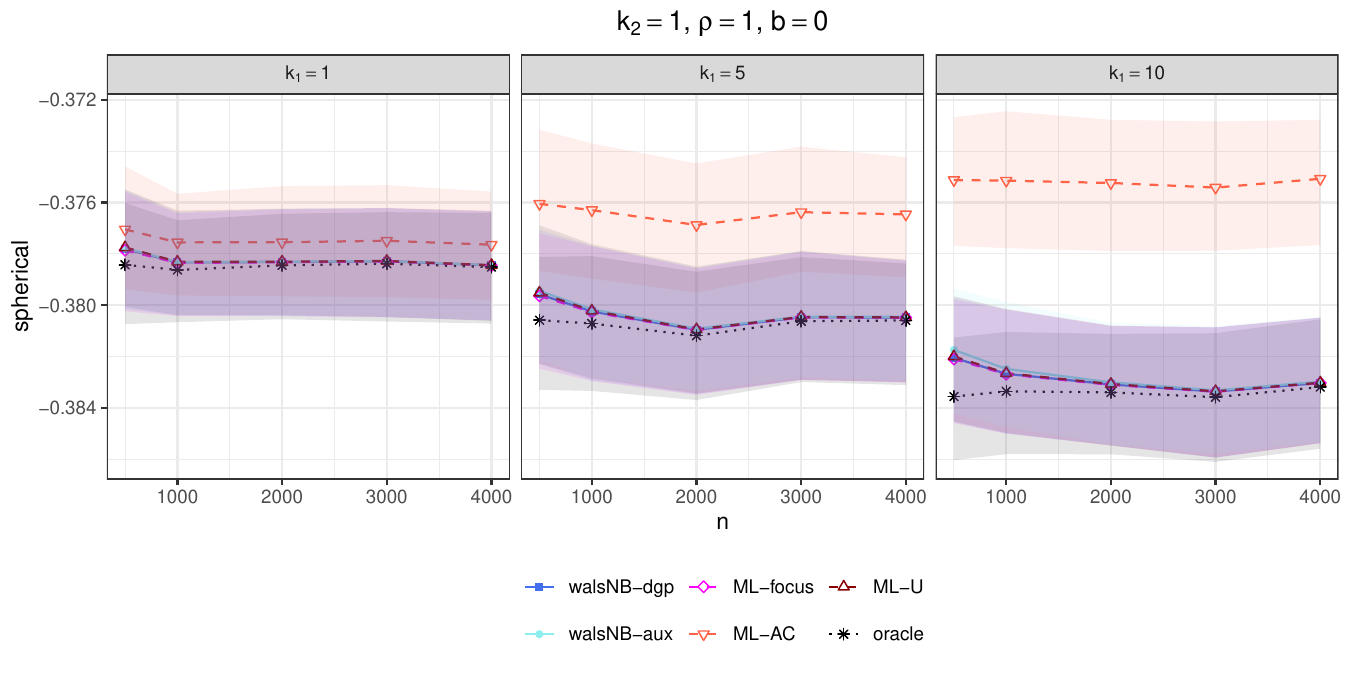}
	\caption{Mean validation spherical score and quartiles varying $n$ and $k_1$}\label{fig:simSphk1}
	\justifying \footnotesize \noindent The remaining parameters are fixed at $k_2 = 1, \rho = 1$ and $b = 0$. Each point represents the mean over all successful runs of the experiment, i.e.\ over $R = 300$ when it never fails to converge. The shaded areas show the interquartile range.
\end{figure}

\begin{figure}[h]
	\centering
	\includegraphics[width=0.9\textwidth]{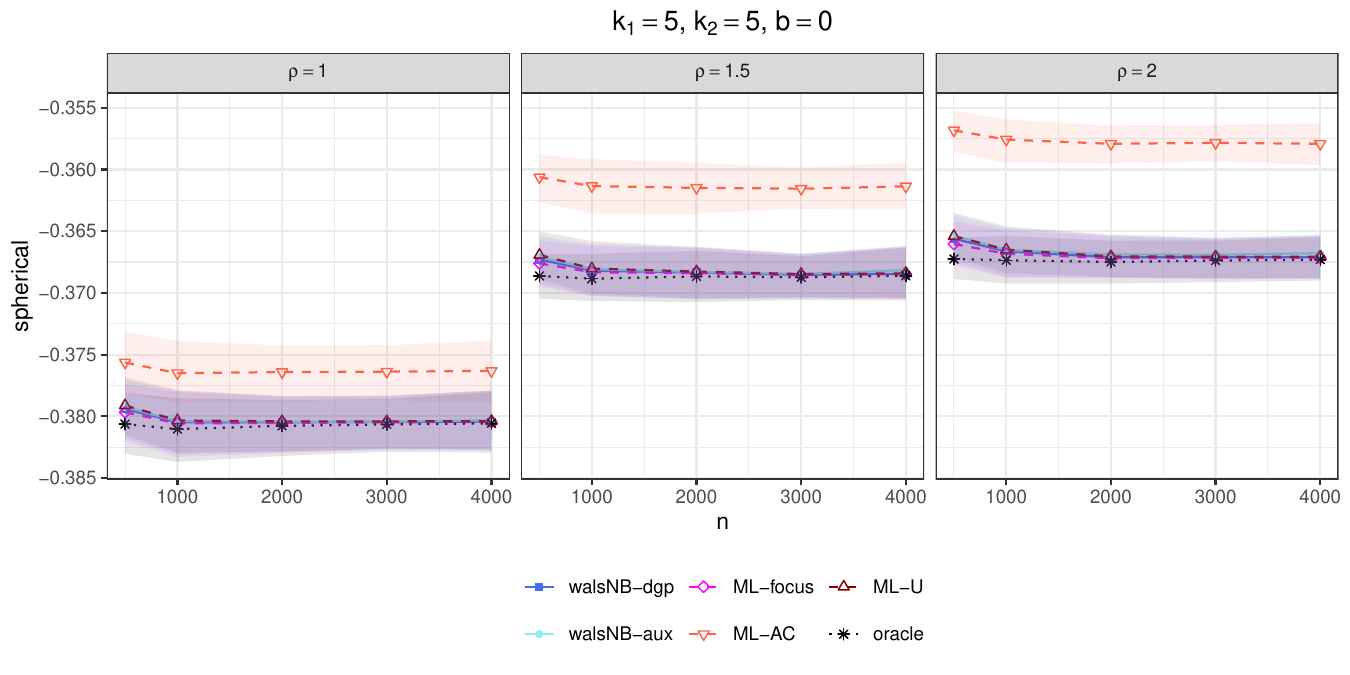}
	\caption{Mean validation spherical score and quartiles varying $n$ and $\rho$}\label{fig:simSphTheta}
	\justifying \footnotesize \noindent The remaining parameters are fixed at $b = 0$ and $k_1 = k_2 = 5$. Each point represents the mean over all successful runs of the experiment, i.e.\ over $R = 300$ when it never fails to converge. The shaded areas show the interquartile range.
\end{figure}

\begin{figure}[h]
	\centering
	\includegraphics[width=0.9\textwidth]{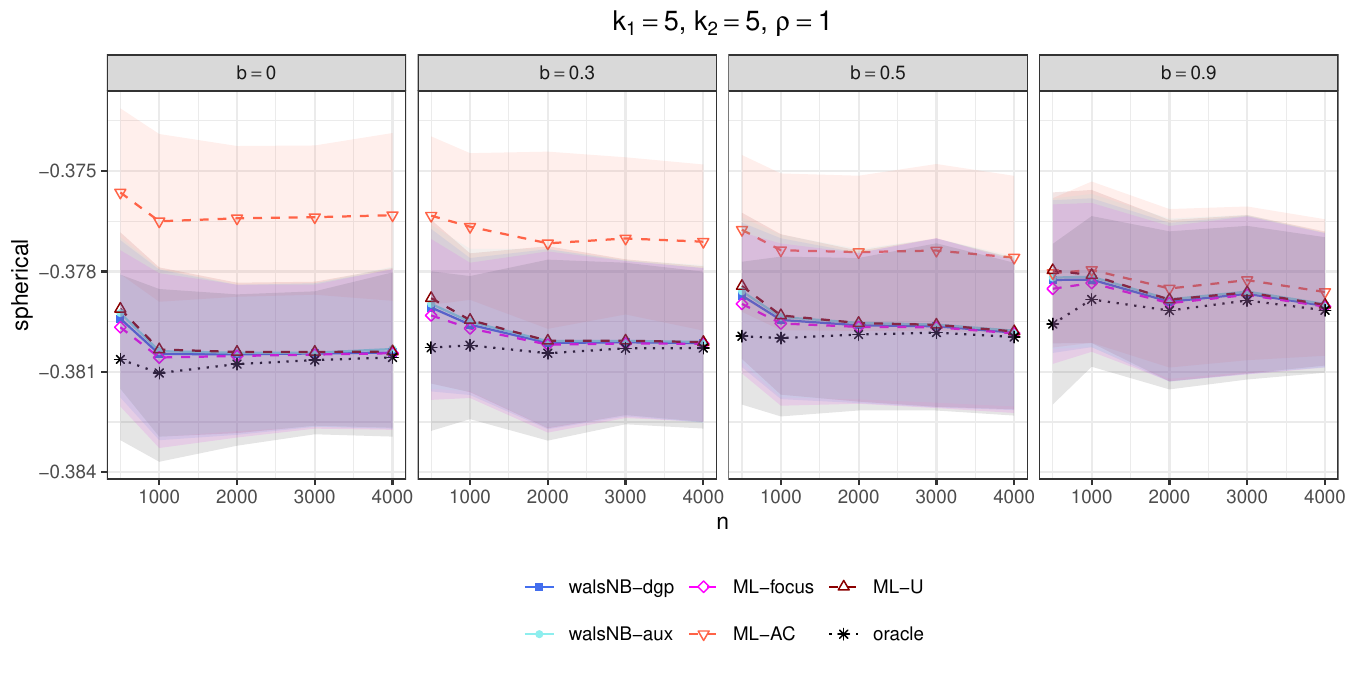}
	\caption{Mean validation spherical score and quartiles varying $n$ and $b$}\label{fig:simSphCorr}
	\justifying \footnotesize \noindent The remaining parameters are fixed at $\rho = 1$ and $k_1 = k_2 = 5$. Each point represents the mean over all successful runs of the experiment, i.e.\ over $R = 300$ when it never fails to converge. The shaded areas show the interquartile range.
\end{figure}

\clearpage

\subsection{Additional results for the empirical illustration}
\setcounter{figure}{0} 
\setcounter{table}{0}

\subsubsection{Results for alternative scoring rules}

\begin{figure}[h]
	\centering
	\includegraphics[width=0.85\textwidth]{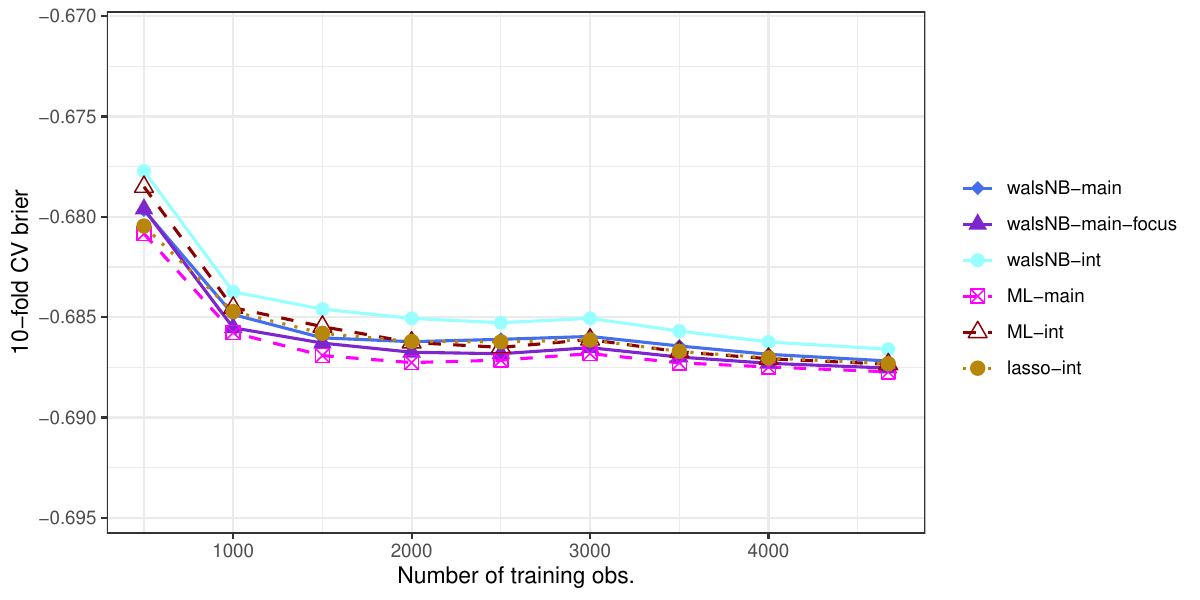}
	\caption{10-fold CV Brier score varying $t_l$, DoctorVisits}\label{fig:DVlearnbrier}
\end{figure}

\begin{figure}[h]
	\centering
	\includegraphics[width=0.85\textwidth]{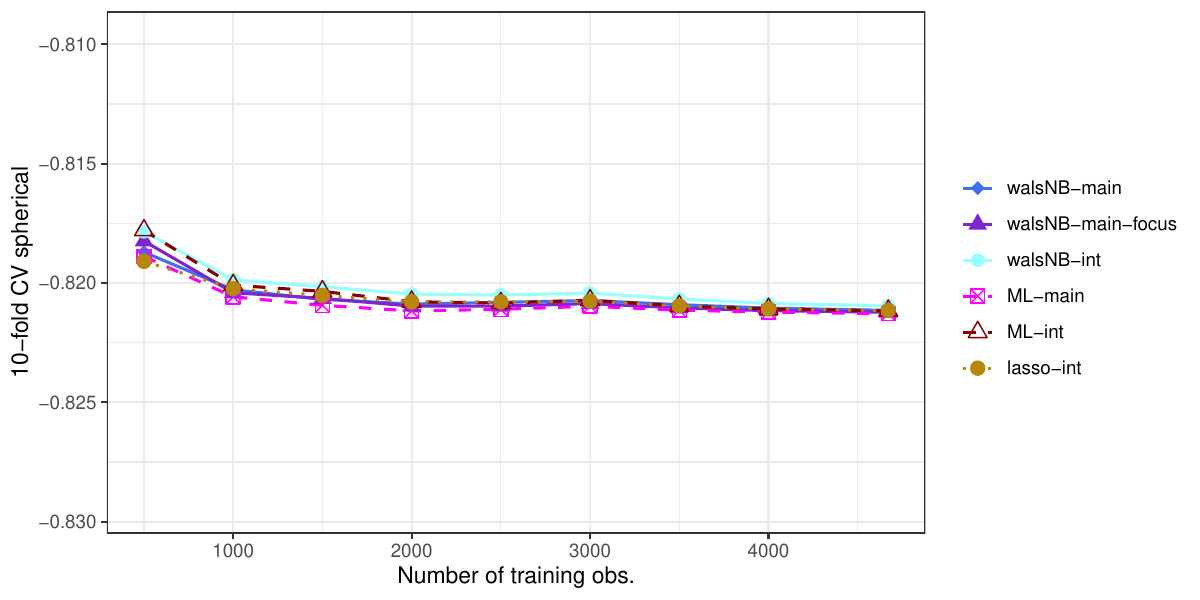}
	\caption{10-fold CV spherical score varying $t_l$, DoctorVisits}\label{fig:DVlearnsp}
\end{figure}

\clearpage


\begin{landscape}

\begin{table}
	\caption{10-fold CV Brier score varying $t_l$, DoctorVisits}\label{tab:DVlearnbrier}
	\begin{center}
		\begin{threeparttable}
		\begin{footnotesize}
			
\begin{tabular}{@{\extracolsep{5pt}} lccccccccc} 
\\[-1.8ex]\hline 
\hline \\[-1.8ex] 
 Training obs. $t_{l}$ & 500 & 1000 & 1500 & 2000 & 2500 & 3000 & 3500 & 4000 & 4671 \\ 
\hline \\[-1.8ex] 
walsNB-main & $$-$0.680$ & $$-$0.685$ & $$-$0.686$ & $$-$0.686$ & $$-$0.686$ & $$-$0.686$ & $$-$0.686$ & $$-$0.687$ & $$-$0.687$ \\ 
walsNB-main-focus & $$-$0.680$ & $$-$0.686$ & $$-$0.686$ & $$-$0.687$ & $$-$0.687$ & $$-$0.687$ & $$-$0.687$ & $$-$0.687$ & $$-$0.688$ \\ 
walsNB-int & $$-$0.678$ & $$-$0.684$ & $$-$0.685$ & $$-$0.685$ & $$-$0.685$ & $$-$0.685$ & $$-$0.686$ & $$-$0.686$ & $$-$0.687$ \\ 
ML-main & $$-$0.681$ & $$-$0.686$ & $$-$0.687$ & $$-$0.687$ & $$-$0.687$ & $$-$0.687$ & $$-$0.687$ & $$-$0.687$ & $$-$0.688$ \\ 
ML-int & $$-$0.679$ & $$-$0.685$ & $$-$0.685$ & $$-$0.686$ & $$-$0.687$ & $$-$0.686$ & $$-$0.687$ & $$-$0.687$ & $$-$0.687$ \\ 
lasso-int & $$-$0.680$ & $$-$0.685$ & $$-$0.686$ & $$-$0.686$ & $$-$0.686$ & $$-$0.686$ & $$-$0.687$ & $$-$0.687$ & $$-$0.687$ \\ 
\hline \\[-1.8ex] 
\end{tabular} 

			\end{footnotesize}
		\begin{tablenotes}
			\footnotesize
			\item[--] All figures rounded to three decimal places.
		\end{tablenotes}
		\end{threeparttable}
	\end{center}
	\end{table}

\begin{table}
	\caption{10-fold CV spherical score varying $t_l$, DoctorVisits}\label{tab:DVlearnsp}
	\begin{center}
		\begin{threeparttable}
		\begin{footnotesize}
			
\begin{tabular}{@{\extracolsep{5pt}} lccccccccc} 
\\[-1.8ex]\hline 
\hline \\[-1.8ex] 
 Training obs. $t_{l}$ & 500 & 1000 & 1500 & 2000 & 2500 & 3000 & 3500 & 4000 & 4671 \\ 
\hline \\[-1.8ex] 
walsNB-main & $$-$0.819$ & $$-$0.820$ & $$-$0.821$ & $$-$0.821$ & $$-$0.821$ & $$-$0.821$ & $$-$0.821$ & $$-$0.821$ & $$-$0.821$ \\ 
walsNB-main-focus & $$-$0.818$ & $$-$0.820$ & $$-$0.821$ & $$-$0.821$ & $$-$0.821$ & $$-$0.821$ & $$-$0.821$ & $$-$0.821$ & $$-$0.821$ \\ 
walsNB-int & $$-$0.818$ & $$-$0.820$ & $$-$0.820$ & $$-$0.820$ & $$-$0.821$ & $$-$0.820$ & $$-$0.821$ & $$-$0.821$ & $$-$0.821$ \\ 
ML-main & $$-$0.819$ & $$-$0.821$ & $$-$0.821$ & $$-$0.821$ & $$-$0.821$ & $$-$0.821$ & $$-$0.821$ & $$-$0.821$ & $$-$0.821$ \\ 
ML-int & $$-$0.818$ & $$-$0.820$ & $$-$0.820$ & $$-$0.821$ & $$-$0.821$ & $$-$0.821$ & $$-$0.821$ & $$-$0.821$ & $$-$0.821$ \\ 
lasso-int & $$-$0.819$ & $$-$0.820$ & $$-$0.821$ & $$-$0.821$ & $$-$0.821$ & $$-$0.821$ & $$-$0.821$ & $$-$0.821$ & $$-$0.821$ \\ 
\hline \\[-1.8ex] 
\end{tabular} 

			\end{footnotesize}
		\begin{tablenotes}
			\footnotesize
			\item[--] All figures rounded to three decimal places.
		\end{tablenotes}
		\end{threeparttable}
	\end{center}
	\end{table}
\end{landscape}

\clearpage

\pagestyle{literature}
\phantomsection
\addcontentsline{toc}{section}{References}

\bibliography{bib}
\bibliographystyle{elsarticle-harv}

\end{document}